\documentclass[10pt,a4paper]{article}
\usepackage[english]{babel}
\usepackage[T1]{fontenc}
\usepackage[dvips]{graphicx}
\usepackage{bbm}
\usepackage{amsmath, amsfonts, amsthm}
\usepackage{verbatim}

\numberwithin{equation}{section}
\usepackage{stmaryrd}
\usepackage{enumerate}

\usepackage[latin1]{inputenc}
\usepackage{graphicx}
\usepackage[numbers]{natbib}
\usepackage[all]{xy}
\usepackage{geometry}

\DeclareMathAlphabet{\mathpzc}{OT1}{pzc}{m}{it}
\def\uno{\mathbbm{1}}
\def\F{\mathcal{F}}
\def\A{\mathcal{A}}
\def\X{\mathcal{X}}
\def\Y{\mathcal{Y}}
\def\E{\mathcal{E}}
\def\Fz{\F^{(0)}}
\def\C_w^N{\ell^2_w}
\def\B{{\cal B}}
\def\be{{\bf e}}
\def\bi{{\bf i}}
\def\bii{({\bf i}, {\boldsymbol \iota})}
\def\bb{{\bf b}}

\def\Tc{\mathcal{T}}
\def\mod#1{\underline{#1}}
\def\Nc{\mathcal{N}}
\def\di{{ d}}
\def\Ph{\mathcal{P}}
\def\norma#1{\left\|#1\right\|}

\def\L{\mathcal{L}}
\def\Cs{\mathcal{C}}
\def\Csr{\mathcal{C}_{\mathbb{R}}}
\def\spazio#1{\Ph^{#1}}
\def\spazior#1{\Ph^{#1}_{{\mathbb R}}}
\def\Br{B_{{\mathbb{R}}}}
\def\tdue#1{Q^{#1}}
\def\im{{\rm i}}

\newcommand{\derz}[1]{\left.\frac{d}{d#1}\right|_{#1 =0}}

\newcommand{\norm}[1]{\left\|#1\right\|}
\newcommand*{\diff}{\mathop{}\!\mathrm{d}}

\newcommand{\multindex}[3]{ #1=(#1_1,\ldots,#1_{#2})\in \mathbb{N}^{#2}, \, |#1|=#3}
\newcommand{\R}{\mathbb{R}}
\newcommand{\N}{\mathbb{N}}
\newcommand{\Z}{\mathbb{Z}}
\newcommand{\T}{\mathbb{T}}
\renewcommand{\C}{\mathbb{C}}
\newcommand{\<}{\left\langle}
\renewcommand{\>}{\right\rangle}
\newcommand{\und}[1]{\underline{#1}}      
\newcommand{\mmod}[1]{\left| #1 \right|}

\newcommand{\llambda}{\hat{\lambda}}
\newcommand{\pb}[2]{\lbrace #1, \, #2 \rbrace}

\renewcommand{\P}{\mathcal{P}}
\renewcommand{\B}{\mathcal{B}}
\newcommand{\om}[1]{\omega\left(\tfrac{#1}{N} \right)}
\newcommand{\hp}{\hat{p}}
\newcommand{\hq}{\hat{q}}
\newcommand{\reg}{{s, \sigma}}
\newcommand{\eN}{\frac{\epsilon}{N^2}}

\newtheorem{teo}{Theorem}[section]
\newtheorem{theorem}{Theorem}[section]
\newtheorem{lem}[teo]{Lemma}
\newtheorem{lemma}[teo]{Lemma}
\newtheorem{cor}[teo]{Corollary}
\newtheorem{corollary}[teo]{Corollary}
\newtheorem{prop}[teo]{Proposition}
\newtheorem{df}[teo]{Definition}
\newtheorem{oss}[teo]{Remark}
\newtheorem{rem}[teo]{Remark}
\newtheorem{remark}[teo]{Remark}

\allowdisplaybreaks

\author{D. Bambusi\footnote{Dipartimento di Matematica, Universit\`a degli Studi di Milano, Via Saldini 50, I-20133
Milano. \newline
 \textit{Email: } \texttt{dario.bambusi@unimi.it}},
 A. Maspero\footnote{Dipartimento di Matematica, Universit\`a degli Studi di Milano, Via Saldini 50, I-20133
Milano. \newline
 \textit{Email: } \texttt{alberto.maspero@unimi.it}}} 

\title{Birkhoff coordinates for the Toda
  Lattice in the limit of infinitely many particles with an application to FPU}
\begin{document}
\maketitle
\begin{abstract}
In this paper we study the Birkhoff coordinates (Cartesian action
angle coordinates) of the Toda lattice with periodic boundary
condition in the limit where the number $N$ of the particles tends to
infinity. We prove that the transformation introducing such
coordinates maps analytically a complex ball of radius $R/N^\alpha$
(in discrete Sobolev-analytic norms) into a ball of radius
$R'/N^\alpha$ (with $R,R'>0$ independent of $N$) if and only if
$\alpha\geq2$. Then we consider the problem of equipartition of energy
in the spirit of Fermi-Pasta-Ulam. We deduce that corresponding to
initial data of size $R/N^2$, $0<R\ll 1$, and with only the first
Fourier mode excited, the energy remains forever in a packet of
Fourier modes exponentially decreasing with the wave number. Finally
we consider the original FPU model and prove that energy remains
localized in a similar packet of Fourier modes for times one order of
magnitude longer than those covered by previous results which is the
time of formation of the packet.  The proof of the theorem on Birkhoff
coordinates is based on a new quantitative version of a Vey type
theorem by Kuksin and Perelman which could be interesting in itself.
\end{abstract}

\section{Introduction and main result}
\label{section1}
It is well known that the Toda lattice, namely the system with
Hamiltonian 
\begin{equation}
H_{Toda}(p,q)=\frac{1}{2} \sum_{j=0}^{N-1}{p_j^2} + \sum_{j=0}^{N-1} e^{q_j-q_{j+1}} \ ,
\label{toda}
\end{equation}
and periodic boundary conditions $q_N=q_0\,$, $p_N=p_0$, is integrable
\cite{toda,henon}. Thus, by standard Arnold-Liouville theory the
system admits action angle coordinates. However the actual
introduction of such coordinates is quite complicated (see
\cite{Flaschka_McLaughlin,Ferguson}) and the corresponding
transformation has only recently been studied analytically in a series
of papers by Henrici and Kappeler \cite{kapphen1,kapphen2}. In
particular such authors have proved the existence of global Birkhoff
coordinates, namely canonical coordinates $(x_k,y_k)$ analytic on the
whole $\R^{2N}$, with the property that the $k^{th}$ action is given
by $(x^2_k+y^2_k)/2$.  The construction of Henrici and Kappeler,
however is not uniform in the size of the chain, in the sense that the
map $\Phi_N$ introducing Birkhoff coordinates is globally analytic for
any fixed $N$, but it could (and actually does) develop singularities
as $N\to + \infty$. Here we prove some analyticity properties
fulfilled by $\Phi_N$ uniformly in the limit $N \to +
\infty$. Precisely we consider complex balls centered at the origin
and prove that $\Phi_N$ maps analytically a ball of radius
$R/N^\alpha$ in discrete Sobolev-analytic norms into a ball of radius
$R'/N^\alpha$, with $R,R'>0$ independent of $N$ if and only if
$\alpha\geq 2$. Furthermore we prove that the
supremum of $\Phi_N$ over a complex ball of radius $R/N^\alpha$ diverges as
$N\to+\infty$ when $\alpha<1$.

In order to prove upper estimates on $\Phi_N$ we apply
to the Toda lattice a Vey type theorem \cite{vey} for infinite
dimensional systems recently proved by Kuksin and Perelman
\cite{kuksinperelman}. Actually, we need to prove a new quantitative
version of Kuksin-Perelman's theorem. We think that such a result
could be interesting in itself.

The lower estimates on the size of $\Phi_N$ are proved by constructing
explicitly the first term of the Taylor expansion of $\Phi_N$ through
Birkhoff normal form techniques; in particular we prove that the
second differential $d^2\Phi_N(0)$ at the origin diverges like $N^2$.

We finally apply the result to the problem of equipartition of energy
in the spirit of Fermi-Pasta-Ulam. We prove that in the Toda lattice,
corresponding to initial data with energy $E/N^3$ ($0<E\ll1$)
and with only the first Fourier mode excited, the energy remains
forever in a packet of Fourier modes exponentially decreasing with the
wave number. Then we consider the original FPU model and prove that,
corresponding to the same initial data, energy remains in an
exponentially localized packet of Fourier modes for times of order
$N^4$ 
(see Theorem \ref{N.1} below), namely for times one order of
magnitude longer then those covered by previous results (see
\cite{BamPon}, see also \cite{schneider_wayne,lubich}).  This is
relevant in view of the fact that the time scale of formation of the
packet is $N^3$ (see \cite{BamPon}), so the result of the
present paper allows to conclude that the packet persists over a time
much longer then the one needed for its formation.
 
\vspace{1em}

\noindent {\it Acknowledgments.} A particular thank goes to Antonio
Ponno who suggested the argument leading to the proof of Theorem
\ref{inverso}. We thank also Giancarlo Benettin, Andrea Carati, Luigi
Galgani, Antonio Giorgilli, Thomas Kappeler for several discussions on
this work.  This research was founded by the Prin project 2010-2011
``Teorie geometriche e analitiche dei sistemi Hamiltoniani in
dimensioni finite e infinite''. The second author was partially supported by the Swiss National Science Foundation.

\subsection{Birkhoff coordinates for the Toda lattice}
We come to a  precise statement of the main results of the present
paper. Consider the Toda lattice in the subspace characterized by
\begin{equation}
\label{media}
\sum_jq_j=0=\sum_{j}p_j
\end{equation}
which is invariant under the dynamics. Introduce the discrete
Fourier transform $\F(q) = \hat{q}$ defined by
\begin{equation}
\label{fou}
\hat{q}_k=\frac{1}{\sqrt{N}} \sum_{j=0}^{N-1}{q_j e^{2 \im \pi jk/N}},
\qquad k\in\Z\ ,
\end{equation}
and consider $\hat p_k$ defined analogously.  Due to \eqref{media} one
has $\hat p_0=\hat q_0=0$ and furthermore $\hp_k = \hp_{k+N}, \, \hq_k
= \hq_{k+N}$, $\forall k \in \Z$, so we restrict to $\{\hat p_k, \hat
q_k \}_{k=1}^{N-1}$. Corresponding to real sequences $(p_j,q_j)$ one
has $\overline{\hat{q}_k}=\hat{q}_{N-k}$ and $
\overline{\hat{p}_k}=\hat{p}_{N-k}$.

Introduce the linear Birkhoff variables 
\begin{equation}
\label{lin.bir.real}
X_k=\frac{\hat p_{k}+\hat p_{N-k}-\im\omega_k(\hat q_k-\hat q_{N-k})
}{\sqrt{2\omega_k} }\ ,\quad  Y_k=\frac{\hat p_{k}-\hat
  p_{N-k}+\im\omega_k(\hat q_k+\hat q_{N-k}) 
}{\im \sqrt{2\omega_k} }\ ,\quad k=1,...,N-1\ ,
\end{equation}
where $\omega_k\equiv\om k:=2\sin(k\pi/N)$; using such coordinates,
which are symplectic, the quadratic part
\begin{equation}
\label{h0}
H_0:=\sum_{j=0}^{N-1}\frac{p_j^2+(q_j-q_{j+1})^2}{2}
\end{equation} 
of the Hamiltonian takes the form
\begin{equation}
\label{quad.part}
H_0
= \sum_{k=1}^{N-1}\om k \frac{X_k^2+Y_k^2}{2}\ .
\end{equation}
With an abuse of notations, we re-denote by $H_{Toda}$ the Hamiltonian
\eqref{toda} written in the coordinates $(X,Y)$. 
The
following theorem is due to Henrici and Kappeler:
\begin{theorem}[{{\cite{kapphen2}}}]
\label{HK}For any integer $N \geq 2$ there exists a global real analytic symplectic diffeomorphism
$\Phi_N:\R^{N-1}\times\R^{N-1} \to \R^{N-1}\times\R^{N-1}$, $(X, Y)=\Phi_N(x,y)$
with the following properties:
\begin{itemize}
\item[(i)] The Hamiltonian $H_{Toda}\circ \Phi_N$ is a function of the
  actions $I_k:=\frac{x^2_k+y_k^2}{2}$ only, i.e. $(x_k,y_k)$ are 
  Birkhoff variables for the Toda Lattice.
\item[(ii)] The differential of $\Phi_N$ at the origin is the identity:
  $d\Phi_N(0,0)=\uno$. 
\end{itemize}
\end{theorem}

Our main results concern the analyticity properties of the map
$\Phi_N$ as $N\to \infty $. To come to a precise statement we have to
introduce a suitable topology in $\C^{N-1}\times \C^{N-1}$.

For any $s\geq 0$, $\sigma\geq0$ introduce in
$\C^{N-1}\times \C^{N-1}$ the discrete Sobolev-analytic norm
\begin{equation}
\norm{(X,Y)}^2_{\spazio{\reg}}:=\frac{1}{N}\sum_{k
  =1}^{N-1}[k]_N^{2s}\, e^{2\sigma [k]_N}\,\om{k}\,
\frac{\left|X_k\right|^2 + \left|Y_k\right|^2}{2}
\label{nor.bir}
\end{equation}
where $$[k]_N := \min (|k|, |N-k|)\ .$$ The space $\C^{N-1}\times \C^{N-1}$ endowed
by such a norm will be denoted by $\spazio{\reg}$. 
We denote 
by $B^{{\reg}}(R)$  the ball of
radius $R$ and center $0$ in the topology defined by the norm
$\norm{.}_{\spazio{\reg}}$. We will also denote by
$\Br^{{\reg}}:=B^{{\reg}}(R)\cap(\R^{N-1}\times\R^{N-1})$ the {\it real}
ball of radius $R$.
\begin{remark}
\label{sullenorme}
When $\sigma=s=0$ the norm \eqref{nor.bir} coincides with the energy
norm rescaled by a factor $1/N$ (the rescaling factor will be discussed in
Remark \ref{size}). We are particularly interested in the case
$\sigma>0$ since, in such a case, states belonging to $\spazio{\reg}$
are exponentially decreasing in Fourier space. The consideration of
positive values of $s$ will be needed in the proof of the main
theorem.
\end{remark}

Our main result is the following Theorem.

\begin{theorem}
\label{main}
For any $s\geq 0$, $\sigma\geq 0$ there exist
strictly positive constants
$R_{\reg}$, $C_\reg$,
such that for any $N\geq 2$, the map $\Phi_N$ is analytic as a map from
$B^{\reg}(R_{\reg}/N^2)$ to
$\spazio{\reg}$ and fulfills
\begin{align}
\label{Phi_N.est}
\sup_{\norm{(x,y)}_{\spazio{\reg}} \leq R/N^2
}{\norm{\Phi_N(x,y)-(x,y)}_{\spazio{s+1, \sigma}}}  \leq
C_\reg \frac{R^2}{N^2}\ ,\quad \forall R<R_{\reg} .  
\end{align}
The same estimated is fulfilled by the inverse map $\Phi_N^{-1}$
possibly with a different $R_{\reg}$. 
\end{theorem}

\begin{remark}
\label{nonlin}
%
The estimate \eqref{Phi_N.est} controls the size of the nonlinear
corrections in a norm which is stronger then the norm of
$(x,y)$, showing that $\Phi_N - \uno$ is 1-smoothing. The proof
of this kind of smoothing effect was actually the main aim of the work
by Kuksin and Perelman \cite{kuksinperelman}, which proved it for
KdV. Subsequently Kappeler, Schaad and Topalov \cite{kappelerschaad} proved
that such a smoothing property holds also globally for the KdV
Birkhoff map.
\end{remark}

\begin{remark}
\label{rem.ds}
As a consequence of \eqref{Phi_N.est} one has 
\begin{equation}
\label{phimain}
\Phi_N\left(B^{\reg}\left(\frac{R}{N^2}\right)\right)\subset
B^{\reg}\left(\frac{R}{N^2}\left(1+C_{\reg}R\right)\right)
 ,\quad \forall R<R_{\reg} ,\forall N\geq 2
\end{equation}
and the same estimate is fulfilled by the inverse map $\Phi_N^{-1}$,
possibly with a different $R_{\reg}$. 
\end{remark}

\begin{cor}
\label{corkuksinperelman}
For any $s\geq0$, $\sigma\geq 0$ there exist
strictly positive constants $R_{\reg}$, $C_\reg$, with the
following property. Consider the solution $v(t)\equiv(X(t),Y(t))$
of the Toda Lattice corresponding to initial data $v_0\in
B^{\reg}\left(\frac{R}{N^2}\right) $ with $R \leq R_{\reg}$ then one has
\begin{equation}
\label{unift}
v(t)\in B^{\reg}\left(\frac{R}{N^2}(1+C_{\reg}R)
\right)\ ,\quad \forall t\in\R\ .
\end{equation}
\end{cor}

In order to state a converse of Theorem \ref{main} consider the second
differential $Q^{\Phi_N}:=d^2\Phi_N(0,0)$ of $\Phi_N$ at the origin;
$Q^{\Phi_N}:\spazio{\reg}\to\spazio{\reg}$ is a
quadratic polynomial in the phase space variables\footnote{actually
  according to the estimate \eqref{Phi_N.est} it is smooth as a map
  $\spazio{\reg}\to\spazio{s+1,\sigma}$}. 

\begin{theorem}
\label{inverso}
For any $s \geq 0$, $\sigma \geq 0$ there exist strictly positive $R,
C$, $N_{\reg} \in \N$, such that, for any $N \geq N_{\reg}$, $\alpha
\in \R$, the quadratic form $Q^{\Phi_N}$ fulfills
\begin{equation}
\label{assuQ}
\sup_{v \in B^{\reg}_\R \left( \frac{R}{N^{\alpha}}\right)}
\norm{Q^{\Phi_N}(v,v)}_{\spazio{\reg}} \geq C R^2 N^{2-2\alpha} \ .
\end{equation}
\end{theorem}

\begin{remark}
\label{realdiv}
Roughtly speaking, one can say that, as $N\to\infty$, the real
diffeomorphism $\Phi_N$
develops a singularity at zero in the second derivative.
\end{remark}

Using Cauchy estimate (see subsect. \ref{p.inverso}) one immediately
gets the following corollary.

\begin{corollary}
\label{inverso1}
Assume that for some $s\geq0$, $\sigma\geq 0$ there exist strictly
positive $R,R'$ and $\alpha\geq0$, $\alpha'\in\R$, $N_{\reg}\in\N$,
s.t., for any $N\geq N_{\reg}$, the map $\Phi_N$ is analytic in the
complex ball $B^{\reg}(R/N^\alpha)$ and fulfills
\begin{equation}
\label{assu}
\Phi_N\left(B^{\reg}\left(\frac{R}{N^\alpha}\right)  \right)\subset
B^{\reg}\left(\frac{R'}{N^{\alpha'}}\right)\ ,
\end{equation}
then one has $\alpha'\leq 2(\alpha-1)$.
\end{corollary}
\begin{remark}
\label{alpha}
A particular case of Corollary \ref{inverso1} is $\alpha<1$, in which one
has that the image of a ball of radius $R N^{-\alpha}$ under $\Phi_N$
is unbounded as $N\to\infty$.

A further interesting case is that of $\alpha=\alpha'$, which implies
$\alpha\geq 2$, thus showing that the scaling $R/N^2$ is the best
possible one in which a property of the kind of \eqref{phimain} holds.
\end{remark}

\begin{remark}
\label{size}
A state $(X,Y)$ is in the ball $B^{{\reg}}(R/N^2)$ if
and only if there exist interpolating periodic functions
$(\beta,\alpha)$, namely functions s.t.
\begin{equation}
\label{interpol}
 p_j=\beta\left(\frac{j}{N}\right)\ ,\quad
q_j-q_{j+1}=\alpha\left(\frac{j}{N}\right)\ ,
\end{equation}
which are analytic in a strip of width $\sigma$ and have a
Sobolev-analytic norm of size $R/N^2$. More precisely, given a state
$(p, q)$ one considers its Fourier coefficients $(\hat{p}, \hat{q})$
and the corresponding $X,Y$ variables; define
$$ \alpha(x)= \frac{1}{\sqrt{N}} \sum_{k=0}^{N-1} \hat{q}_k
\left(1-e^{-2\pi \im k/N }\right) e^{-2\pi \im x k} , \qquad \beta(x)=
\frac{1}{\sqrt{N}} \sum_{k=0}^{N-1} \hat{p}_k e^{-2\pi \im x k}
$$ which fulfill \eqref{interpol}. Then the Sobolev-analytic norms of $\alpha$
and $\beta$ are controlled by $\norm{(X,Y)}_{\P^\reg}$. For
example one has
$$
\norm{(\alpha, \beta)}^2_{H^s}:= \norm{\alpha}_{L^2}^2 + \norm{\beta}_{L^2}^2 + \frac{1}{(2\pi)^{2s}} \norm{\partial_x^s \alpha}_{L^2}^2 + \frac{1}{(2\pi)^{2s}} \norm{\partial_x^s \beta}_{L^2}^2 = \norm{(X,Y)}^2_{\spazio{s,0}},
$$ where $ \norm{\alpha}_{L^2}^2 := \int_0^1 \mmod{\alpha(x)}^2\, dx$.
In particular we consider here states with Sobolev-analytic norm of
order $R/N^2$ with $R \ll 1$. The factor $1/N$ in the definition of
the norm was introduced to get correspondence between the norm of a
state and the norm of the interpolating functions.
\end{remark}

\begin{remark}
\label{compa}
As a consequence of Remark \ref{size}, the order in $N$ of the solutions we
are describing with Theorem \ref{main} is the same of the solutions
studied in the papers \cite{BamPon} and
\cite{bambuthomas3,bambuthomas,bambuthomas2}. 
\end{remark}

\begin{rem}
The results of Theorem \ref{main} and Theorem \ref{inverso}  extend to states with discrete Sobolev-Gevrey norm defined by
\begin{equation}
\norm{(X,Y)}^2_{\spazio{s, \sigma, \nu}}:=\frac{1}{N}\sum_{k
  =1}^{N-1}[k]_N^{2s}\, e^{2\sigma [k]_N^\nu }\,\om{k}\,
\frac{\left|X_k\right|^2 + \left|Y_k\right|^2}{2}
\end{equation}
where $0 \leq \nu \leq 1$. As a consequence of Remark \ref{size}, these states are interpolated by periodic functions with regularity Gevrey  $\nu$.
\end{rem}


This paper is part of a project aiming at studying the
dynamics of periodic Toda lattices with a large number of particles,
in particular its asymptotics.  First results in this project were
obtained in the papers \cite{bambuthomas3,bambuthomas,bambuthomas2}.
They are based on the Lax pair representation of the Toda lattice in
terms of periodic Jacobi matrices. The spectrum of these matrices
leads to a complete set of conserved quantities and hence determines
the Toda Hamiltonian and the dynamics of Toda lattices, such as their
frequencies.  In order to study the asymptotics of Toda lattices for a
large number $N$ of particles one therefore needs to work in two
directions: on the one hand one has to study the asymptotics of the
spectrum of Jacobi matrices as $N\to\infty$ and on the other hand, one
needs to use tools of the theory of integrable systems in order to
effectively extract information on the dynamics of Toda lattices from
the periodic spectrum of periodic Jacobi matrices.

The limit of a class of sequences of $N\times N$ Jacobi matrices as
$N\to\infty$ has been formally studied already at the beginning of the
theory of the Toda lattices (see e.g. \cite{toda}). However, as
pointed out in \cite{bambuthomas}, these studies only allowed to
(formally) compute the asymptotics of the spectrum in special
cases. In particular, Toda lattices, which incorporated right and left
moving waves could not be analyzed at all in this way.
In \cite{bambuthomas}, based on an approach pioneered in
\cite{thierrygolse}, the asymptotics of the spectra
 of sequences of Jacobi matrices corresponding to states of the form
\eqref{interpol} were rigorously derived by the means of semiclassical
analysis.  It turns out that in such a limit the spectrum splits into
three parts: one group of eigenvalues at each of the two edges of the
spectrum within an interval of size $O(N^{-2})$, whose asymptotics are
described by certain Hill operators, and a third group of eigenvalues,
consisting of the bulk of the spectrum, whose asymptotics coincides
with the one of Toda lattices at the equilibrium --
see \cite{bambuthomas} for details.

In \cite{bambuthomas2} the asymptotics of the eigenvalues obtained in
\cite{bambuthomas} were used in order to compute the one of the actions 
and of the frequencies of Toda lattices. In particular it was shown
that the asymptotics of the frequencies at the two edges involve the
frequencies of two KdV solutions. The tools used
in \cite{bambuthomas2} are those of the theory of infinite dimensional
integrable systems as developed in \cite{kamkdv} and adapted to the
Toda lattice in \cite{kapphen1}. 

The present paper takes up another important topic in the large number 
of particle limit of periodic Toda lattices: we study the Birkhoff coordinates 
near the equilibrium in the limit of large $N$ to provide precise estimates 
on the size of complex balls around the equilibrium in Fourier coordinates 
and the corresponding size in Birkhoff coordinates. Our analysis allows
to describe the evolution of Toda lattices with large number of particles 
in the original coordinates and to obtain an application to the study of FPU
lattices (on which we will comment in the next section).

We remark that the obtained estimates on the size of the complex balls 
are optimal. In our view this is a strong indication that beyond such a regime
 the standard tools of integrable systems become inadequate for studying
the asymptotic features of the dynamics of the periodic Toda lattices
as $N \to \infty$.

The proofs of our results are based on a novel technique developed
in \cite{kuksinperelman} to show a Vey type theorem for the KdV equation on
the circle which we adapt here to the study of Toda lattices, developing
in this way another tool for the study of periodic Toda lattices with
a large number of particles. We remark that for our arguments to go
through, we need to assume an additional smallness condition on the
set of states admitted as initial data: the states are required to be
interpolated by functions $\alpha$ and $\beta$ with Sobolev-analytic
norm of size $R/N^2$, with $R\ll1$ sufficiently small.  (In the
papers \cite{bambuthomas3,bambuthomas,bambuthomas2}, the size $R$ can
be arbitrarily large.)


\subsection{On the FPU metastable packet}\label{FPU}

In this subsection we recall the phenomenon of the formation of a
packet of modes in the FPU chain and state our related results.  First
of all we recall that the FPU $(\alpha, \beta)$-model is the
Hamiltonian lattice with Hamiltonian function which, in suitable
rescaled variables, takes the form
\begin{align}
& H_{FPU}(p, q)=\sum_{j=0}^{N-1}\frac{p_j^2}{2}+U(q_j-q_{j+1})\ \
, \label{HN}
\\ 
& U(x)=\frac{x^2}{2}+\ \frac{x^3}{6} + \beta \frac{x^4}{24}\  .
\label{pot} 
\end{align}
We will consider the case of periodic boundary conditions: $q_{0}=q_{N}, \, p_0 = p_{N}$.
\begin{rem}
\label{rem:H2H3}
One has 
$$H_{FPU}(p,q) = H_{Toda}(p,q) + (\beta-1)H_2(q) + H^{(3)}(q), $$ 
where
\begin{align*}
H_l( q) &:= \sum_{j=0}^{N-1} \frac{(q_j - q_{j+1})^{l+2}}{(l+2)!}\ ,\quad \forall
l\geq2\ ,
\\
H^{(3)}&:= -\sum_{l \geq 3} H_l\ .
\end{align*}
\end{rem}
Introduce the energies of the normal
modes by 
\begin{equation}
\label{Emodi}
E_k:=\frac{|\hat p_k|^2+\om{k}^2|\hat q_k|^2}{2}\ ,\quad 1 \leq k \leq N-1
\ ,
\end{equation}
correspondingly denote by 
\begin{equation}
\label{specific.energy}
\E_k:= \frac{E_k}{N}
\end{equation}
the specific energy in the $k^{th}$ mode. Note that since $p, q$ are
real variables, one has $\E_k = \E_{N-k}$.  \\ In their celebrated
numerical experiment Fermi Pasta and Ulam \cite{FPU}, being interested
in the problem of foundation of statistical mechanics, studied both
the behaviour of $\E_k(t)$ and of its time average
$$
\langle \E_k\rangle(t):=\frac{1}{t}\int_0^t \E_k(s)ds\ .
$$
They observed that, corresponding to initial data with
$\E_1(0)\not=0$ and $\E_k(0)=0$ $\forall k\not=1, N-1$, the quantities
$\E_k(t)$ present a recurrent behaviour, while their averages
$\langle \E_k\rangle(t)$ quickly relax to a sequence $\bar \E_k$
exponentially decreasing with $k$. This is what is known under the
name of FPU packet of modes. 

Subsequent numerical observations have investigated the
persistence of the phenomenon for large $N$ and have also shown that
after some quite long time scale (whose precise length is not yet
understood) the averages $\langle \E_k\rangle(t)$ relax to
equipartition (see
e.g. \cite{berchiallagalganigiorgilli,berchiallagiorgillipaleari,benettin_ponno,benettin_ponno2}).
This is the phenomenon known as metastability of the FPU packet.

The idea of exploiting the vicinity of FPU with Toda in order to study
the dynamics of FPU goes back to \cite{Ferguson}, in which the authors
performed some numerical investigations studying the evolution of the
Toda invariants in the dynamics of FPU. A systematic numerical study
of the evolution of the Toda invariants in FPU, paying particular
attention to the dependence on $N$ of the phenomena, was performed by
Benettin and Ponno \cite{benettin_ponno} (see also \cite{benettin_ponno2}). In
particular such authors put into evidence the fact that the FPU packet
seems to have an infinite lifespan in the Toda lattice. Furthermore
they showed that the relevant parameter controlling the lifespan of
the packet in the FPU model is the distance of FPU from the
corresponding Toda lattice.

Our Theorem \ref{main} yields as a corollary the effective existence and
infinite persistence of the packet in the Toda lattice and also an
estimate of its lifespan in the FPU system, estimate in which the
effective parameter is the distance between Toda and FPU. 

\vskip10pt
It is convenient to state the results for Toda and FPU using the small
parameter
$$
\mu := \frac{1}{N}
$$
as in \citep{BamPon}.

The following corollary is an immediate consequence of Corollary
\ref{corkuksinperelman}.
\begin{corollary}
\label{M.1}
Consider the Toda lattice \eqref{toda}. Fix $\sigma>0$,
then there exist constants $R_0,$ $C_1,$ such that the following
holds true. Consider an initial datum with
\begin{equation}
\label{M.1.1}
\E_{1}(0)=\E_{N-1}(0)=R^2e^{-2\sigma} \mu^4\ \ ,\quad
\E_k(0)\equiv \E_k(t)\big\vert_{t=0}=0\ ,\quad \forall
k\not=1, N-1\ 
\end{equation}
with $R<R_0 $.  Then, along the corresponding solution, one has
\begin{equation}
\label{M.1.2}
\E_k(t)\leq R^2(1+ C_1R) \mu^4 e^{-2\sigma k}\ ,\quad \forall \, 1\leq k \leq \lfloor N/2 \rfloor \ ,\quad
\forall t\in\R \ .
\end{equation}
\end{corollary}
For the FPU model we have the following corollary
\begin{theorem}
\label{N.1}
Consider the FPU system \eqref{HN}. Fix $s\geq 1$ and $\sigma\geq0$;
then there exist constants $R'_0,$ $C_2,$ $T,$ such that the
following holds true. Consider a real initial datum fulfilling
\eqref{M.1.1} with $R<R'_0$, then, along the corresponding solution,
one has
\begin{equation}
\label{N.1.2}
\E_k(t)\leq \frac{16R^2 \mu^4 e^{-2\sigma k}}{k^{2s}}
\ ,\quad \forall \, 1\leq k \leq \lfloor N/2 \rfloor \ ,\quad
|t| \leq \frac{T}{R^2\mu^4}\cdot\frac{1}{|\beta-1| + C_2 R\mu^2} \ .
\end{equation}
Furthermore, for $1 \leq k \leq N-1$, consider the action
$I_k:=\frac{x_k^2+y_k^2}{2}$ of the Toda lattice and let $I_k(t)$ be
its evolution according to the FPU flow. Then one has
\begin{equation}
\label{N.1.3}
\frac{1}{N}\sum_{k =1}^{N-1}{[k]_N^{2(s-1)}e^{2\sigma  [k]_N}\om{k}
\mmod{I_k(t) - I_k(0)}}
  \leq C_3 R^2\mu^5 \qquad \mbox{ for } t \mbox{ fullfilling } \eqref{N.1.2}
\end{equation}
\end{theorem}

\begin{rem}
The estimates \eqref{N.1.2} are stronger then the corresponding
estimates given in \cite{BamPon}, which are
$$
\E_k(t)\leq C_1 \mu^4 e^{-\sigma k} + C_2 \mu^5\ ,\quad \forall \, 1\leq k \leq \lfloor N/2 \rfloor \ ,\quad
|t| \leq \frac{T}{ \mu^3}.
$$ First, the time scale of validity of \eqref{N.1.2} is one order
longer than that of \cite{BamPon}. Second we show that as $\beta$
approaches the value corresponding to the Toda lattice (1 in our
units) the time of stability improves. Third the exponential estimate
of $\E_k$ as a function of $k$ is shown to hold also for large values
of $k$ (the $\mu^5$ correction is missing). Finally in \cite{BamPon}
it was shown that $T/\mu^3$ is the time of formation of the metastable
packet. So we can now conclude that the time of persistence of the
packet is at least one order of magnitude larger (namely ${\mu^{-4}}$)
with respect to the time needed for its formation.
\end{rem}


\begin{remark}
\label{lubich}
We recall also the result of \cite{lubich} in which the authors obtained
a control of the dynamics for longer time scales, but for initial data
with much smaller energies.
\end{remark}

\begin{remark}
\label{thermo}
 Recently some results on energy sharing in FPU in the thermodynamic
 limit \cite{maiocchi_bambusi_carati}(see also
 \cite{carati,carati_maiocchi,giorgilliPP}) have also been obtained,
 however such results are not able to explain the formation and the
 stability of the FPU packet of modes.
\end{remark}



\section{A quantitative Kuksin-Perelman Theorem}
\label{KP.section}
\subsection{Statement of the theorem}

In this section we state and prove a quantitative version of Kuksin-Perelman
Theorem which will be used to prove Theorem \ref{main}. It is
convenient to formulate it in the framework of weighted $\ell^2$ spaces,
that we are going now to recall. 
\newline
For any $N \leq \infty$, given a sequence $w=\{w_k\}_{k= 1}^N$,
$w_k>0$ $\forall k\geq 1$, consider the space $\ell^2_w $ of complex sequences $\xi = \{ \xi_k\}_{k=1}^N $ with norm
\begin{equation}
\label{ell_w.norm}
 \norm{\xi}_w^2:=\sum_{k
  = 1}^N w_k^2 |\xi_k|^2  < \infty .
\end{equation}
   Denote by $\spazio{w}$ the complex Banach space $\spazio{w}:=\ell^2_w \oplus \ell^2_w \ni (\xi, \eta)$ endowed with the norm $\norm{(\xi, \eta)}_w^2 := \norm{\xi}_w^2 + \norm{\eta}_w^2$. We denote by $\spazior{w}$ the real subspace of $\spazio{w}$ defined by
  \begin{equation}
\label{C^N_w}
\spazior{w}:=\left\{(\xi, \eta) \in \spazio{w}\, :  \ \eta_k = \overline{\xi}_k \ \forall \, 1 \leq k \leq N  \right\}.
  \end{equation}
We will denote by $B^{w}(\rho)$ (respectively $\Br^w(\rho)$) the
 ball in the topology of $\spazio w$ (respectively $\spazior{w}$) with center $0$ and radius $\rho>0$.

\begin{remark}
\label{rem.toxi}
In the case of the Toda lattice the variables $(\xi,\eta)$ are defined
by 
\begin{equation}
\label{xi_variable}
\xi_k = \frac{\hat{p}_k + \im \om{k}\hat{q}_{k}}{\sqrt{2{\om{k}}}},
\quad \eta_k = \frac{\hat{p}_{N-k} -\im
  \om{k}\hat{q}_{N-k}}{\sqrt{2\om{k}}}, \qquad 1 \leq k \leq N-1 \ ,
\end{equation}
and their connection with the real Birkhoff variables is given by 
\begin{equation}
\label{rebirl}
X_{k}=\frac{\xi_k+\eta_k}{\sqrt 2}\ ,\quad
Y_{k}=\frac{\xi_k-\eta_k}{i\sqrt 2}\ , \quad 1 \leq k \leq N-1 \ .
\end{equation}
\end{remark}

 We denote by $\spazio 1$ the  Banach space of sequences in which
 all the weights $w_k$ are equal to $1$. For $\X, \Y$ Banach spaces, we
 shall write $\L(\X,\Y)$ to denote the set of linear and bounded
 operators from $\X$ to $\Y$. For $\X=\Y$ we will write just $\L(\X)$.
 \begin{rem}
 In the application to the Toda lattice with $N$ particles we will use a finite, but not
 fixed $N$ and weights of the form
$w_k^2=w_{N-k}^2= N^3\, k^{2s}\,e^{2\sigma k}$, $1 \leq k \leq
 \lfloor N/2 \rfloor$.
\end{rem}
Given two weights $w^1$ and $w^2$, we will say that
$w^1 \leq w^2$  iff $w^1_k \leq w^2_k,$  $\forall k$.
Sometimes, when there is no risk of confusion, we will omit the index $w$ from the different quantities.\\
In $\spazio1$ we will use the  scalar product 
\begin{equation}
\label{scalar_product}
\<(\xi^1, \eta^1), (\xi^2, \eta^2)\>_c:= \sum_{k=1}^{N} \xi_k^1 \overline{\xi}^2_k + \eta^1_k \overline{\eta}^2_k  \ .
\end{equation}
Correspondingly,  the scalar product and symplectic form on  the real subspace $\spazior{w}$ are given for $\xi^1 \equiv (\xi^1, \bar \xi^1)$ and $\xi^2 \equiv (\xi^2, \bar \xi^2)$ by
\begin{equation}
\label{scalar_productR}
\<\xi^1,\xi^2\>:= 2 Re \sum_{k=1}^{N} \xi^1_k\,\overline{\xi^2_k}\ , 
\qquad 
\omega_0(\xi^1,\xi^2):=\< E \, \xi^1,\xi^2\> \ ,
\end{equation}
where $E:= -\im$. 

Given a smooth $F: \spazior w \to \C$, we denote by $X_F$ the
Hamiltonian vector field of $F$, given by $ X_F = J \nabla F$, where
$J = E^{-1}$.  For $F, G: \spazior w \to \C$ we denote by $\pb{F}{G}$ the
Poisson bracket (with respect to $\omega_0$): $\pb{F}{G}:= \< \nabla F,
J \nabla G\>$ (provided it exists).  We say that the
functions $F, G$ \textit{commute} if $\pb{F}{G} = 0$.
\vspace{1em}\newline
In order to state the main abstract theorem  we start by  recalling the notion of normally analytic map, exploited also in \cite{nikolenko} and  \cite{bambusi.grebert}. 
\newline
First we recall that a map $\tilde P^r:(\spazio{w})^r\to \B$, with $\B$ a  Banach space,
is said to be  $r$-\textit{multilinear} if $\tilde P^r(v^{(1)},\ldots,v^{(r)})$ is linear in
each variable $v^{(j)}\equiv (\xi^{(j)}, \eta^{(j)})$; a $r$-multilinear map is said to be \textit{bounded} if there
exists a constant $C>0$ such that
$$\norm{\tilde P^r(v^{(1)},\ldots,v^{(r)})}_\B \leq C \norm{v^{(1)}}_w\ldots \norm{v^{(r)}}_w \quad
\forall v^{(1)},\ldots, v^{(r)} \in \spazio{w}.
$$ Correspondingly its norm is defined by
$$\norm{\tilde P^r}:=\sup_{\norm{v^{(1)}}_w, \cdots, \norm{v^{(r)}}_w\leq 1}{\norm{\tilde P^r(v^{(1)},\cdots,v^{(r)})}_\B}.$$
A map $P^r:\spazio{w}\rightarrow \B $ is a \textit{homogeneous polynomial} of order $r$ if there exists
a $r$-multilinear map $\tilde{P}^r:(\spazio{w})^r\rightarrow \B $ such that
\begin{equation}
\label{polin}
P^r(v)=\tilde{P}^r(v,\ldots,v)\quad \forall v\in \spazio{w}\ .
\end{equation}
A $r$- homogeneous  polynomial is bounded if it has  finite norm
$$\norm{P^r}:=\sup_{\norm{v}_w\leq 1}\norm{P^r(v)}_\B.$$

\begin{oss}
Clearly $\norm{P^r}\leq \norm{\tilde{P}^r}$. Furthermore one has 
$\norm{\tilde{P}^r}\leq e^r\norm{P^r}$  -- cf. \cite{mujica}.
\end{oss}

It is easy to see that a multilinear map and the corresponding polynomial are continuous (and analytic) if and only if they are bounded.

Let $P^r:\spazio{w}\to\B$ be a homogeneous polynomial of order $r$; assume $\B$ separable and let
$\left\{\bb_n\right\}_{n \geq 1}\subset \B$ be a basis for the space $\B$. Expand $P^r$ as follows
\begin{equation}
\label{exp.1}
P^r(v)\equiv P^r(\xi, \eta)=\sum_{\substack{ |K| + |L| = r \\ n \geq 1}} P^{r,n}_{K, L} \xi^K \eta^L \bb_n ,  
\end{equation}
where $K, L \in \N^N_0$, $\N_0 = \N \cup \{ 0 \}$, $|K|:= K_1 + \cdots + K_N$, $\xi \equiv \{ \xi_j\}_{j \geq 1}$ and $\xi^K \equiv \xi_1^{K_1} \cdots \xi_{N}^{K_N}$, $\eta^L \equiv \eta_1^{L_1} \cdots \eta_N^{L_N}$.
\begin{df}
The modulus of a polynomial $P^r$ is the polynomial $\mod {P^r} $
defined by
\begin{equation}
\label{exp.11}
\mod{P^r}(\xi, \eta):=\sum_{\substack{ |K| + |L| = r \\ n \geq 1}} \mmod{P^{r,n}_{K, L}} \xi^K \eta^L \bb_n .
\end{equation}
A polynomial $P^r$ is said to have {\sl bounded modulus} if $\mod{P^r}$ is
a bounded polynomial.
\end{df}
A map $F: \spazio{w} \rightarrow \B$ is said to be an \textit{ analytic
  germ} if there exists $\rho>0$ such that $F: B^{w}(\rho) \to
\B$ is analytic. Then $F$ can be written as a power series
absolutely and uniformly convergent in $B^{w}(\rho)$:
$F(v)=\sum_{r\geq 0}{F^r(v)}$.  Here $F^r(v)$ is a homogeneous
polynomial of degree $r$ in the variables $v=(\xi, \eta)$. We will write $F =
O(v^n)$ if in the previous expansion $F^r(v)= 0$ for every $r < n$.
\begin{df}
\label{def.na} 
An  analytic germ $F: \spazio{w} \to \B$ is said to be  {\sl normally
  analytic} if there exists $\rho>0$ such that 
\begin{equation}
\label{exp.3}
\mod F(v):=\sum_{r\geq 0} \mod{F^r}(v)
\end{equation}
is absolutely and uniformly convergent in  $B^{w}(\rho)$. In such a case we will write
$F\in \Nc_\rho(\spazio w, \B)$. $\Nc_\rho(\spazio w, \B)$ is a Banach space when endowed by the norm
\begin{equation}
\label{Nc.norm}
\left|\mod F \right|_{\rho}:=
\sup_{v\in B^{w}(\rho)} \| \mod{F}(v)\|_{\B}.
\end{equation}
Let $U \subset \spazior{w}$ be open.  A map $F: U \to \B$ is said to
be a {\sl real analytic germ} (respectively {\sl real normally
  analytic}) on $U$ if for each point $u \in U$ there exist a
neighborhood $V$ of $u$ in $\spazio{w}$ and an analytic germ
(respectively normally analytic germ) which coincides with $F$ on $U
\cap V$.
\end{df}
\begin{oss}
It follows from Cauchy inequality that the Taylor polynomials $F^r$ of $F$ 
satisfy
\begin{equation}
\label{exp.4}
\norm{\underline{F}^r(v)}_\B \leq \left|\mod F\right|_{\rho} \frac{\norm{v}^r_w}{\rho^r} \qquad
\forall v \in B^{w}(\rho)\ .
\end{equation}
\end{oss}
\begin{oss}
Since $\forall r \geq 1$ one has 
$\norm{F^r}\leq \norm{\und{F}^r},$
if $F \in \Nc_\rho(\spazio{w}, \B)$ then the Taylor series of $F$ is uniformly convergent in $B^{w}(\rho)$.
\end{oss}

The case $\B=\spazio{w}$ will be of particular importance; in this case
the basis $\{ \bb _j\}_{j \geq 1}$ will coincide with the natural
basis $\{\be_j\}_{j \geq 1}$ of such a space (namely the vectors with
all components equal to zero except the $j^{th}$ one which is equal to
$1$). We will consider also the case
$\B=\L(\P^{w^1},\P^{w^2})$ (bounded linear operators from
$\P^{w^1}$ to $\P^{w^2}$), where $w^1$ and $w^2$
are weights. Here the chosen basis is $\bb_{jk}=\be_j\otimes \be_k$
(labeled by 2 indexes).

\begin{remark}
\label{Diff}
For $v\equiv (\xi, \eta) \in \spazio 1$, we denote by $|v|$ the vector
of the modulus of the components of $v$: $|v| = (|v_1|, \ldots,
|v_N|)$, $|v_j|:= (|\xi_j|, |\eta_j|)$. If
$F\in\Nc_{\rho}(\P^{w^1},\P^{w^2})$ then $\underline{dF}(|v|) |u|\leq
d\underline F(|v|)|u|$ (see \cite{kuksinperelman}) and therefore, for
any $0<d<1$, Cauchy estimates imply that $dF\in
\Nc_{(1-d)\rho}(\P^{w^1}, \L(\P^{w^1},\P^{w^2}))$ with
\begin{equation}
\label{diff.1}
\left|\underline{dF}\right|_{\rho(1-d)}\leq
\frac{1}{d\rho}\left|{\underline F}\right|_{\rho}\ ,
\end{equation}   
where $\mod{ dF}$ is computed with respect to the basis $\be_j\otimes \be_k$.
\end{remark}
Following Kuksin-Perelman \cite{kuksinperelman} we will need also a further property. 
\begin{df} A normally analytic  germ $F\in\Nc_\rho(\P^{w^1},\P^{w^2})$ will be said to be of
  class $\mathcal{A}_{ w^1, \rho}^{w^2}$ if $F=O(v^2)$ and the map $v
  \mapsto dF(v)^* \in\Nc_\rho(\P^{w^1}, \L(\P^{w^1},\P^{w^2}))$. Here
  $dF(v)^*$ is the adjoint operator of $dF(v)$ with respect to the
  standard scalar product \eqref{scalar_product}.
 On $\A_{ w^1, \rho}^{w^2}$ we will use the norm
\begin{equation}
\label{nor.arho}
\norma{F}_{\A_{ w^1, \rho}^{w^2}}:=\left|\underline F\right|_\rho+{\rho}
\left|\underline{d F}\right|_\rho+\rho \left|\underline{d
  F^*}\right|_\rho .
\end{equation}
\end{df}
\begin{rem}
\label{rem:norm.in.A}
Assume that for some $ \rho > 0$ the map $F \in \A_{ w^1,
  \rho}^{w^2}$, then for every $0 < d \leq \tfrac{1}{2}$ one has
$\mmod{\und{F}}_{d\rho} \leq 2 d^2 \mmod{\und{F}}_{\rho}$ and
$\norm{F}_{\A_{ w^1, d\rho}^{w^2}} \leq 6d^2 \norm{F}_{\A_{ w^1,
    \rho}^{w^2}}$.
\end{rem}
A real normally analytic germ $F: \Br^{w^1}(\rho)\to \spazior{w^2}$
will be said to be of class $\Nc_\rho(\spazior{w^1}, \spazior{w^2})$ (respectively $\mathcal{A}_{ w^1, \rho}^{w^2}$) if there
exists a map  of class  $\Nc_\rho(\spazio{w^1}, \spazio{w^2})$ (respectively
$\mathcal{A}_{ w^1, \rho}^{w^2}$), which coincides with $F$ on
$\Br^{w^1}(\rho)$. In this case we will also denote by $\mmod{\und{F}}_\rho$ (respectively
$\norma{F}_{\A_{ w^1, \rho}^{w^2}}$) the norm
defined by \eqref{Nc.norm} (respectively \eqref{nor.arho}) of the complex extension of $F$. 

Let now $F:U\subset \spazio{w^1} \to \spazio{w^2}$ be an analytic
map. We will say that $F$ {\sl is real for real sequences} if $ F (U
\cap \spazior{w^1}) \subseteq \spazior{w^2}$, namely $F(\xi,
\eta)=(F_1(\xi, \eta), F_2(\xi, \eta))$ satisfies $\overline{F_1(\xi,
  \bar \xi)} = F_2(\xi, \bar \xi)$. Clearly, the restriction $\left. F
\right|_{U \cap \spazior{w^1}}$ is a real analytic map.

\vspace{1em} 

We come now to the statement of the Vey Theorem.

Fix  $\rho>0$ and let $\Psi: \Br^{w^1}(\rho) \to \spazior{w^1} $,  $\Psi= \uno + \Psi^0$ with $\uno$
the identity map and $\Psi^0 \in \A_{w^1,\rho}^{w^2}$.  Write
$\Psi$ component-wise, $\Psi=\left\{(\Psi_j, \overline{\Psi}_j )\right\}_{j\geq 1}$, and
consider the foliation defined by the functions
$\left\{\left|\Psi_j(v)\right|^2/2\right\}_{j\geq 1}$.  Given
$v\in\spazior{w}$ we define the leaf through $v$ by
\begin{equation}
\label{leaf.1}
\F_v:=\left\{u\in \spazior{w} : \ \frac{|\Psi_j(u)|^2}{2}=\frac{|\Psi_j(v)|^2}{2}\  
\ ,\ \forall 
j\geq 1 \right\}\ .
\end{equation}
Let $\F=\bigcup_{v \in \spazior{w}} \F_v$ be the collection of all the leaves of the foliation.  We will
denote by $T_v\F$ the tangent space to $\F_v$ at the point $v\in \spazior{w}$. 
A relevant role will also be played by the function $I = \{ I_j \}_{j\geq 1}$ whose components are defined by
\begin{equation}
\label{Ij}
I_j(v)\equiv I_j(\xi, \bar \xi):=\frac{|\xi_j|^2}{2} \quad \forall j \geq 1\ .
\end{equation}
The foliation they define will be denoted by $\Fz$. 
\begin{oss}
$\Psi$ maps the foliation $\F$ into the foliation $\Fz$, namely $\F^{(0)} = \Psi(\F)$.
\end{oss}
The main theorem of this section is the following
\begin{teo}\label{KP} (Quantitative version of Kuksin-Perelman Theorem)
 Let $w^1$ and $w^2$ be weights with $w^1 \leq w^2$. Consider the
 space $\spazior{w^1}$ endowed with the symplectic form $\omega_0$
 defined in \eqref{scalar_productR}. Let $\rho >0$ and assume $\Psi:
 \Br^{w^1}(\rho) \to \spazior{w^1}$, $\Psi= \uno + \Psi^0$ and $\Psi^0
 \in \A_{ w^1, \rho}^{w^2}$. Define
\begin{equation}
\label{th}
\epsilon_1:=\norma{\Psi^0}_{\A_{ w^1, \rho}^{w^2}}\ .
\end{equation}
Assume that the functionals $\lbrace
\frac{1}{2}\left|\Psi_j(v)\right|^2\rbrace_{ j\geq 1}$ pairwise
commute with respect to the symplectic form $\omega_0$, and that
$\rho$ is so small that
\begin{equation}
\epsilon_1<2^{-34}\rho. 
\label{th.1}
\end{equation}
Then there exists a real normally analytic map
$\widetilde{\Psi}:\Br^{{w^1}}(a\rho)\to \spazior{w^1}$, $a=2^{-48}$,
with the following properties:
\begin{itemize}
\item[i)] $\widetilde{\Psi}^*\omega_0=\omega_0$, so that the coordinates
  $z:=\widetilde\Psi(v) $ are canonical;
\item[ii)] the functionals $\left\{ \frac{1}{2}\left|\widetilde\Psi_j(v)\right|^2\right\}_{ j\geq 1}$ pairwise commute with respect to the symplectic form $\omega_0$;
\item[iii)] $\F^{(0)}=\widetilde{\Psi}(\F),$ namely the foliation
  defined by $\Psi$ coincides with the foliation defined by
  $\widetilde{\Psi}$;
\item[iv)] $\widetilde{\Psi} = \uno + \widetilde{\Psi}^0$ with $\widetilde{\Psi}^0\in\A_{w^1, a\rho}^{w^2}$ and  $\norm{\widetilde{\Psi}^0}_{\A_{w^1, a\rho}^{w^2}}\leq 2^{17}\epsilon_1$.
\end{itemize}
\label{teokuksinperelman}
\end{teo}
The following corollary holds:
\begin{cor}
\label{cor.KP}
Let $H: \spazior{w^1} \to \R$ be a real analytic Hamiltonian
function. Let $\Psi$ be as in Theorem \ref{KP} and assume that for
every $j \geq 1$, $\mmod{\Psi_j(v)}^2$ is an integral of motion for
$H$, i.e.
\begin{equation}
\label{ass.1}
\lbrace H, |\Psi_j|^2 \rbrace = 0 \quad \forall\,j \geq 1.
\end{equation}
Then the coordinates $(x_j,y_j)$ defined by $x_j+iy_j=\widetilde
\Psi_j(v)$ are real Birkhoff coordinates for $H$, namely canonical
conjugated coordinates in which the Hamiltonian depends only on
$(x_j^2+y_j^2)/2 $.
\end{cor}

\noindent {\em Proof of Corollary \ref{cor.KP}.}  Since $\Psi = \uno +
\Psi^0$, the functions $\Psi_j(v)$ can be used as coordinates in a
suitable neighborhood of $0$ in $\spazior{w}$. Let $\widetilde{\Psi}$ be the map
in the statement of Theorem \ref{KP}. Denote $F_l(v) :=
\frac{1}{2}\mmod{\widetilde{\Psi}_l(v)}^2$.  Since the foliation defined
by the functions $\{F_l\}_{l \geq 1}$ and the foliation defined by
$\{\mmod{\Psi_j}^2\}_{j \geq 1}$ coincide (Theorem \ref{KP} $iii)$),
each $F_l$ is constant on the level sets of $\{\mmod{\Psi_j}^2\}_{j
  \geq 1}$. It follows that each $F_l$ is a function of
$\{\mmod{\Psi_j}^2\}_{j \geq 1}$ only.  Since $\forall \, j\geq 1$,
$\mmod{\Psi_j}^2$ is an integral of motion for $H$, the same is true
for $F_l$, $\forall l \geq 1$.  Define now, in a suitable neighborhood
of the origin, the coordinates $(z, \bar z)$ by $z_j \equiv
\widetilde{\Psi}_j$, $ \bar{z}_j \equiv \overline{\widetilde{\Psi}_j}$. Of course $F_l=\tfrac{|z_l|^2}{2}$. By
\eqref{ass.1} it follows then that
\begin{equation}
\label{comm.rel.2}
0 = \lbrace H, z_l \bar z_l \rbrace  =
\frac{1}{\im}\left(\frac{\partial H}{\partial z_l} z_l - \frac{\partial
  H}{\partial \bar z_l} \bar z_l \right).
\end{equation}
Since $d\widetilde{\Psi}(0) = \uno \,$ (Theorem \ref{KP} $iv)$),
$\widetilde{\Psi}$ is invertible and its inverse $\widetilde\Phi$
satisfies $\widetilde{\Phi} = \uno + \widetilde{\Phi}^0$ with
$\widetilde{\Phi}^0\in\A_{w^1, a\mu\rho}^{w^2}$ and
$\norm{\widetilde{\Phi}^0}_{\A_{w^1,a\mu\rho}^{w^2}} \leq 2
\norm{\widetilde{\Psi}^0}_{ \A_{w^1,a\rho}^{w^2}} \leq
2^{18}\epsilon_1$ (Lemma \ref{FGinA} $ii)$ in Appendix
\ref{propnagerms}).

Expand now $H\circ \widetilde{\Phi}$ in Taylor
series in the variables $(z, \bar z)$:
$$ H\circ \widetilde \Phi (z, \bar z) = \sum_{\substack{r \geq 2,
    \\ |\alpha| + |\beta| =r }} H^r_{\alpha, \beta} z^\alpha \bar
z^\beta.
$$ Then equation \eqref{comm.rel.2} implies that in each term of the
summation $\alpha = \beta$, therefore $H\circ \widetilde\Phi$ is a
function of $|z_1|^2, \ldots, |z_N|^2$.  Define now the real variables
$(x, y)$ as in the statement, then the claim follows immediately.
\qed

\subsection{Proof of the Quantitative Kuksin-Perelman Theorem}

In this section we recall and adapt Eliasson's proof \cite{eliasson}
of the Vey Theorem following \cite{kuksinperelman}. As we anticipated
in the introduction, the novelty of our
approach is to add quantitative estimates on the Birkhoff map
$\widetilde{\Psi}$ of Theorem \ref{KP}.  In Appendix \ref{propnagerms}
we show that the class of normally analytic maps is closed under
several operations like composition, inversion and flow-generation,
and provide new quantitative estimates which will be used during the
proof below.
\vspace{1em}\newline
The idea of the proof of Theorem \ref{KP} is to consider the functions $\{\Psi_j(v)\}_{j \geq 1} $ as noncanonical
coordinates, and to look for a coordinate transformation introducing
canonical variables and preserving the foliation $\Fz$ (which is the
image of $\F$ in the noncanonical variables).

This will be done in two steps both based on the standard procedure of
Darboux Theorem that we now recall.  In order to construct a
coordinate transformation $\varphi$ transforming the closed
nondegenerate form $\Omega_1$ into a closed nondegenerate form
$\Omega_0$, then it is convenient to look for $\varphi$ as the time 1
flow $\varphi^t$ of a time-dependent vector field $Y^t$. To construct
$Y^t$ one defines $\Omega_t:= \Omega_0 + t(\Omega_1 - \Omega_0)$ and
imposes that
$$
0 = \left.\tfrac{d}{d t}\right|_{t=0} \varphi^{t*} \Omega_t = \varphi^{t*} \left(\L_{Y^t}\Omega_t + \Omega_1 - \Omega_0 \right) = \varphi^{t*} \left( d (Y^t\lrcorner \Omega_t ) + d(\alpha_1 - \alpha_0) \right)
$$
where $\alpha_1, \alpha_0$ are potential forms for $\Omega_1$ and $\Omega_0$ (namely $d\alpha_i = \Omega_i$, $i=0,1$) and $\L_{Y^t}$ is the Lie derivative of $Y^t$. Then one gets
\begin{equation}
\label{darboux.eq}
Y^t\lrcorner \Omega_t  + \alpha_1 - \alpha_0 = df
\end{equation}
for each $f$ smooth; then, if $\Omega_t$ is nondegenerate, this
defines $Y^t$. If $Y^t$ generates a flow $\varphi^t$ defined up to
time 1, the map $\varphi:= \left. \varphi^t\right|_{t=1}$ satisfies
$\varphi^*\Omega_1 = \Omega_0$.  Thus, given $\Omega_0$ and
$\Omega_1$, the whole game reduces to study the analytic properties of
$Y^t$ and to prove that it generates a flow.
\vspace{1em}

A non-constant symplectic form $\Omega$
will always be represented through a
linear skew-symmetric invertible operator $E$ as follows:
\begin{equation}
\label{symp.1}
\Omega(v)(u^{(1)};u^{(2)})=\langle E(v)u^{(1)};u^{(2)} \rangle\ ,\quad \forall
u^{(1)},u^{(2)}\in T_v\spazior w\simeq\spazior w .
\end{equation}
We denote by $\pb{F}{G}_{\Omega}$ the
Poisson bracket with respect to $\Omega$ : $\pb{F}{G}_\Omega:=
\< \nabla F, J \nabla G\>$, $J:=E^{-1}$. 

Similarly we will represent   $1$-forms through the vector field $A$ such that
\begin{equation}
\label{1form.1}
\alpha(v)(u) = \langle A(v), u \rangle, \quad \forall u \in T_v\spazior
w.
\end{equation}
\vspace{1em}
Define $\omega_1:=(\Psi^{-1})^*\omega_0$, and let $E_{\omega_1}$ be the  operator representing the symplectic form $\omega_1$. 
The first step consists in transforming $\omega_1$ to a symplectic form whose "average over $\F^{(0)}$" coincides with $\omega_0$. 

So we start by defining precisely what ``average of $k$-forms''
means. To this end consider the Hamiltonian vector fields $X^0_{I_l}$ of
the functions $I_l\equiv \frac{|v_l|^2}{2}$ through the symplectic form $\omega_0$; they are
given by
\begin{equation}
X^0_{I_l}(v)=\im \nabla I_l(v)=\im v_l {\bf e}_l ,  \quad \forall\, l \geq 1.
\label{X_I_l}
\end{equation}
  For every $l \geq 1$ the corresponding flow
$\phi_l^t \equiv \phi_{X^0_{I_l}}^t$ is given  by
$$ \phi^t_l(v) = \left(v_1, \cdots, v_{l-1}, e^{\im t}v_l, v_{l+1},
  \cdots \right) \ .$$ 
Remark that the map $\phi_l^t$ is linear in $v$, $2\pi$ periodic in $t$
and its adjoint satisfies
 $(\phi_l^t)^*=\phi_l^{-t}$.\\
Given a $k$-form $\alpha$ on $\spazior{w}$ $(k\geq0)$, we define its average by 
\begin{align}
\label{ave.1}
M_j\alpha(v)=\frac{1}{2\pi}\int_0^{2\pi}{((\phi_j^t)^*\alpha)(v)dt},
\quad j \geq 1\ , \qquad \mbox{and} \qquad
M\alpha(v)=\int_{\Tc}[(\phi^\theta)^*\alpha]\, d\theta
\end{align}
where $\Tc$ is the (possibly infinite dimensional) torus, the map
$\phi^\theta=(\phi_1^{\theta_1}\circ \phi_2^{\theta_2}\cdots)$ and
$d\theta$ is the Haar measure on $\Tc$.

\begin{oss}
\label{ave.J}
In the particular cases of 1 and 2-forms it is useful to compute the
average in term of the representations \eqref{symp.1} and
\eqref{1form.1}. Thus, for $v$, $u^{(1)}, u^{(2)} \in
\spazior{w}$, if
$$ \alpha(v)u^{(1)}=\langle A(v); u^{(1)}\rangle\ ,\quad
\omega(v)(u^{(1)},u^{(2)})=\langle E(v)\,u^{(1)};u^{(2)} \rangle\ ,
$$
one has 
\begin{equation}
(M\alpha)(v)u^{(1)}=\langle (MA)(v);\, u^{(1)}\rangle \ ,\quad \text{with}\quad
MA (v)=\int_{\Tc}{\phi^{-\theta}
  A(\phi^\theta(v))} 
\;d\theta 
\label{MA}
\end{equation}
and
\begin{equation}
(M\omega)(v)(u^{(1)},u^{(2)})=\langle (ME)(v)u^{(1)};\,u^{(2)}
\rangle\ ,\ \text{with}\ ME(v)=\int_{\Tc}{\phi^{-\theta}
  E(\phi^\theta(v))\phi^\theta} 
\; d\theta .
\label{ME}
\end{equation}
\end{oss}
\begin{oss}
\label{rem.M}
 The operator $M$ commutes with the differential operator $d$ and the rotations $\phi^\theta$. In particular  $MA(v)$ and $ME(v)$ as in  \eqref{MA}, \eqref{ME} satisfy 
$$
\phi^\theta MA(v) = MA(\phi^\theta v), \quad 
\phi^\theta ME(v)u = ME(\phi^\theta v)\phi^\theta u, \qquad \forall \, \theta \in \Tc  \ .
$$
\end{oss}
We study now the analytic properties of $\omega_1$ and of its potential form $\alpha_{\omega_1}$. In the rest of the section denote by $S:=\sum_{n=1}^\infty 1/n^2$ and by
\begin{equation}
\label{mu.def} 
\mu:=1/e(32 S)^{1/2}\approx 0.0507 \ .
\end{equation}
\begin{lem}
\label{lem.analy} Let $\Phi:=\Psi^{-1}$ and  $\omega_1$ be as above.
Assume that 
$\epsilon_1\leq \rho/e$. Then the following holds:
\begin{itemize}
\item[(i)] $E_{\omega_1}=-\im +\Upsilon_{\omega_1}$, with $\Upsilon_{\omega_1}\in\Nc_{\mu\rho}(\spazior{w^1}, \L(\spazior{w^1}, \spazior{w^2}))$ and
\begin{equation}
\label{JM}
|\mod{\Upsilon_{\omega_1}}|_{\mu\rho}\leq \frac{8 \epsilon_1}{\mu\rho}\ .
\end{equation}
\item[(ii)] Define 
\begin{equation}
\label{WM}
W_{\omega_1}(v):=\int_0^1 \Upsilon_{\omega_1}(tv)tv \;dt \ ,
\end{equation}
then $W_{\omega_1}\in\A_{w^1, \mu^3\rho}^{w^2}$ and $\norma{W_{\omega_1}}_{\A_{w^1, \mu^3\rho}^{w^2}}\leq 8 \epsilon_1 $. Moreover the 1-form $\alpha_{W_{\omega_1}}:=\langle W_{\omega_1}; .\rangle$ satisfies
$
d\alpha_{W_{\omega_1}}=\omega_1-\omega_0\ .
$
\end{itemize}
\end{lem}
\begin{proof}
By Lemma \ref{FGinA} one has that $\Phi = \left(\uno+ \Psi^0 \right)^{-1} = \uno + \Phi^0$ with $\Phi^0 \in \A_{w^1, \mu\rho}^{w^2} $ and  $\norm{\Phi^0}_{\A_{w^1, \mu\rho}^{w^2}} \leq 2 \norm{\Psi^0}_{\A_{w^1, \rho}^{w^2}} \leq 2 \epsilon_1$.
To prove $(i)$, just remark that 
\begin{align*}
E_{\omega_1}(v)&=d\Phi^*(v)(-\im)d\Phi(v)=-\im + d\Phi^0(v)^*(-\im) d\Phi(v)-\im d\Phi^0(v) =: -\im + \Upsilon_{\omega_1}(v)
\end{align*}
and use the results of Lemma  \ref{FGinA}. To prove $(ii)$, use
Poincar\'e construction of the potential of $\omega_1$ which gives 
$$\alpha_{\omega_1}(v)u:=\langle \int_0^1 E_{\omega_1}(tv)tv,u\rangle
dt=\alpha_0(v)u +\langle W_{\omega_1}(v),u \rangle, \quad W_{\omega_1}(v)=\int_0^1
\Upsilon_{\omega_1}(tv)tv \;dt \ ,$$
where $\alpha_0$ is the potential for $\omega_0$. In order to prove the analytic properties of $W_{\omega_1}$, note that $W_{\omega_1}(v) = \int_0^1 (H_1(tv) + H_2(tv)) dt$ where $H_1(v):= -\im \, d\Phi^0(v)v$ and $H_2(v):= d\Phi^0(v)^* (-\im) d\Phi(v)v \equiv d\Phi^0(v)^*(-\im v+ H_1(v))$. Thus, by Lemma \ref{FGinA}, one gets that
$
\norm{H_1}_{\A_{w^1, \mu^2\rho}^{w^2}}\leq 2 \norm{\Phi^0}_{\A_{w^1, \mu\rho}^{w^2}} \leq 4 \epsilon_1
$
and 
$
\norm{H_2}_{\A_{w^1, \mu^3\rho}^{w^2}}\leq 2 \norm{\Phi^0}_{\A_{w^1, \mu^2\rho}^{w^2}} \leq 4 \epsilon_1 \ .
$
Thus the estimate on $W_{\omega_1}$ follows.
\end{proof}

\begin{remark}
\label{rem.me}
One has $M\alpha_{\omega_1} - \alpha_0 = M\alpha_{W_{\omega_1}} = \< MW_{\omega_1}, \cdot \>$ and  $\norma{MW_{\omega_1}}_{\A_{w^1, \mu^3\rho}^{w^2}}
\leq\norma{W_{\omega_1}}_{\A_{w^1, \mu^3\rho}^{w^2}} $.
\end{remark}

We are ready now for the first step.
\begin{lem}There exists a map $\hat \varphi: \Br^{{w^1}}(\mu^5 \rho) \to \spazior{w^1}$ such that $(\uno-\hat\varphi)\in\A_{w^1, \mu^5\rho}^{w^2} $ and
\begin{equation}
\label{hat.sti.3}
\norma{\uno-\hat\varphi}_{\A_{w^1, \mu^5\rho}^{w^2}}\leq 2^5\epsilon_1\ .
\end{equation}
Moreover $\hat{\varphi}$ satisfies the following properties:
\begin{enumerate}[(i)]
\item $\hat \varphi$ commutes with the rotations $\phi^\theta$, namely $\phi^\theta \hat\varphi(v) = \hat\varphi(\phi^\theta v)$ for every $\theta \in \Tc $.
\item Denote $\hat{\omega}_1 := \hat{\varphi}^* \omega_1$, then  $M \hat{\omega}_1 = \omega_0$. 
\end{enumerate} 
\label{map.phi1}
\end{lem}
\begin{proof}
We apply the  Darboux procedure described at the beginning of this section with $\Omega_0 = \omega_0$ and $\Omega_1 = M\omega_1$. Then $\Omega_t$ is represented by the operator $\hat{E}_{\omega_1}^t := \left(-\im + t(ME_{\omega_1} +\im )\right) $. Write equation \eqref{darboux.eq}, with $f \equiv 0$,  in terms of the operators defining the symplectic forms, getting the equation 
$\hat{E}_{\omega_1}^t \hat Y^t = - MW_{\omega_1}$ (see also Remark \ref{rem.me}). This equation can be solved by inverting the operator $\hat{E}_{\omega_1}^t$ by Neumann series:
\begin{equation}
\label{equazione3.1}
\hat Y^t:=-(-\im +tM\Upsilon_{\omega_1})^{-1}
MW_{\omega_1}\ .
\end{equation}
By the results of Lemma \ref{lem.analy} and Remark \ref{rem.me}, $\hat{Y}^t$ is of class $\A_{w^1, \mu^4\rho}^{w^2}$ and fulfills
\begin{equation}
\label{sti.hat.1}
\sup_{t \in [0,1]}\norma{\hat Y^t}_{\A_{w^1, \mu^4\rho}^{w^2}}\leq 2 \norm{MW_{\omega_1}}_{\A_{w^1, \mu^3\rho}^{w^2}} \leq 2^4\epsilon_1 \ .
\end{equation}
By Lemma \ref{flussoinA} the vector field $\hat Y^t$ generates a flow $\hat{\varphi}^t: \Br^{{w^1}}(\mu^5 \rho) \to \P^{w^1} $ such that $\hat{\varphi}^t- \uno $ is of class $\A_{w^1, \mu^5\rho}^{w^2}$ and satisfies 
$$\norm{\hat \varphi^t - \uno}_{\A_{w^1, \mu^5 \rho}^{w^2}} \leq 2 \sup_{t \in [0,1]}\norma{\hat Y^t}_{\A_{w^1, \mu^4\rho}^{w^2}}\leq 2^5\epsilon_1 .$$
Therefore the map $\hat \varphi \equiv \hat \varphi^t\vert_{t=1}$ exists, satisfies the claimed estimate \eqref{hat.sti.3} and furthermore $\hat \varphi^* M\omega_1 = \omega_0$.
\newline
We prove now item $(i)$. The claim follows if we show  that the vector field  $\hat Y^t$ commutes with  rotations. To this aim consider  equation \eqref{equazione3.1}, and define $\hat{J}_{\omega_1}^t(v) = (\hat{E}_{\omega_1}^t(v))^{-1}$. By construction the operator $\hat{E}_{\omega_1}^t $ commutes with  rotations (cf. Remark \ref{rem.M}), namely  $\forall \, \theta_0 \in \Tc $  one has $\phi^{\theta_0}\hat{E}_{\omega_1}^t(v)
u=
\hat{E}_{\omega_1}^t(\phi^{\theta_0}(v))\phi^{\theta_0}u.$
Then it follows that 
\begin{align*}
\phi^{\theta_0}\hat{Y}^t(v)&=
-\phi^{\theta_0}\hat{J}_{\omega_1}^t(v)MW_{\omega_1}(v)
=-\hat{J}_{\omega_1}^t(\phi^{\theta_0}(v))\phi^{\theta_0}MW_{\omega_1}(v)\\
&=-\hat{J}_{\omega_1}^t(\phi^{\theta_0}(v))MW_{\omega_1}(\phi^{\theta_0}(v))
=\hat{Y}^t(\phi^{\theta_0}(v)).
\end{align*}
This proves item $(i)$. Item $(ii)$ then follows from item $(i)$ since, defining $\hat{\omega}_1=
\hat{\varphi}^*\omega_1$, one has the chain of identities
$
M \hat{\omega}_1 =  M \hat{\varphi}^*\omega_1 = \hat{\varphi}^* M \omega_1 = \omega_0
$.
\end{proof}

The analytic properties of the symplectic form  $\hat \omega_1$ can be studied in the same way as in  Lemma \ref{lem.analy}; we get therefore the following corollary:
\begin{corollary}
\label{cor.e}
Denote by $E_{\hat{\omega}_1}$ the symplectic operator describing $\hat{\omega}_1=\hat\varphi ^*\omega_1$. Then 
\begin{itemize}
\item[(i)] $E_{\hat{\omega}_1}=-\im +\Upsilon_{\hat{\omega}_1}$, with $\Upsilon_{\hat{\omega}_1}\in
\Nc_{\mu^5\rho}(\spazior{w^1}, \L(\spazior{w^1}, \spazior{w^2}))$ and
$\left|\mod{\Upsilon_{\hat{\omega}_1}}\right|_{\mu^5\rho}\leq 2^7\frac{\epsilon_1}{\mu\rho}.$
\item[(ii)] Define 
$W(v):=\int_0^1 \Upsilon_{\hat{\omega}_1}(tv)tv \;dt$, then $W\in\A_{w^1, \mu^7\rho}^{w^2}$ and
$\norma{W}_{\A_{w^1, \mu^7\rho}^{w^2}}\leq 2^7\epsilon_1 .$
\end{itemize}
Furthermore the 1-form $\alpha_{W}:= \langle W, . \rangle $ satisfies 
 $d\alpha_{W} = \hat{\omega}_1 - \omega_0.$
\end{corollary}
Finally we will need also some analytic and geometric properties of the map 
\begin{equation}
\label{checkPsi}
\check{\Psi}:= \hat{\varphi}^{-1} \circ \Psi.
\end{equation}
The functions $\{\check{\Psi}(v)\}_{j \geq 1}$ forms a new set of coordinates in a suitable neighborhood of the origin whose properties are given by the following corollary:
\begin{cor}
\label{cor.psi.1}
The map $\check{\Psi}:\Br^{{w^1}}(\mu^8\rho) \to \spazior{w^1}$, defined in \eqref{checkPsi},  satisfies the following properties:
\begin{enumerate}[(i)] 
\item $d \check{\Psi}(0) = \uno$ and $\check\Psi^0:= \check \Psi- \uno \in \A_{w^1,\mu^8 \rho}^{w^2}$ with $ \norm{\check \Psi^0}_{\A_{w^1,\mu^8 \rho}^{w^2}} \leq 2^8 \epsilon_1$.
\item  $\F^{(0)} = \check{\Psi}(\F)$, namely  the foliation defined by  $\check{\Psi}$ coincides with the foliation defined by $\Psi$.
\item The functionals $\lbrace \frac{1}{2}\left|\check{\Psi}_j\right|^2\rbrace_{ j\geq 1}$ pairwise commute with respect to the symplectic form $\omega_0$.
\end{enumerate}
\end{cor}
\begin{proof}
By Lemma \ref{FGinA} the map $\hat\varphi$ is invertible in
$\Br^{w^1}({\mu^6\rho})$ and $\hat\varphi^{-1}= \uno + g$, with $g
\in \A_{w^1,\mu^6\rho}^{w^2}$ and $\norm{g}_{\A_{w^1, \mu^6
    \rho}^{w^2}} \leq 2^6 \epsilon_1.$ Then $\check{\Psi} = \uno +
\check \Psi^0$ where $\check \Psi^0 = \Psi^0 + g\circ(\uno +
\Psi^0)$. By Remark \ref{rem:norm.in.A}, $\norm{\Psi^0}_{\A_{ w^1,
    \mu^7\rho}^{w^2}}\leq 6\mu^{14}\epsilon_1 $, thus Lemma
\ref{FGinA} $i)$ implies that $\check{\Psi}^0 \in
\A_{w^1,\mu^8\rho}^{w^2}$ and moreover $\norm{\check \Psi^0}_{\A_{
    w^1, \mu^8\rho}^{w^2}} \leq 6\mu^{14}\epsilon_1 + 2^7 \epsilon_1
\leq 2^8 \epsilon_1$.  Item $(ii)$ and $(iii)$ follow from the fact
that, by Lemma \ref{map.phi1} $(i)$, $\hat \varphi$ commutes with the
rotations (see also the proof of Corollary \ref{cor.KP}).
\end{proof}

\vspace{1em} The second step consists in transforming $\hat{\omega}_1$
into the symplectic form $\omega_0$ while preserving the functions
$I_l$. In order to perform this transformation, we apply once more the
Darboux procedure with $\Omega_1 = \hat{\omega}_1$ and $\Omega_0 =
\omega_0$. However, we require each leaf of the foliation to be
invariant under the transformation. In practice, we look for a change
of coordinates $\varphi$ satisfying
\begin{align}
	&\varphi^*\Omega_1=\Omega_0 \ , \label{fi.1}\\
	&I_l(\varphi(v))=I_l(v), \quad \forall \, l \geq 1  \ . \label{fi.2}
\end{align}
In order to fulfill the second equation, we take advantage of the arbitrariness of $f$ in equation \eqref{darboux.eq}. 
It turns out that if $f$ satisfies the set of differential equations given by
 \begin{equation}
\label{fi.3}
df(X^0_{I_l})-(\alpha_1-\alpha_0)(X^0_{I_l})=0, \quad \forall \, l \geq 1\,
\end{equation}
then equation \eqref{fi.2} is satisfied (as it will be proved below).
Here $\alpha_1$ is the potential form of $\hat\omega_1$ and is given by
$\alpha_1 := \alpha_0 + \alpha_W$, 
where $\alpha_W$ is defined in Corollary
\ref{cor.e}.  However,  \eqref{fi.3} is essentially a system of equations for the potential of a 1-form on a torus, so there is a solvability condition. In Lemma \ref{moser} below we will prove that the system \eqref{fi.3} has a solution if the following conditions are satisfied:
\begin{align}
\label{comp.1}
&d(\alpha_1-\alpha_0)\arrowvert_{T\Fz}=0 \ ,
\\
\label{comp.2}
&M(\alpha_1-\alpha_0)\arrowvert_{T\Fz}=0\ .
\end{align}

In order to show that these two conditions are fulfilled, we need a preliminary result.
First, for
$v\in \spazior{w}$ fixed, define the symplectic orthogonal of $T_v\Fz$ with respect to the form $\omega^t := \omega_0 + t(\hat\omega_1 - \omega_0) $ by
\begin{equation}
\label{fol.1}
(T_v\Fz)^{\angle_t}:=\left\{h \in\spazior{w}:
\omega^t(v)(u,h)=0 \; \forall u \in T_v\Fz\right\}.
\end{equation}

\begin{lem}\label{tv}For  $v \in \Br^{{w^1}}(\mu^5 \rho)$, one has $T_v\Fz = (T_v\Fz)^{\angle_t}$.
\end{lem}
\begin{proof} 
First of all we have that, since for any couple of functions $F,G$ and
any change of coordinates $\Phi$, one has
$$
\left\{ F\circ \Phi,G\circ\Phi\right\}_{\Phi^*\omega_0}=\left\{F,G
\right\}_{\omega_0} \circ \Phi\ ,
$$
it follows that
$$
\left\{I_l,I_m\right\}_{\omega_1}=\left\{\left|\Psi_l\right|^2,
\left|\Psi_m\right|^2\right\} _{\omega_0}=0\ ,\quad \forall l,m \geq 1
$$
and
$$
\left\{ I_l,I_m\right\}_{\widehat
  \omega_1}\circ\hat\varphi^{-1}=\left\{I_l\circ \hat\varphi^{-1} , I_m\circ
\hat\varphi^{-1} \right\}_{\omega_1}
$$
but, by the property of invariance with respect to rotations of
$\widehat \varphi$ (and therefore of $\widehat \varphi^{-1}$),
$I_j\circ\widehat \varphi^{-1}$ is a function of $\left\{
I_l\right\}_{l\geq 1}$ only, and therefore the above quantity vanishes
and one has $\forall l,m$
\begin{equation}
\label{eq:IlIm}
0=\left\{I_l(v), I_m(v)\right\}_{\hat\omega_1}=\left\langle   \nabla I_l(v),J_{\hat\omega_1}(v) \nabla I_m(v) \right\rangle = \left\langle    v_l \be_l,  J_{\hat\omega_1}(v) v_m \be_m \right\rangle \quad \forall \, l,m \geq 1.
\end{equation}
Define $\Sigma_v:= \mbox{span}\left\{v_l \be_l,\; l\geq 1\right\}$.
The identities \eqref{eq:IlIm} imply that
$J_{\hat\omega_1}(v)(\Sigma_v)\subseteq \Sigma_v^\bot \equiv
\im \Sigma_v$. By Corollary \ref{cor.e} $(i)$, $E_{\hat\omega_1}(v)$ is
an isomorphism for $v \in \Br^{{w^1}}(\mu^5\rho)$, so the same is true
for its inverse $J_{\hat\omega_1}(v)$.  Hence
$J_{\hat\omega_1}(v)(\Sigma_v)= \im \Sigma_v $ and
$\Sigma_v=E_{\hat\omega_1}(v)( \im \Sigma_v) $ and
\begin{equation}
\label{omega1X} \hat\omega_1(X^0_{I_l},X^0_{I_m})=\left\langle E_{\hat\omega_1}(v) (\im v_l \be_l),   \im v_m \be_m \right\rangle =0, \quad \forall \, l , m \geq 1.
\end{equation}
  Since $\omega^t$ is a linear combination of $\omega_0$ and
$\hat\omega_1$, the previous formula implies that
$\omega^t(v)(X^0_{I_l},X^0_{I_m})=0$ for every $t \in [0,1]$ and $v \in \Br^{{w^1}}(\mu^5 \rho)$, hence $T_v\Fz \subseteq
(T_v\Fz)^{\angle_t}$. Now assume by contradiction that the inclusion is strict: then there
exists  $u \in
(T_v\Fz)^{\angle_t}, \; \norm{u}=1,$ such that  $u \notin T_v\Fz$. Decompose
$u=u_\top + u_\bot$ with $u_\top \in T_v\Fz$ and
$u_\bot \in (T_v\Fz)^\bot$. Due to the bilinearity of
$\omega(v)^t$, we can always assume that $u\equiv u_\bot$.
Then for every $l \geq 1$
$$dI_l(v)(-\im u)=\< \nabla I_l(v),-\im u \>=\< -\im X^0_{I_l}(v),-\im u \> = \<
X^0_{I_l}(v),u \>=0 \qquad \forall \, l \geq 1$$
 since $X^0_{I_l}(v) \in T_v\Fz$. Hence $\im u
\in T_v\Fz$ and therefore $\omega^t(v)(-\im u,u)=0$.  Furthermore it holds that 
$$\omega^t(0)(\im u,u)=\omega_0(-\im u,u)=\<\im^2 u,u\>=-1.$$
It follows that for $v \in \Br^{{w^1}}(\mu^5 \rho)$ one has $\norm{tM\Upsilon_{\hat{\omega}_1}(v)}_{\L(\spazior{w^1}, \spazior{w^1})} \leq 1/2 $, thus
 $\omega^t(v)(\im u,u)=-1 + \langle tM\Upsilon_{\hat{\omega}_1}(v) \im u, u \rangle <0$,  leading to a contradiction.
\end{proof}

We can now prove the following lemma: 
 
\begin{lem} The solvability conditions   \eqref{comp.1}, \eqref{comp.2}
are fulfilled. 
 \end{lem}
\begin{proof}
Condition \eqref{comp.1} follows by equation 
\eqref{omega1X}, since
$$d(\alpha_1-\alpha_0)(X^0_{I_l},X^0_{I_m})=\hat\omega_1(X^0_{I_l},X^0_{I_m})-
\omega_0(X^0_{I_l},X^0_{I_m})=0, \qquad \forall l, m \geq 1.$$ 
We analyze now \eqref{comp.2}. We claim that in order to fulfill this condition, one must have that $\hat\omega_1$ satisfies 
 $M\hat\omega_1=\omega_0$, which holds by  Lemma \ref{map.phi1} $(ii)$. Indeed, since
$$0=M\hat\omega_1-\omega_0=M(\hat\omega_1-\omega_0)=Md(\alpha_1-\alpha_0)=dM(\alpha_1-\alpha_0),$$
 there exists a function $g$ such that $M(\alpha_1-\alpha_0)=dg$. But 
$Mdg=M(M(\alpha_1-\alpha_0))=M(\alpha_1-\alpha_0)=dg,$
therefore $g=Mg$, so $g$ is invariant by rotations. Hence $0=\derz{t}g(\phi_l^t)=dg(X^0_{I_l})=M(\alpha_1-\alpha_0)(X^0_{I_l}),$ $\, \forall l \geq 1,$ thus also 
\eqref{comp.2} is satisfied.
\end{proof}
We show now that the system  \eqref{fi.3} can be solved and its solution has good analytic properties:

\begin{lem}
\label{moser}
(Moser) If conditions \eqref{comp.1} and \eqref{comp.2} are fulfilled, then equation \eqref{fi.3} has a solution $f$. Moreover, denoting $h_j:=(\alpha_1-\alpha_0)(X^0_{I_j})$, the solution $f$  is given by the explicit formula
\begin{equation}
f(v)=\sum_{j=1}^\infty f_j(v) ,  \qquad f_j(v)=M_1\cdots
M_{j-1}L_jh_j
\label{f.formula}
\end{equation}
where
$$L_jg=\frac{1}{2\pi}\int_0^{2\pi}{tg(\phi^t_j)dt}\ .$$
Finally $f \in \Nc_{\mu^7\rho}(\spazior{w^1}, \C)$, $\nabla f\in \Nc_{\mu^7\rho}(\spazior{w^1}, \spazior{w^2})$ and
\begin{equation}
\label{sti.nabla.f}
\mmod{\und{f}}_{\mu^7\rho} \leq 2^{10}  \epsilon_1 \mu^7 \rho, \quad 
\mmod{\mod{\nabla f}}_{\mu^7\rho}\leq 2^{11} \epsilon_1\ .
\end{equation}
\end{lem}
\begin{proof}
Denote by $\theta_j$ the time along the flow generated by $X^0_{I_j}$,
then one has
$dg(X^0_{I_j}) = \frac{\partial g}{\partial \theta_j}\ ,$
so that the
equations to be solved take the form 
\begin{equation}
\label{eqf}
\frac{\partial f}{\partial \theta_j}=h_j, \qquad \forall j\geq 1.
\end{equation}
Clearly $\frac{\partial}{\partial\theta_j}M_jh_j=0 $, and by \eqref{comp.1} it follows that 
$$ 
\frac{\partial}{\partial\theta_l}M_jh_j = M_j \frac{\partial h_j}{\partial\theta_l} 
=M_j \frac{\partial h_l}{\partial\theta_j} = \frac{\partial}{\partial\theta_j}M_jh_l =0, \qquad \forall l, j \geq 1,$$
which shows that $M_jh_j$ is independent of all the $\theta$'s, thus $M_j h_j = M h_j$. Furthermore,  by \eqref{comp.2} one has $ M h_j = 0, \,$ $\forall \, j \geq 1$.
Now, using that $\frac{\partial}{\partial\theta_j}L_j g = g - M_j g $,  one verifies that  $f_j$ defined in \eqref{f.formula} satisfies
\begin{equation*}
\frac{\partial f_j}{\partial \theta_l}=\left\{
\begin{matrix}
0 & \text{if}\ l<j
\\
M_1\cdots M_{j-1}h_j & \text{if}\ l=j
\\
M_1 \cdots M_{j-1}h_l-M_1 \cdots M_jh_l  & \text{if}\ l>j 
\end{matrix}
   \right.
\end{equation*}
where, for $j=1$, we defined $M_1 \cdots M_{j-1} h_l = h_l$. Thus the series $f(v):=\sum_{j \geq 1} f_j(v)$, if convergent,  satisfies \eqref{eqf}. \\
We prove now the convergence of the series for $f$ and $\nabla f$.  First we define, for $\theta \in \Tc $,
$$\Theta_j^\theta:=\phi_1^{\theta_1} \cdots \phi_j^{\theta_j} \qquad \forall\, j \geq 1 \ ,$$
 then by \eqref{f.formula} one has 
\begin{align}
\label{f.g1}
f_j(v)=\int_{\Tc^j}\theta_jh_j(\Theta_j^\theta v)\;d\theta^j \ , 
\\
\label{f.g2}
\nabla f_j(v)=\int_{\Tc^j}\Theta_j^{-\theta} \theta_j\nabla
h_j(\Theta_j^\theta v)\;d\theta^j\ ,
\end{align}
where $\Tc^j$ is the $j$-dimensional torus and $d\theta^j = \frac{d\theta_1}{2\pi}\cdots \frac{d\theta_j}{2\pi}$.
Now, using that 
$$h_j(v)=\langle W(v), X^0_{I_j}(v) \rangle = Re(\im W_j(v)\bar v_j) \qquad \forall \, j \geq 1
$$ 
 one gets that
$
\und{f_j}(|v|) \leq 2\pi\, \underline{h_j}(|v|) \leq 2\pi\, \underline{W_j}(|v|) |v_j| 
$,
therefore
$
\und{f}(|v|) \leq \sum_{j=1}^\infty \und{f_j}(|v|) \leq 2\pi \, \norm{\und{W}(|v|)}_{w^1}  \norm{v}_{w^1}$ and it follows that $\mmod{\und{f}}_{\mu^7\rho} \leq 2\pi\, \mmod{\und{W}}_{\mu^7\rho} \mu^7 \rho.$
This  proves the convergence of the series defining  $f$.\\
Consider now the gradient of $h_j$, whose $k^{th}$ component is given by
$$
\left[\nabla h_j(v)\right]_k=Re\left(\im \frac{\partial W_j(v)}{\partial
  v_k}\bar v_j \right)+ \delta_{j,k} \, Re\, (\im W_j(v)) \ .
$$
Inserting the formula displayed above in \eqref{f.g2} we get that $\nabla f_j$ is the sum of two
terms. We begin by estimating the second one, which we denote by $(\nabla
f_j)^{(2)}$. 
The $k^{th}$ component of 
$(\nabla f)^{(2)} := \sum_j (\nabla f_j )^{(2)}$ is given by
\begin{equation}
\label{fg.3}
\left[\left( \nabla f(v)\right)^{(2)} \right]_k=\left[\sum_j(\nabla f_j(v))^{(2)}\right]_k = \int_{\Tc^k}\Theta_k^{-\theta}
\theta_k \, Re\, (\im W_k(\Theta_k^\theta v))\; d\theta^k\ ,
\end{equation}
thus, for any $v \in \Br^{{w^1}}(\mu^7\rho)$ one has
$\left[(\und{\nabla f}(|v|))^{(2)}\right]_k  \leq 2\pi\,
\mod{W_k}(|v|)\ ,$
and therefore  
$$
\mmod{\mod{\left(\nabla f\right)^{(2)}}}_{\mu^7 \rho}\leq 2\pi\, \mmod{ \mod{W}}_{\mu^7\rho}\leq \pi 2^8 \epsilon_1.
$$
We come to the other term, which we denote by $\left(\nabla f_j \right)^{(1)}$. Its $k^{th}$ component is given by
\begin{equation}
\label{fg.4}
 \left[ (\nabla f_j(v))^{(1)}\right]_k=\int_{\Tc^j}\Theta_j^{-\theta}
\theta_j Re\left(\im \frac{\partial W_j}{\partial v_k}(\Theta_j^\theta
v) \overline{\phi^{\theta_j}_jv_j}\right)   d\theta\ .
\end{equation}
Then 
$\und{\nabla f_j}(|v|) \leq 2\pi \und{\frac{\partial W_j}{\partial v_k}}(|v|) |v_j| = 2\pi  [\und{ d W}(|v|)]^j_k|v_j|.$\\
 It follows that the $k^{th}$ component of the function $(\nabla f)^{(1)} := \sum_j (\nabla f_j)^{(1)}$ satisfies
$$\left[ (\und{\nabla f}(|v|))^{(1)}\right]_k \leq \left[
   \sum_{j}(\und{\nabla f_j}(|v|))^{(1)} \right]_k \leq 2 \pi \,
 \sum_j [\und{d W}(|v|)]^j_k|v_j| \ .$$ Therefore $ \mmod{\und{
     (\nabla f)^{(1)}}}_{\mu^7\rho}\leq 2\pi\, \norm{W}_{\A_{w^1,
     \mu^7\rho}^{w^2}} \leq \pi 2^8 \epsilon_1$. This is the step at
 which the control of the norm of the modulus $\und{dW^*}$ of $dW^*$ is
 needed.  Thus the claimed estimate for $\nabla f$ follows.
\end{proof}

We can finally apply the Darboux procedure in order to construct an
analytic change of coordinates $\varphi $ which satisfies \eqref{fi.1}
and \eqref{fi.2}.
\begin{lem}
\label{varphi.1}
There exists a map $\varphi: \Br^{{w^1}}(\mu^9\rho) \to \spazior{w^1}$ which satisfies 
\eqref{fi.1}. Moreover $\varphi - \uno \in \Nc_{\mu^9\rho}(\spazior{w^1}, \, \spazior{w^2})$, $\varphi-\uno = O(v^2)$ and  
\begin{equation}
\label{eq.varphi.1}
\mmod{\underline{\varphi-\uno}}_{\mu^9\rho}\leq
2^{14} \epsilon_1\ .
\end{equation}
\end{lem}
\begin{proof}
As anticipated just after Corollary \ref{cor.psi.1}, we apply the Darboux procedure with $\Omega_0 = \omega_0$, $\Omega_1 = \hat{\omega}_1$ and $f$ solution of \eqref{fi.3}. Then equation \eqref{darboux.eq} takes the form
\begin{equation}
\label{Yt}
Y^t=( -\im +t\Upsilon_{\hat{\omega}_1})^{-1}(\nabla f-W),
\end{equation}
where $\Upsilon_{\hat{\omega}_1}$ and $W$ are defined in Corollary \ref{cor.e}.  
By Lemma \ref{moser} and Corollary \ref{cor.e}, the vector field $Y^t$ is of class $\Nc_{\mu^8\rho}(\spazior{w^1}, \spazior{w^2})$ 
 and
$$
\sup_{t \in [0,1]}\mmod{\und{Y^t}}_{\mu^8\rho} < 2(2^{11}\epsilon_1 + 2^7\epsilon_1) < 2^{13} \epsilon_1.
$$
Thus  $Y^t$ generates a flow $\varphi^t: \Br^{{w^1}}(\mu^9\rho)\to \spazior{w^1}$, defined for every $t \in [0,1]$,  which satisfies (cf. Lemma \ref{flussoinA})
$$
\mmod{\und{\varphi^t-\uno}}_{\mu^9\rho}\leq
2^{14} \epsilon_1, \quad \forall t \in [0,1]\ .
$$
Thus the map $\varphi: = \left.\varphi^t\right|_{t=1}$ exists and satisfies the claimed properties.
\end{proof}

We prove now that the map $\varphi$ of Lemma \ref{varphi.1} satisfies also equation \eqref{fi.2}. 

\begin{lem}\label{lem.f} Let  $f$ be as in \eqref{f.formula} and $\varphi^t$ be the flow map of the vector field $Y^t$ defined in \eqref{Yt}. Then $\forall \, l \geq 1$ one has $I_l(\varphi^t(v))=I_l(v)$, for each $t \in [0,1]$.
\end{lem}
\begin{proof} The following chain of equivalences follows from Lemma \ref{tv} and the Darboux equation \eqref{darboux.eq}:
\begin{align*}
I_l(\varphi^t(v))=I_l(v) & \iff 0=\frac{d}{dt}I_l(\varphi^t(v))=dI_l(Y^t(v)) \iff
Y^t(v) \in T_v\Fz \\
& \iff  Y^t(v) \in (T_v\Fz)^{\angle_t} 
\iff \left( \omega^t_v(Y^t(v),X^0_{I_l}(v))=0\ ,\ \forall l \geq 1\right) \\
& \iff \alpha_1(X^0_{I_l}) - \alpha_0(X^0_{I_l}) = df(X^0_{I_l}) \, \quad \forall l \geq 1 \ .
\end{align*} 
In turn the last property follows since $f$ is a solution of \eqref{fi.3}.
\end{proof}

\vspace{1em}
We can finally prove the quantitative version of the Kuksin-Perelman Theorem.
\vspace{1em}\\
\noindent{\em Proof of Theorem \ref{KP}.} 
Consider the map $\varphi$ of Lemma \ref{varphi.1}. Since $d\varphi(0) = \uno$, $\varphi$ is invertible in $\Br^{{w^1}}(\mu^{10}\rho)$ and $\varphi^{-1}= \uno + g_1$ with $g_1 \in \Nc_{\mu^{10}\rho}(\spazior{w^1}, \spazior{w^2})$ and $ \mmod{\und{g_1}}_{\mu^{10} \rho} \leq 2\mmod{\und{\varphi- \uno}}_{\mu^9 \rho} \leq 2^{15}\epsilon_1 $ (cf. Lemma \ref{Ginverso}).
Define now
$$\widetilde{\Psi}:= \varphi^{-1} \circ \check \Psi.$$
It's easy to check  that $\widetilde{\Psi}^* \omega_0 = \omega_0$, thus proving that $\widetilde{\Psi}$ is symplectic.
By equation \eqref{fi.2} one has   $I_l(\widetilde{\Psi}(v)) = I_l(\check{\Psi}(v))$ for every $l \geq 1$, therefore $\widetilde{\Psi}$ and $\check{\Psi}$ define the same foliation, which coincides also with the foliation defined by $\Psi$, c.f.   Corollary \ref{cor.psi.1}. Similarly one proves that the functionals  $\left\{ \frac{1}{2}\mmod{\widetilde{\Psi_j}(v)} \right\}_{j \geq 1}$ pairwise commute with respect to the symplectic form $\omega_0$. We have thus proved item $i) - iii)$ of Theorem \ref{KP}. 
\newline
We prove now item $iv)$. Clearly $d\widetilde{\Psi}(0) = \uno$, and
$\widetilde{\Psi}^0 := \widetilde{\Psi} - \uno = \check{\Psi}^0 + g_1\circ(\uno + \check{\Psi}^0)$ is of class $\Nc_{\mu^{11}\rho}(\spazior{w^1}, \spazior{w^2})$. Moreover, by Remark \ref{rem:norm.in.A} and Corollary \ref{cor.psi.1} $(i)$, one has $\mmod{\und{\check{\Psi}^0}}_{\mu^{11}\rho}
\leq 2\mu^{6} \mmod{\und{\check{\Psi}^0}}_{\mu^{8}\rho} \leq \mu^6 2^9 \epsilon_1 \leq \mu^{11} \rho$ by condition \eqref{th.1}. Thus $\mmod{\und{\uno + \check{\Psi}^0}}_{\mu^{11}\rho} \leq \mu^{10}\rho$ and by Lemma \ref{FGinN}
\begin{align*}
\mmod{\und{\widetilde{\Psi}_0}}_{\mu^{11}\rho} & \leq \mmod{\und{\check{\Psi}^0}}_{\mu^{11}\rho} + \mmod{\und{g_1\circ (\uno + \check{\Psi}^0)}}_{\mu^{11}\rho} \leq  \mmod{\und{\check{\Psi}^0}}_{\mu^{11}\rho} + \mmod{\und{g_1}}_{\mu^{10}\rho} \leq 2^8\epsilon_1 + 2^{15}\epsilon_1 \leq 2^{16}\epsilon_1 .
\end{align*}
 We are  left to prove that $\widetilde{\Psi}^0 \in \A_{w^1, \mu^{12} \rho}^{w^2}$.
 Since $\widetilde{\Psi}^* \omega_0 = \omega_0$, one has $d\widetilde{\Psi}(v)^* (-\im) \, \widetilde{\Psi}(v) = -\im$, from which it follows that  $\widetilde{\Psi}^0$ satisfies 
$$
d\widetilde{\Psi}^0(v)^* = \im \, d\widetilde{\Psi}^0(v) \left(\uno+ d\widetilde{\Psi}^0(v) \right)^{-1} \im
$$
and therefore $\widetilde{\Psi}^0 \in \A_{w^1, \mu^{12} \rho}^{w^2}$ with $\norm{\widetilde{\Psi}^0}_{\A_{w^1, \mu^{12} \rho}^{w^2}} < 2^{17} \epsilon_1$.
\qed


\section{Toda lattice}

\subsection{Proof of Theorem \ref{main} and  Corollary
  \ref{corkuksinperelman}.} 

We consider the Toda lattice with $N$ particles and periodic boundary
conditions on the positions $q$ and momenta $p$: $q_{j+N}= q_j\, ,$
$p_{j+N}= p_j$, $ \forall \,j \in \mathbb{Z}$.  As anticipated in
Section \ref{section1}, we restrict to the invariant subspace
characterized by \eqref{media}.  The phase space of the system is
$\spazio{\reg}$, where $s\geq 0, $ $\sigma \geq 0$ and it is defined
in terms of the linear, complex, Birkhoff variables $(\xi, \eta)$
(defined in \eqref{xi_variable}). We endow the phase space with the
symplectic form \footnote{so that the Hamilton equations become
\begin{equation}
\label{Omega0} 
\dot{\xi}_k = \im \frac{\partial H}{\partial \eta_k}, \qquad \dot{\eta}_k = -\im \frac{\partial H}{\partial  \xi_k} \ , 
\end{equation}
} $ \Omega_0 = -\im \sum_{k=1}^{N-1} d\xi_k \wedge d\eta_k$.

We will denote by $\spazior{\reg}$
the real subspace of $\spazio{\reg}$ in which $\eta_k = \bar \xi_k\,$
$\forall 1 \leq k \leq N-1$, endowed with the norm \eqref{nor.bir},
and by $\Br^{{\reg}}(\rho)$ the ball in $\spazior{\reg}$ with
center $0$ and radius $\rho >0$. 
The main step of the proof of Theorem \ref{main} is the construction
of the functions $\{ \Psi_j \}_{1 \leq j \leq N-1}$. This is based on
a detailed analysis of the spectrum of the Jacobi matrix appearing in
the Lax pair representation of the Toda lattice. So we start by
recalling the elements of the theory needed for our development.
Introduce the translated Flaschka coordinates \cite{flaschka} by
\begin{equation}
(b, a) = \Theta(p, q), \qquad (b_j, \; a_j) := (-p_j,\, e^{\frac{1}{2}(q_j-q_{j+1})}-1).
\label{bavariable}
\end{equation}
The translation of the $a$ variables by $1$ is useful in order to keep the equilibrium point at $(b,a) = (0,0)$. Recall that the variables $b, a$ are constrained by the conditions 
$$ \sum_{j=0}^{N-1}{b_j}=0, \, \prod_{j=0}^{N-1}{(1+a_j)}=1 \ .$$
Introduce Fourier variables $(\hat{b}, \hat{a}) $ for the Flaschka coordinates by \eqref{fou}. In these variables 
\begin{equation}
\label{Ek.ba}
E_k = \frac{|\hat{b}_k|^2 + 4|\hat{a}_k|^2}{2} + O(\hat{a}^3), \qquad 1 \leq k \leq N-1 \ .
\end{equation}
The  Jacobi matrix  whose spectrum forms a complete set of integrals of motions for the Toda lattice is given by \cite{moerbeke}
\begin{equation} \label{jacobi}
L(b,a) := \left( \begin{array}{ccccc}
b_{0} & 1+a_{0} & 0 & \ldots &  1+a_{N-1} \\
1+a_{0} & b_{1} & 1+a_{1} & \ddots & \vdots \\
0 & 1+a_{1} & b_{2} & \ddots & 0 \\
\vdots & \ddots & \ddots & \ddots & 1+a_{N-2} \\
 1+a_{N-1} & \ldots & 0 & 1+a_{N-2} & b_{N-1} \\
\end{array} \right) .
\end{equation}
It is useful to  double the size of $L(b,a)$, redefining
\small
\begin{equation}
L_{b,a} := \left( 
\begin{array}{cccc|cccc}
b_{0} & 1+a_{0} & \ldots & 0 & 0 & \ldots & 0 & 1+a_{N-1} \\
1+a_{0} & b_{1} & \ddots & \vdots & 0 & \ldots & & 0 \\
\vdots & \ddots & \ddots & 1+a_{N-2} & \vdots & & & \vdots \\
0 & \ddots & \;\; 1+a_{N-2} & b_{N-1} & 1+a_{N-1} & \ldots & 0 & 0 \\
\hline
0 & \ldots & 0 & 1+a_{N-1} & b_{0} & 1+a_{0} & \ldots & 0 \\
0 & \ldots & & 0 & 1+a_{0} & b_{1} & \ddots & \vdots \\
\vdots & & & \vdots & \vdots & \ddots & \ddots & 1+a_{N-2} \\
1+a_{N-1} & \ldots & 0 & 0 & 0 & \ddots & \;\; 1+a_{N-2} & b_{N-1} \\
\end{array} 
\right).
\label{Lperturbata}
\end{equation}
\normalsize
 Consider the eigenvalues of $L_{b,a}$ and  order them in the non-decreasing sequence 
$$\lambda_0(b,a)< \lambda_1(b,a) \leq \lambda_2(b,a)<\ldots<\lambda_{2N-3}(b,a) \leq \lambda_{2N-2}(b,a)< \lambda_{2N-1}(b,a)$$
where one has that where the sign $\leq$ appears equality is possible, while it is impossible in the correspondence of a sign $<$. Define the quantities
\begin{equation}
\gamma_j(b,a):=\lambda_{2j}(b,a)-\lambda_{2j-1}(b,a), \qquad 1\leq j \leq N-1;
\label{gap}
\end{equation}
$\gamma_j(b,a)$ is called $j^{th}$ \textit{spectral gap}. The quantities $\{\gamma_j^2\}_{1 \leq j \leq N-1}$ form a complete set of commuting integrals of motions, which are regular also at $(b,a)=(0,0)$. Furthermore one has $H(b,a) = H(\gamma_1^2(b,a), \ldots, \gamma_{N-1}^2(b,a))$ \cite{kappelerfibrationtoda}. 
 A spectral gap  is said to be  \textit{closed} if $\gamma_j(b,a)=0$.

The following Theorem \ref{main2} ensures that the assumptions of
Theorem \ref{KP} are fulfilled by the Toda lattice.
\begin{teo}
 	\label{main2}
There exists $\epsilon_* >0$, independent of $N$, and an  analytic  map
\begin{equation}
\label{mappa.Psi} 
\Psi: \left(B^{\reg}\left(\frac{\epsilon_*}{N^2} \right), \Omega_0 \right)\to \spazio{\reg}, \quad (\xi, \eta) \mapsto (\phi(\xi, \eta), \psi(\xi, \eta))
\end{equation}
such that:
\begin{enumerate}
\item[$(\Psi 1)$] $\Psi $ is  real for real sequences, namely 
  $\overline{\phi_k(\xi, \bar \xi)}=\psi_{k}(\xi, \bar \xi)\,$ $\forall k$.
    
\item[$(\Psi 2)$] For every $1 \leq j \leq N-1$, and for
  $(\phi,\psi) \in B^{\reg}\left(\frac{\epsilon_*}{N^2} \right)\cap \spazior{\reg}$, one has
  $$\gamma_{j}^2 = \tfrac{2}{N}\om{j} \left|\psi_j\right|^2 = \tfrac{2}{N}\om{k} \left|\varphi_j\right|^2 \ .$$
  \item[$(\Psi 3)$] $\Psi(0,0)=(0,0)$ and $d\Psi(0,0)= \uno$.
  \item[$(\Psi 4)$] There exist  constants $C_1, C_2>0$, independent of $N$, such that for every $0 < \epsilon \leq \epsilon_*$, the map $\Psi^0:= \Psi- \uno \in  \Nc_{\epsilon/N^2}\left(\spazio{\reg}, \spazio{s+1, \sigma} \right)$ and     $[d\Psi^{0}]^* \in \Nc_{\epsilon/N^2}\left(\spazio{\reg},\, \L(\spazio{\reg}, \spazio{s+1, \sigma}) \right)$. Furthermore one has
  \begin{equation}
\mmod{\und{\Psi^0}}_{\epsilon/N^2}\leq C_1 \frac{\epsilon^2}{N^2}; \qquad \mmod{\und{[d\Psi^{0}]^*}}_{\epsilon/N^2} \leq C_2 \epsilon \ .
\end{equation}
 	\end{enumerate}
\end{teo}
The main point is $(\Psi 4)$, in which the estimates of the domain of
definition of the map $\Psi$  holds uniformly in the limit  $N \to \infty$.

\vspace{1em} 

We show now how Theorem \ref{main} follows from
Kuksin-Perelman Theorem \ref{KP}.
\vspace{0.5em}

\noindent\textit{Proof of Theorem \ref{main}.}
Introduce the weights $w^1:=\{N^{3/2} [k]_N^{s} e^{\sigma[k]_N} \om{k}^{1/2} \}_{k=1}^{N-1}$
and  $w^2:=\{N^{3/2} [k]_N^{s+1} e^{\sigma[k]_N} \om{k}^{1/2} \}_{k=1}^{N-1}$ and consider the map $\Psi$ of Theorem \ref{main2} as a map from $\spazio{w^1}$ in itself. 
 Since  for any $(\xi, \eta) \in \spazio{w^1}$ one has that
\begin{equation}
\label{norm.scaled}
\norm{(\xi, \eta)}_{\spazio{w^1}} \equiv N^2 \norm{(\xi, \eta)}_{\spazio\reg} \ ,
\end{equation}
it follows  by
scaling that there exists a constant $C_3>0$, independent of $N$,  such that
$$
\norm{\Psi^0}_{\A_{w^1, \rho}^{w^{2}}} \leq C_3 \rho^2 \ . 
$$ Thus, for any $\rho \leq \rho_* \equiv \min
\left(\frac{2^{-34}}{C_3 }, \epsilon_* \right)$, $\Psi$ satisfies
condition \eqref{th.1}. Thus we can apply Theorem \ref{KP} to the map
$\Psi$, getting the existence of a symplectic real analytic map $\widetilde{\Psi}$
defined on $B^{{w^1}}(a \rho_*)$ which satisfies $i)-iv)$ of
Theorem \ref{KP}.\\ 
By Lemma \ref{FGinA} the map $\widetilde{\Psi}$
is invertible in $B^{{w^1}}(\mu a \rho_*)$ and its inverse
$\Phi$ satisfies $\Phi = \uno + \Phi^0$
with $\Phi^0 \in \A_{w^1, \mu a\rho_*}^{w^2}$. To get the statement of
the theorem simply reexpress the map $\Phi$ in terms of real variables
$(x,y)$, $(X,Y)$
and denote such a map by $\Phi_N$.
%
\qed 

\begin{rem}
\label{PhiN.trasp.est}
By the proof of Theorem \ref{main} above one deduces the estimate
\begin{equation}
 \sup_{\norm{(\phi, \psi)}_{\spazio{\reg}} \leq R_\reg/N^2 }{\norm{d\Phi^0(\phi, \psi)^*}_{\L(\spazio{\reg}, \spazio{s+1, \sigma })}}  \leq C_\reg R_\reg \ ,
\end{equation}
for some $C_\reg >0$, independent of $N$.
\end{rem}

The rest of this subsection is devoted to the proof of Theorem
\ref{main2}. 

In the following it will be convenient to consider the variables
$(b,a)$ defined in \eqref{bavariable} dropping the conditions
$\sum_{j=0}^{N-1} b_j =0$ and $\prod_{j=0}^{N-1} (1+a_j) =1$. Equation
\eqref{Ek.ba} suggests to introduce on the variables $b,a$ the norm
\begin{equation}
\label{fouriernorm2}
\norm{(b,a)}_{\Cs^\reg}^2 := \frac{1}{2N} \sum_{k =0 }^{N-1}{\max(1,[k]_N^{2s}) e^{2\sigma [k]_N}\left( 
|\hat{b}_k|^2 + 4| \hat{a}_k|^2\right)}
\end{equation}
and to define the space
\begin{equation}
\label{Cspace}
\Csr^\reg := \left\lbrace (b, a) \in \R^N \times \R^N\ :\, \norm{(b,a)}_{\Cs^\reg} < \infty \right\rbrace \ .
\end{equation}
 We will write $\Cs^\reg$ for the complexification of $\Csr^\reg$.
\vspace{0.5em}\\ In the following we will consider
normally analytic map
between the spaces  $\spazio\reg$ and $\Cs^\reg$. We need to specify the basis of $\Cs^\reg$ that we will use to verify the property of being normally analytic. While it is
quite hard to verify this property when the
basis is general, it turns out that it is quite easy to verify it
using the basis of complex exponentials defined in \eqref{fou}. Indeed
the norm \eqref{fouriernorm2} is given in term of
the Fourier variables. 
For the same reason, it will be convenient to express a map from $\Cs^\reg$ to $\spazio{\reg}$ as a function of the Fourier variables $\hat b, \, \hat a$.
%
%
\vspace{0.5em}\\
We prove now some analytic properties of  the map $\Theta$ defined in \eqref{bavariable}.  In the following we will denote by $\Theta_\Xi$ the map $\Theta$ expressed in the $(\xi, \eta)$ variables.
\begin{prop}
\label{prop:Thata_na} The map $\Theta_\Xi$  satisfies the following properties:
\begin{enumerate}
\item[$(\Theta 1)$] $\Theta_\Xi(0,0) = (0,0)$. Furthermore let $d\Theta_\Xi(0,0)$ be the linearization of $\Theta_\Xi$ at $(\xi, \eta) = (0, 0)$. Then $(B, A ) = d{\Theta_\Xi}(0,0)[(\xi, \eta)]$ iff
\begin{equation}
\label{Theta.lin}
\begin{aligned}
& \widehat{B}_0 = 0, \qquad \widehat{B}_k = -\left(\tfrac{1}{2} \om{k} \right)^{1/2} (\xi_k + \eta_{N-k}), \qquad 1\leq k \leq N-1 \ , \\
& \widehat{A}_0 = 0, \qquad \widehat{A}_k = -\im \varpi_k \left(2 \om{k} \right)^{-1/2} (\xi_k - \eta_{N-k}), \qquad 1\leq k \leq N-1.
\end{aligned}
\end{equation}
where $\varpi_k := (1-e^{-2\im \pi k/N})/2 $,  $\forall \, 1 \leq k \leq N-1$.

Moreover for any $s \geq 0$, $\sigma \geq 0$ there exist constants $C_{\Theta_1}, C_{\Theta_2} >0$, independent of $N$, such that
\begin{equation}
\label{dtheta.est}
\norm{\und{d\Theta_\Xi}
(0,0)}_{\L(\spazio{\reg},\, \Cs^{\reg} )} \leq C_{\Theta_1}, \qquad
 \norm{\und{d\Theta_\Xi}
(0,0)^{*}}_{\L(\Cs^{s+2, \sigma },  \, \spazio{s+1, \sigma})} \leq \frac{C_{\Theta_2}}{N} \ .
\end{equation}

\item[$(\Theta 2)$] Let $\Theta_\Xi^0 := \Theta_\Xi - d\Theta_{\Xi}(0,0)$. For any $s \geq 0, \, \sigma \geq 0 $, there exist constants $C_{\Theta_3}, C_{\Theta_4}, \epsilon_* >0$, independent of $N$, such that the map  $\Theta^0_\Xi \in \Nc_{\epsilon_* /N^2}(\spazio{\reg}, \Cs^{s+1, \sigma})$ and the map $[d\Theta_\Xi^0]^* \in \Nc_{\epsilon_* /N^2}(\spazio{\reg}, \, \L(\Cs^{s+2,\sigma}, \, \spazio{s+1,\sigma}))$, and 
\begin{equation}
\begin{aligned}
\label{eq:Thata_na}
&\mmod{\und{\Theta_\Xi^0}}_{\epsilon /N^2} \equiv 
\sup_{\norm{(\xi, \eta)}_{\spazio{\reg}} \leq \epsilon /N^2} \norm{\und{\Theta_\Xi^0}(\xi, \eta)}_{\Cs^{s+1, \sigma}} \leq \frac{C_{\Theta_3} \epsilon^2}{N^2}; \\ 
&\mmod{\und{[d\Theta^0_\Xi]^*}}_{\epsilon /N^2}  \equiv 
\sup_{\norm{(\xi, \eta)}_{\spazio{\reg}} \leq \epsilon /N^2}\norm{\und{d\Theta_\Xi^0}
(\xi , \eta)^{*}}_{\L(\Cs^{s+2,\sigma}, \, \spazio{s+1,\sigma})}  \leq  \frac{C_{\Theta_4} \epsilon}{N^2}.
\end{aligned}
\end{equation}
\end{enumerate}
\end{prop}
The proof of the proposition is postponed in Appendix \ref{Theta_na}.
Note that the estimates \eqref{dtheta.est} and \eqref{eq:Thata_na} imply that there exists a constant $C_{\Theta_5} >0$, independent of $N$, such that for any $\rho \leq \tfrac{\epsilon_*}{N^2}$ one has $\Theta_\Xi \in \Nc_\rho(\spazio{\reg}, \Cs^\reg) $ and 
\begin{equation}
\label{theta.sup5} 
\mmod{\und{\Theta_\Xi}}_{\rho} \leq C_{\Theta_5} \, \rho \ .
\end{equation}
\vspace{1em}\\
We start now the perturbative construction of the Birkhoff coordinates for the Toda lattice, which is based on the construction of the spectrum and of the eigenfunctions of $L_{b,a}$ (defined in \eqref{Lperturbata}) as a perturbation  of the free operator $L_0:= L_{b,a}\vert_{(b,a) = (0,0)}$. 
 More precisely we decompose  $L_{b,a}=L_0 + L_p$, where
\begin{equation}
	L_0= \left( 
	\begin{array}{ccccc}
	0 & 1 & 0 & \ldots & 1\\
	1 & 0 & 1 & \ddots & \vdots\\
	0   & 1 & 0 & \ldots & 0\\
  \vdots & \ddots & \ddots & \ddots & 1\\
	1 & \ldots & \ldots & 1 & 0
	\end{array}
	\right),	
	\ \ 
	L_p=\left( 
	\begin{array}{ccccc}
	b_{0} & a_{0} & 0 & \ldots & a_{N-1}\\
	a_{0} & b_{1} & a_{1} & \ddots & \vdots\\
	0   & a_{1} & b_{2} & \ldots & 0\\
  \vdots & \ddots & \ddots & \ddots & a_{N-2}\\
	a_{N-1} & \ldots & \ldots & a_{N-2} & b_{N-1}
	\end{array}
	\right)
	\label{L0Lp}
\end{equation}
and  following  the approach in \cite{kuksinperelman,kappelerfibrationtoda, kappelerKdV} we apply 
 Kato perturbation theory \cite{kato}.
  The next lemma characterizes completely the spectrum of $L_0$ as an operator on $\C^{2N}$:
\begin{lem}
Consider $L_0$ as an operator on $\C^{2N}$, then its eigenvalues and
normalized eigenvectors are:
$$
\begin{array}{ll}
\mbox{eigenvalues} & \mbox{eigenvectors}\\
\lambda_0^0=-2,  & f_{00}(k)=\frac{1}{\sqrt{2N}} \left(-1\right)^k \\
\lambda_{2j-1}^0=\lambda_{2j}^0
=-2\cos\left(\frac{j\pi}{N} \right),   & f_{2j-1,0}(k)=
\frac{1}{\sqrt{2N}}e^{- \im \rho_j k}, \ 
f_{2j,0}(k)=
\frac{1}{\sqrt{2N}}e^{ \im \rho_j k} \ , \qquad 1\leq j \leq N-1 \\
\lambda_{2N-1}^0=2,  & f_{2N-1,0}(k)=\frac{1}{\sqrt{2N}} \\
\end{array}
$$
where $0 \leq k \leq 2N-1$ and $\rho_j:=\left(1+\frac{j}{N}\right)\pi$.
In particular the gaps of $L_0$ are all closed.
\label{lem:spettroL0}
\end{lem}
The proof is an easy computation and can be found in \cite{kapphen1}. 
\begin{rem} For $0 \leq j, k \leq \lfloor N/2 \rfloor $ one has 
$\left|\lambda_{2j}^0 - \lambda_{2k}^0 \right|, \;  \left|\lambda_{2N-j}^0 - \lambda_{2N-k}^0 \right| \geq \frac{4 |j^2-k^2|}{N^2}.
$\\
In particular if $j\neq k$ then $\left|\lambda_{2j}^0 - \lambda_{2k}^0 \right|\geq 1/N^2.$
\label{lem:stimaAutovalori}
\end{rem}

We use now  Kato perturbation theory of operators  in order to introduce the main objects needed in the following and to give some preliminary estimates.

For $1 \leq j \leq N-1$ let $E_j(b,a)$ be the two-dimensional  subspace spanned by the eigenvectors corresponding to the eigenvalues $\lambda_{2j-1}(b,a)$  and $\lambda_{2j}(b,a)$ of $L_{b,a}$. 
Analogously, let $E_0(b,a)$ (respectively $E_N(b,a)$) be the one-dimensional subspace spanned by the eigenvector of $\lambda_{0}(b,a)$ (respectively $\lambda_{2N-1}(b,a)$).  Introduce the spectral projector on $E_j(b,a)$ defined by
\begin{equation}
P_j(b,a)=-\frac{1}{2\pi \im }\oint_{\Gamma_j}{\left(L_{b,a}-\lambda\right)^{-1}\diff{\lambda}}, \qquad 0\leq j \leq N
\label{P}
\end{equation}
where, for $1 \leq j \leq N-1$, $\Gamma_j$ is a closed path counter-clockwise oriented  in $\C$ which encloses the
eigenvalues $\lambda_{2j-1}(b,a)$ and $\lambda_{2j}(b,a)$ and does not contain any other
eigenvalue of $L_{b,a}$. Analogously, $\Gamma_0$ (respectively $\Gamma_N$) encloses the eigenvalue $\lambda_{0}(b,a)$ (respectively $\lambda_{2N-1}(b,a)$) and no other eigenvalue of $L_{b,a}$. $P_j(b,a)$ maps $\C^{2N}$ onto $E_j(b,a)$ and, as we will prove,  is well defined for $(b, a)$ small enough. 
$P_j(0,0)$ will be denoted by $P_{j0}$ and its range $E_j(0,0)$, which will be denoted by $E_{j0}$, is given by
$$ \mbox{Im}\;P_{j0}=E_{j0}, \ E_{j0}=\mbox{span}\left\langle f_{2j,0},f_{2j-1,0}\right\rangle.$$
Define also the  transformation operators
\begin{equation}
U_j(b,a)=\left(\uno - \left(P_j(b,a)-P_{j0}\right)^2\right)^{-1/2}P_j(b,a), \quad 1\leq j\leq N-1.
\label{defUj}
\end{equation}
$U_j$ has the property of mapping isometrically  $E_{j0}$ into the subspace $E_j(b,a)$ spanned by the perturbed eigenvectors \cite{kato}.
 Remark, however, that in general the image of an unperturbed eigenvector is \textit{not} an eigenvector itself.
We prove now some  properties of the just defined objects.
\begin{lem}
\label{lem:spec.prop1}
 There exist a constant $C_\reg >0$, independent of $N$, such that  the map $(b,a) \mapsto L_p(b,a)$ is analytic as a map from $\Cs^{\reg}$ to $\L\left(\C^{2N}\right)$. Moreover 
\begin{equation}
\label{L_p_esti}
\norm{L_p(b,a)}_{\L(\C^{2N})} \leq C_\reg \norm{(b,a)}_{\Cs^\reg}.
\end{equation}
\end{lem}
Then by Kato theory one has the corollary
\begin{cor}\label{cor:spec.prop1}  There exist  constants $C_\reg, \, \epsilon_* >0$, independent of $N$, such that the following holds true:
\begin{enumerate}[(i)]
\item The spectrum of $L_{b,a}$ is close to the spectrum of $L_0$; in particular for any $(b,a) \in B^{\Cs^\reg}\left(\frac{\epsilon_*}{N^2} \right)$
\begin{equation}
\mmod{\lambda_{2j}(b,a) - \lambda_{2j}^0}, \, \mmod{\lambda_{2j-1}(b,a) - \lambda_{2j-1}^0} \leq C_\reg \norm{(b,a)}_{\Cs^\reg}.
\label{rem:pert_eigen}
\end{equation}
\item One has that $(b,a) \mapsto P_j(b,a)$ is analytic as a map from $B^{\Cs^\reg}\left(\tfrac{\epsilon_*}{N^2} \right) $ to $\L(\C^{2N})$. Moreover  for $(b,a) \in B^{\Cs^\reg}\left(\frac{\epsilon_*}{N^2} \right)$ one has
\begin{equation}
\label{P_j-P_j0}
\norm{P_j(b,a) - P_{j0}}_{\L(\C^{2N})} \leq C_\reg \norm{(b,a)}_{\Cs^\reg}.
\end{equation} 
\item For each $1 \leq j \leq N-1$, the maps $U_j$, defined in \eqref{defUj}, are well defined from $ B^{\Cs^\reg}\left(\frac{\epsilon_*}{N^2} \right)$ to $ \L(\C^{2N}) $ and satisfy the following algebraic  properties:
\begin{enumerate}
\item[$(U1)$] $\mbox{Im } U_j(b,a)=E_j(b,a)$;
\item[$(U2)$] for $(b,a)$ real, one has $\overline{U_j(b,a)f}=U_j(b,a)\bar{f}$;
\item[$(U3)$]for $(b,a)$ real and  $f\in E_{j0}$, one has $\norm{U_j(b,a)f}_{\C^{2N}}=
\norm{f}_{\C^{2N}}$.
\end{enumerate}
Finally the following analytic property holds:
\begin{enumerate}
\item[$(U4)$] One has that $(b,a) \mapsto U_j(b,a)$ is analytic as a map from $B^{\Cs^\reg}\left(\tfrac{\epsilon_*}{N^2} \right) $ to $\L(\C^{2N})$. Moreover  for $(b,a) \in B^{\Cs^\reg}\left(\frac{\epsilon_*}{N^2} \right)$ one has
\begin{equation}
\label{U_j-P_j}
\norm{U_j(b,a) - P_{j}(b,a)}_{\L(\C^{2N})} \leq C_\reg \norm{(b,a)}^2_{\Cs^\reg}.
\end{equation}
\end{enumerate}
\end{enumerate}
\label{lem:spec.prop}
\end{cor}
The proofs of  Lemma \ref{lem:spec.prop1}  and Corollary \ref{cor:spec.prop1} can be found in Appendix \ref{spectral.object}.

\vspace{1em}
For $1 \leq j \leq N-1$ and $(b,a) \in B^{\Cs^\reg}\left(\frac{\epsilon_*}{N^2} \right)$ define now the vectors
\begin{equation}
f_{2j-1}(b,a):=U_{j}(b,a)f_{2j-1,0},  \qquad \mbox{ and } \qquad 
f_{2j}(b,a):=U_{j}(b,a)f_{2j,0}  
	\label{def:fj}
\end{equation}
which by property $(U 1)$ belong to $E_{j}(b,a)$. Define also  the maps
\begin{equation}
\begin{aligned}
& z_j(b,a):= \left(\tfrac{2}{N}\om{j}\right)^{-1/2} \left\langle \left(L_{b,a}-\lambda_{2j}^0\right)
f_{2j}(b,a),\overline{f_{2j}(b,a)}\right\rangle,\\
& w_j(b,a):=\left(\tfrac{2}{N}\om{j}\right)^{-1/2} \left\langle \left(L_{b,a}-\lambda_{2j-1}^0\right)
f_{2j-1}(b,a),\overline{f_{2j-1}(b,a)}\right\rangle
	\label{def:zj}
	\end{aligned}
\end{equation}
where  $\left\langle u,v\right\rangle=\sum{u_j\overline{v_j}}$ is the Hermitian product in $\C^{2N}$.
Finally denote  $ z(b,a)= (z_1(b,a), \ldots, z_{N-1}(b,a) )$ and  $w(b,a) = (w_1(b,a), \ldots, w_{N-1}(b,a))$, and let  $Z$ be the map
\begin{equation}
\label{map.Z}
(b,a) \mapsto Z(b,a):= (z(b,a), w(b,a)). 
\end{equation}
The map $\Psi$ of Theorem \ref{main2} will be constructed by expressing $Z$ as a function of the linear Birkhoff coordinates $\xi, \eta$.

The properties of the map $Z$ are collected in the next lemma which constitutes the main technical step for the application of  Kuksin-Perelman Theorem to the Toda lattice.
\begin{lem} 
\label{propZn}
The map $Z$, defined by \eqref{map.Z}, is well defined for $(b,a) \in
B^{\Cs^\reg}\left(\tfrac{\epsilon_*}{N^2} \right)$.  If $b, a$ are real
valued and fulfill $\norm{(b,a)}_{\Cs^\reg} \leq
\tfrac{\epsilon_*}{N^2}$, then, for
every $1 \leq j \leq N-1$, the following 
properties are also fulfilled:
\begin{enumerate}
	\item[$(Z1)$] $\overline{z_j(b,a)}=w_{j}(b,a)$;
	\item[$(Z2)$]  $\gamma_{j}^2 = \tfrac{2}{N}\om{j} \left|z_j(b,a)\right|^2 = \tfrac{2}{N}\om{j} \left|w_j(b,a)\right|^2$;
	\item[$(Z3)$] $z_j(0,0) = w_j(0,0) = 0$; moreover the linearizations of $z_j$ and $w_j$ at $(b,a) = (0,0)$ are given by
	\begin{equation}
	\begin{aligned}
	\label{Zprimordine}
&	dz_j(0,0)[(B, A)]= \left(2\om{j}\right)^{-1/2}\left(\hat{B}_{ j} -2e^{j \im \pi /N}\hat{A}_{j}\right),\\
  & dw_j(0,0)[(B, A)]= \left(2\om{j}\right)^{-1/2}\left(\hat{B}_{N-j} -2e^{-j \im \pi /N}\hat{A}_{N- j}\right).
	\end{aligned}
	\end{equation}
	The map $dZ(0,0)=(dz(0,0), dw(0,0))$ is in the class $ \L(\Cs^{\reg}, \spazio{\reg})$. Its  adjoint $dZ(0,0)^*$ is in the class $\L(\spazio{\reg}, \Cs^{s+1, \sigma})$. Finally  there exist constants $C_{Z_1}, C_{Z_2} >0$, independent of $N$, such that for any $s \geq 0$ and $\sigma \geq 0$
	\begin{equation}
	\label{dZ(0).estimate}
	\norm{\und{dZ}(0,0)}_{\L(\Cs^\reg,\,\spazio{\reg})} \leq C_{Z_1} , \quad 
	\norm{\und{dZ}(0,0)^*}_{\L(\spazio{\reg},\,\Cs^{s+2, \sigma})} \leq C_{Z_2} N^2 \ .
	\end{equation}
\item[$(Z4)$] For any $s \geq 0$, $\sigma \geq 0$, there exist constants $C_{Z_3}, C_{Z_4}, \epsilon_*>0$, independent of $N$,  such that for every $0 < \epsilon \leq \epsilon_*$  the map $Z^0:= Z- dZ(0,0) \in \Nc_{\epsilon/N^2}\left(\Cs^\reg, \spazio{s+1,\sigma} \right)$ and the map $[dZ^0]^* \in \Nc_{\epsilon/N^2}\left(\Cs^\reg, \L(\spazio{\reg}, \Cs^{s+2,\sigma} ) \right)$. Moreover
\begin{equation}
\begin{aligned}
& \sup_{\norm{(b,a)}_{\Cs^\reg} \leq \epsilon/N^2}\norm{\und{Z^0}(b,a)}_{\spazio{s+1, \sigma}} \leq C_{Z_3} \frac{\epsilon^2}{N^2},\\ 
&\sup_{\norm{(b,a)}_{\Cs^\reg} \leq \epsilon/N^2}\norm{\und{dZ^0}(b,a)^*}_{\L(\spazio{\reg}, \, \Cs^{s+2, \sigma})}  \leq C_{Z_4} N \epsilon.
\label{dZ0*.ext}
\end{aligned}
\end{equation} 
\end{enumerate}
\end{lem}
The proof of the lemma is very technical, and  is postponed in Appendix \ref{proofZ}.

\begin{rem}
In the limit of infinitely many particles, the linearization $
dz_j(0,0)(b,a)$ at the different edges of the spectrum are given
by \begin{equation} dz_j(0,0)(B,A)\approx \frac{\hat{B}_j
    -2\hat{A}_j}{\sqrt{2\omega(j/N)}} \quad \mbox{if} \; j/N \ll 1
  \qquad dz_j(0,0)(B,A) \approx \frac{\hat{B}_j
    +2\hat{A}_j}{\sqrt{2\omega(j/N)}} \quad \mbox{if} \; 1-j/N \ll 1
  \ .
\label{ukasymptotic}
\end{equation}
The existence of  two different sequences is in   agreement with the works \cite{bambuthomas,bambuthomas2}, in which the spectrum of   the Lax operator associated to the Toda lattice is approximated, up to   a small error, by the spectrum of \textit{two} Sturm-Liouville   operators associated to \textit{two} KdV equations.  More   explicitly, in \cite{bambuthomas} the following result  is proved: take   $\alpha, \beta \in C^\infty(\T)$ such that   $\int_{\T}\alpha=\int_{\T}\beta=0$, 
  $a_j=1+ \frac{1}{N^2}\alpha(j/N)$ and
  $b_j=\frac{1}{N^2}\beta(j/N)$. Then  the spectrum of the Lax matrix   \eqref{Lperturbata} with $a_j, b_j$ as elements can be approximated at the two edges by the spectrum of the two Sturm-Liouville operators   $L=-\frac{d^2}{dx^2} + \left(\beta \pm 2 \alpha \right)$ on   $C^\infty (\T)$. 
\end{rem}

We are ready to define the map $\Psi$ of Theorem \ref{main2}: let
\begin{equation}
\begin{aligned}
\Psi: \spazio{\reg} \to \spazio{\reg},  \qquad & (\xi, \eta)\mapsto \left(\phi(\xi, \eta), \psi(\xi, \eta) \right)
\end{aligned}
\end{equation}
defined by 
\begin{equation}
\Psi = - Z \circ \Theta_\Xi; \quad \mbox{ i. e. } \quad \phi = -z \circ \Theta_\Xi, \qquad \psi = - w \circ \Theta_{\Xi}.
\label{mapPsi}
\end{equation}
We show now that $\Psi$  satisfies the properties $(\Psi 1)-(\Psi 4)$ claimed in Theorem \ref{main2}.

\vspace{1em}

{\em Proof of Theorem \ref{main2}.}  Property $(\Psi 1)$ and $(\Psi
2)$ follows by $(Z1)$ respectively $(Z2)$.  We prove now $(\Psi
3)$. By $(\Theta 1)$ and $(Z3)$ one has $\Psi(0,0)=(0,0)$. In order to
compute $d\Psi(0,0)= (d\phi(0,0), d\psi(0,0))$ note that
 $$
 d\phi(0,0)= - dz(0,0) \, d\Theta_\Xi(0,0) = -(dz(0,0) \F^{-1})\circ (\F d\Theta_\Xi(0,0)) \ .
 $$
 Let $(\hat B, \hat A) = \F d\Theta_\Xi(0,0) (\xi, \eta)$. Then   \eqref{Zprimordine} and  \eqref{Theta.lin} imply that, for $1 \leq j \leq N-1$,  
\begin{align*}
d\phi_j(0,0)(\xi, \eta) & =  -\frac{1}{\sqrt{2\omega(j/N)}}\left( \hat{B}_{j} -2e^{  \im \pi j/N} \hat{A}_j \right)  \\
& = \frac{1}{\sqrt{2\omega(j/N)}} \left( \sqrt{\frac{\omega(j/N)}{2}}
 (\xi_j + \eta_{N-j})-\im \frac{2e^{\im \pi j/N}\varpi_j}{\sqrt{2 \omega(j/N)}}(\xi_j - \eta_{N-j})\right) \equiv \xi_j \ ,
\end{align*}
where we used that $2e^{\im \pi j/N}\varpi_j = \im \om{j}$. One verifies analogously that $d\psi_j(0,0)(\xi, \eta) = \eta_j$. 
\vspace{0.5em}\\
We prove now property $(\Psi 4)$, which  is a consequence of the fact that the space of normally analytic maps is closed by composition  (see Lemma \ref{FGinN}). Fix $s \geq 0$ and $\sigma \geq 0$. Let $0 <\epsilon \leq \tfrac{\epsilon_*}{C_{\Theta_5}}$, where $C_{\Theta_5}$ is the constant in \eqref{theta.sup5}. Since 
$Z= dZ(0,0) + Z^0$ and $\Theta_\Xi = d\Theta_\Xi(0,0) + \Theta_\Xi^0$, one gets that
\begin{equation}
\begin{aligned}
&\Psi^0 = - Z^0 \circ \Theta_\Xi - dZ(0,0) \circ \Theta^0_\Xi \ .
\end{aligned}
\end{equation}
Thus properties $(Z3)$, $(\Theta 2)$ and estimate \eqref{theta.sup5} imply that there exists a constant $C >0$, independent of $N$, such that
\begin{equation*}
\mmod{\und{\Psi^0}}_{\epsilon/N^2} \equiv \sup_{\norm{(\xi, \eta)}_{\spazio{\reg}} \leq \epsilon/N^2} \norm{\und{\Psi^0}(\xi, \eta)}_{\spazio{s+1, \sigma}} \leq  \frac{C \, \epsilon^2}{N^2} \ ,
\end{equation*}
which proves the first estimate of $(\Psi 4)$. We study now the adjoint map  $d\Psi^0(\xi, \eta)^*$. Writing $d\Theta_\Xi = d\Theta_\Xi(0,0) + d\Theta_\Xi^0$ one gets that
\begin{align*}
 d\Psi^0(\xi, \eta)^* &= -d\Theta_\Xi(0,0)^*\, dZ^0(\Theta_\Xi(\xi, \eta))^* - d\Theta_\Xi^0(\xi, \eta)^*\, dZ^0(\Theta_\Xi(\xi, \eta))^* - d\Theta_\Xi^0(\xi, \eta)^*\, dZ(0,0)^*\\
 & = I + II + III.
\end{align*}
We estimate each term in the expression displayed above. In the following, if $A \in \Nc_\rho(\spazio{\reg}, \L(\spazio{\reg}, \, \spazio{s+1, \sigma}))$, we denote by 
$$
\mmod{\und{A}}_\rho \equiv \sup_{\norm{(\xi, \eta)}_{\spazio{\reg}} \leq \epsilon/N^2} \norm{\und{A}(\xi, \eta)}_{\L(\spazio{\reg}, \, \spazio{s+1, \sigma})}.
$$
We begin  by estimating $I$:
\begin{align*}
 \mmod{\und{I}}_{\epsilon/N^2} &\leq \frac{C_{\Theta_2}}{N} \sup_{\norm{(\xi, \eta)}_{\spazio{\reg}} \leq \epsilon/N^2} \norm{\und{dZ^0}(\und{\Theta_\Xi}(\xi, \eta))^*}_{\L(\spazio{\reg}, \, \Cs^{s+2, \sigma})}  \leq \frac{C_{\Theta_2}}{N} C_{Z_4} C_{\Theta_5} N\, \epsilon \leq C \epsilon , 
\end{align*}
where in the first inequality we used the second estimate of \eqref{dtheta.est}  and in the second inequality  we used the second estimate in \eqref{dZ0*.ext}. Now we study $II$:
\begin{align*}
 \mmod{\und{II}}_{\epsilon/N^2}  \leq \frac{C_{\Theta_4} \epsilon}{N^2} \sup_{\norm{(\xi, \eta)}_{\spazio{\reg}} \leq \epsilon/N^2} \norm{ \und{dZ^0}(\und{\Theta_\Xi}(\xi, \eta))^*}_{\L(\spazio{\reg}, \, \Cs^{s+2, \sigma})}  \leq \frac{C_{\Theta_4} \epsilon}{N^2} C_{Z_4}C_{\Theta_5} N \epsilon \leq  \frac{C \epsilon^2}{N},
\end{align*}
where we used the second estimate in \eqref{eq:Thata_na} and again  $(Z4)$. Finally, using again $(\Theta 2)$ and the second estimate of \eqref{dZ(0).estimate},  one has
\begin{align*}
\mmod{\und{III}}_{\epsilon/N^2}  \leq \frac{C_{\Theta_4} \epsilon}{N^2} \norm{\und{dZ}(0,0)^*}_{\L(\spazio{\reg}, \, \Cs^{s+2, \sigma})}  \leq \frac{C_{\Theta_4} \epsilon}{N^2} C_{Z_2} N^2 \leq C \epsilon \ .
\end{align*}
Collecting the estimates above one gets
$$
\mmod{[\und{d \Psi^{0}}]^*}_{\epsilon/N^2} \equiv \sup_{\norm{(\xi, \eta)}_{\spazio{\reg}} \leq \epsilon/N^2} \norm{\und{d \Psi^{0}}(\xi, \eta)^*}_{\L(\spazio{\reg}, \, \spazio{s+1, \sigma})} \leq 3C \epsilon,
$$
and $(\Psi 4)$ follows.
\qed

\vspace{1em}
\noindent 
{\it Proof of Corollary \ref{corkuksinperelman}.} Provided $0 < R < R'_\reg$ is small enough, one has that
$w_0:=\Phi^{-1}_N(v_0)$ fulfills 
$$\norma{w_0}_{\spazio{\reg}}\leq \frac{R}{N^2}(1+CR)\ ,$$
and, denoting by $w(t)$ the solution in Birkhoff coordinates, one has 
$\norma{w_0}_{\spazio{\reg}}=\norma{w(t)}_{\spazio{\reg}}$. Thus,
provided $0< R < R'_\reg$ is small enough one has
$$
\norma{v(t)}_{\spazio{\reg}}=\norma{\Phi_N(w(t))}_{\spazio{\reg}}
\leq \frac{R}{N^2}(1+C'R)
$$
which implies the thesis. \qed

\subsection{Proof of Theorem \ref{inverso} }\label{p.inverso}

The proof is based on the construction of the first terms of the
Taylor expansion of $\Phi_N$ through Birkhoff normal form. To this end
we work with the complex variables $(\xi,\eta)$ (defined in
\eqref{xi_variable}) and will eventually restrict to the real subspace
$\spazior{\reg}$.
\begin{remark}
\label{Tay}
Consider the Taylor expansion of $\Phi_N$ at the origin, one has
$$
\Phi_N=\uno+\tdue{\Phi_N}+ O(\norma{(\xi,\eta)}_{\spazio{\reg}}^3)\ ,
$$
then $\tdue{\Phi_N}$ is a bounded quadratic polynomial.
Furthermore, since $\Phi_N$ is canonical, $\tdue {\Phi_N}$ is a
Hamiltonian vector field, i.e. there exists a cubic complex valued
polynomial $\chi_{\Phi_N}$ s.t. $\tdue{\Phi_n}$ is the Hamiltonian
vector field of  $\chi_{\Phi_N}$. 
 \end{remark}

We need a preliminary result about a uniqueness property of the
transformation introducing Birkhoff coordinates (called below Birkhoff
map).

\begin{lem}
\label{uniq.bnf} 
Let $\Phi_N$ and $ \Psi_N$ be
Birkhoff maps for $H_{Toda}$, analytic in some neighborhood of the
origin;  assume that $d\Phi_N(0,0) \equiv d\Psi_N(0,0)= \uno$ and denote
by $\chi_{\Phi_N}$ and $\chi_{\Psi_N}$ the Hamiltonian functions
corresponding to $\tdue{\Phi_N}$ and $\tdue{\Psi_N}$ respectively,
then one has
\begin{equation}
\label{poi.sd}
\left\{H_0;\chi_{\Phi_N}-\chi_{\Psi_N}\right\}=0\ ,
\end{equation}
where $H_0$ is defined in \eqref{quad.part}. 
\end{lem}
\begin{proof}
By a standard computation of the Taylor expansion one has
$$
H_{Toda}\circ\Phi_N=H_0+\left\{H_0,\chi_{\Phi_N}\right\}+H_1+h.o.t.
$$
where $H_1$ is the function
$$
H_1(q)=\sum_{j=0}^{N-1}\frac{(q_j-q_{j+1})^3}{6}\ 
$$ Since $\Phi_N$ is a Birkhoff map, the function $H_{Toda}\circ
\Phi_N$ is in Birkhoff normal form so in particular its Taylor
expansion contains only terms of even degree. Thus the cubic terms in
the expansion above must vanish: $\left\{H_0,\chi_{\Phi_N}\right\}+
H_1 = 0 $.  The same argument holds also for the map $\Psi_{N}$, thus
the thesis follows.
\end{proof}

\begin{remark}
\label{unico}
Writing as usual 
\begin{align*}
\chi_{\Phi_N}(\xi,\eta)=\sum_{|K|+|L|=3}\chi_{K,L}\xi^K\eta^L\ ,
\end{align*}
one gets that, since 
$$
\left\{ H_0,\chi_{\Phi_N}\right\}=-\sum_{|K| + |L|=3} \im \omega\cdot (K-L)
\chi_{K,L}\, \xi^K\eta^L \ ,
$$ eq. \eqref{poi.sd} implies that, if for some $K,L$ one has
$\omega\cdot(K-L)\not=0$, then $\chi_{K,L}$ is unique and coincides
with $\frac{H_{K,L}}{\im\omega\cdot(K-L)}$ with an obvious definition
of $H_{K,L}$.
\end{remark}

\begin{lemma}
\label{huno}
In terms of the variables $(\xi,\eta)$ one has 
 \begin{align*}
\label{huno1}
H_1(\xi, \eta) = \frac{1}{12\sqrt{2N}}\left[\sum_{\substack{k_1+k_2+k_3= 0 \bmod N \\ 1 \leq k_1, k_2, k_3, \leq N-1}}  (-1)^{\frac{k_1 + k_2 + k_3}{N}} \sqrt{\omega_{k_1}}\sqrt{\omega_{k_2}}\sqrt{\omega_{k_3}}\left(
  \xi_{k_1}\xi_{k_2}\xi_{k_3}
  + \eta_{k_1}\eta_{k_2}\eta_{k_3}\right)\right.  \\ \left.
+ 3  \sum_{\substack{k_1+k_2-k_3=0 \bmod N \\ 1 \leq k_1, k_2, k_3 \leq N-1}}(-1)^{\frac{k_1 + k_2 - k_3}{N}} \sqrt{\omega_{k_1}}\sqrt{\omega_{k_2}}\sqrt{\omega_{k_3}}
  \left( \xi_{k_1}\xi_{k_2}\eta_{k_3}
  +\eta_{k_1}\eta_{k_2}\xi_{k_3}\right) \right]
\end{align*}\end{lemma}
\proof First remark that
$$
q_j-q_{j+1}=\frac{1}{\sqrt N}\sum_{k=0}^{N-1} \hat q_k
\left(1-e^{-\frac{2\pi \im k}{N}}\right)e^{-\frac{2\pi \im jk}{N}}=
\frac{1}{\sqrt N}\sum_{k=0}^{N-1} \im \omega_ke^{-\frac{\im\pi
    k}{N}}\hat q_k e^{-\frac{2\pi \im jk}{N}}\ ,
$$
so that 
\begin{align*}
\frac{1}{6}\sum_{j=0}^{N-1}(q_j-q_{j+1})^3=\frac{\im ^3}{6N^{3/2}}
\sum_{k_1,k_2,k_3} \omega_{k_1}\hat q_{k_1}\omega_{k_2}\hat
q_{k_2}\omega_{k_3}\hat q_{k_3}
e^{-\frac{\im\pi}{N}(k_1+k_2+k_3)}\sum_{j=0}^{N-1} e^{\frac{2\pi \im
    j}{N}(k_1+k_2+k_3)}
\\
= \frac{\im ^3}{6N^{1/2}}
\sum_{k_1+k_2+k_3= 0 \bmod  N} (-1)^{\frac{k_1 + k_2 + k_3}{N}} \omega_{k_1}\hat q_{k_1}\omega_{k_2}\hat
q_{k_2}\omega_{k_3}\hat q_{k_3}\ .
\end{align*}
Substituting 
$$
\omega_k\hat q_k=\sqrt{\omega_k}\frac{\xi_k-\eta_{N-k}}{\im\sqrt 2}
$$
and reorganizing the terms one gets the thesis.
\qed

\begin{lemma}
\label{vbar}
For any $s\geq0$, $\sigma\geq 0$, there exists $C>0$ s.t. one has
\begin{equation}
\label{stisec.der}
\norma{\tdue{\phi_N}(\bar v)}_{\spazio{\reg}}\geq CN^2\norma {\bar v}
_{\spazio{\reg}}^2\ , 
\end{equation}
where $\bar v=((\xi_1,0,0,...,0),(\bar\xi_1,0,0,...,0))\in
\spazior{\reg}$. 
\end{lemma}
\proof In this proof, for clarity we denote $\eta_1:=\bar\xi_1$, and
similarly for the other variables. We are going to compute the $\xi_2$
component $[\tdue{\Phi_N}(\bar v)]_{\xi_2}$ of $\tdue{\Phi_N}(\bar v)$
and exploit the inequality
\begin{equation}
\label{chi.3d}
\norma{\tdue{\Phi_N}(\bar v)}_{\spazio{\reg}}  \geq \frac{1}{\sqrt
    N}2^s e^{\sigma 2}\omega_2^{1/2}\frac{1}{\sqrt 2}\left|\left[
    \tdue{\Phi_N}(\bar v)  \right]_{\xi_2} \right|
 =\frac{2^{s}e^{\sigma 2}\omega_2^{1/2}}{\sqrt{2N}}\left|\frac{\partial
  \chi_{\Phi_N}}{\partial \eta_2}(\bar v)\right|\ ;
\end{equation}
the only monomials in $\chi_{\Phi_N}$ contributing to such a quantity are
quadratic in $(\xi_1,\eta_1)$ and linear in $\eta_2$, but due to the
selection rule $k_1\pm k_2\pm k_3=lN$ with a plus for the $\xi$'s and
a minus for the $\eta$'s the only monomial contributing to the
r.h.s. of \eqref{chi.3d} is $\chi_{\bar K,\bar L}\xi^{\bar
  K}\eta^{\bar L}$ with
$\bar K:=(2,0,...,0)$, and $\bar L=(0,1,0,0,...,0)$. 

Since
\begin{align}
\label{non.res}
\omega\cdot(K-L)=2\omega_1-\omega_2=4\sin\frac{\pi}{N}-2\sin\frac{2\pi}{N}
= \frac{2\pi^3}{N^3}+O\left(\frac{1}{N^5}\right) \not=0\ ,
\end{align}
such a coefficient is uniquely defined and, for the 
$\chi_{\Phi_N}$ corresponding to {\it any}
Birkhoff map, one has
\begin{equation}
\label{ilchi}
\chi_{\bar K,\bar
  L}=  \frac{1}{4\sqrt{2N}}\frac{\omega_1\omega_2^{1/2}}{\im
  (2\omega_1-\omega_2)} \ .
\end{equation}
Inserting in \eqref{chi.3d} one has that its r.h.s. is equal to
\begin{align*}
\frac{2^{s}e^{\sigma 2}\omega_2^{1/2}}{\sqrt{2N}}\left|\chi_{\bar K,\bar
  L}\right|\left|\xi_1\right|^2 
=
\frac{C''}{N}\frac{\omega_1\omega_2}{|2\omega_1-\omega_2|}\left|\xi_1\right|^2 
=C'\frac{\omega_2}{|2\omega_1-\omega_2|}\norma{\bar v}_{\spazio{\reg}}^2 
\geq C N^2 \norma{\bar v}_{\spazio{\reg}}^2\ ,
\end{align*}
where $C$, $C'$ and $C''$ are numerical constants independent of $N$
and we used the expansions of $\omega_1$, $\omega_2$ in $1/N$ as well
as equation \eqref{non.res}.\qed

\vspace{1em}
\noindent{\it Proof of Theorem \ref{inverso}.} The thesis immediately
follows taking $\norma{\bar v}_{\spazio{\reg}}=R/N^\alpha$ and
imposing the inequality \eqref{assuQ}.\qed 

\vspace{1em}
\noindent{\it Proof of Corollary \ref{inverso1}.} By
Cauchy inequality and assumption \eqref{assu} $\tdue{\Phi_N }$ fulfills
\begin{equation}
\label{st.cau}
\norma{\tdue{\Phi_N}(\bar v)}_{\spazio{\reg}}\leq
\frac{R'}{N^{\alpha'}}\frac{N^{2\alpha}}{R^2}
\norma{\bar v}_{\spazio{\reg}}^2 \ .
\end{equation}
Comparing this inequality with \eqref{stisec.der}, one
gets
$$
\frac{R'}{R^2}N^{2\alpha-\alpha'}\geq C'' N^2\ ,
$$
which in particular implies the thesis.\qed

\section{FPU packet of modes: proofs.}

In this section we prove the results stated in the subsection \ref{FPU} about
the persistence of the metastable packet in the FPU system. 

To clarify the procedure, we distinguish here between the $(\xi,\eta)$
variables and the variables $(p,q)$. Thus, we denote by $T: (\xi,
\eta) \to (p,q)$ the change of coordinates of the phase space
introducing the linear Birkhoff variables $(\xi, \eta)$ defined in
\eqref{xi_variable}.  Furthermore it is useful to use for the $(p,q)$
variables the following norms
\begin{equation}
\label{norm.reg} 
\norm{q}_\reg^2 := \frac{1}{N}
\sum_{k=0}^{N-1} \max(1, [k]_N^{2s}) \,e^{2\sigma [k]_N}
\,|\hat{q}_k|^2 \ ,
\end{equation}
and 
\begin{equation}
\label{norm2.e}
\norma{(p,q)}_{\spazio{\reg}}:=\norma{T^{-1}(p,q)}_{\spazio{\reg}}\ .
\end{equation}

\begin{lem}
\label{lem:fpu.grad} 
Fix $s\geq 1$, $\sigma\geq 0$, 
then there exist constants $C_1, C_2>0$,
independent of $N$, such that for all $(\xi, \eta) \in
{\spazio{\reg}}$ and $\forall l\geq 2$ one has
\begin{align}
\label{fpu.bc.est} 
&\norm{X_{H_l\circ T}(\xi, \eta)}_{\spazio{\reg}} \leq \frac{C_1^{l}}{(l+1)!}
\norm{(\xi, \eta)}_{\spazio{\reg}}^{l+1}\ ,
\\
\label{fpu.bc.est2}
&\norm{ X_{H_l\circ T}(\xi, \eta) }_{\spazio{s-1, \sigma}} \leq
\frac{C_2^l}{N(l+1)!} \norm{(\xi, \eta)}_{\spazio{\reg}}^{l+1}.
\end{align}
\end{lem}
\proof Define the difference operators by
\begin{equation}
\label{diff.oper}
S_\pm: \{ q_j \}_{0 \leq j \leq N-1} \mapsto \{ q_j-q_{j\pm 1}\}_{0
  \leq j \leq N-1} \ ,\qquad \mbox{where} \ q_N \equiv q_0 \ ,
\end{equation}
and the operator $[ S_{+}(q)]^l $ by
$$
\left\{\left[ S_{+}(q)\right]^l\right\}_j:=(q_j-q_{j+1})^l\ ,
$$
so that 
\begin{equation}
\label{hlcap}
X_{H_l\circ
  T}(\xi,\eta)=\frac{1}{(l+1)!} 
T^{-1}\left( \ S_-\left[S_+(T(\xi,\eta))\right]^l, \ 0 \right)
\ .
\end{equation}
By Lemma  \ref{bullet} and Remark
\ref{rem:s.ext} in Appendix \ref{app.DFT}, there exists a constant $C_\reg >0$, independent of $N$, such that for every integer $n\geq 1$
\begin{equation}
\label{diff.oper.est} 
\norm{[S_{\pm}(q)]^{l+1}}_{\reg} \leq C_\reg^{l+1}  \norm{S_\pm(q)}_\reg^{l+1} \leq
C_\reg^{l+1} \norma{(\xi,\eta)}_{\spazio{\reg}}^{l+1}\ ,
\end{equation}
where for the last inequality we have identified the couple $(0,q)$
with the corresponding $(\xi,\eta)$ vector. 

Then the thesis follows just remarking that
$\norma{T^{-1}(q,0)}_{\spazio{\reg}}=\norma{q}_{\reg}$, and
that $S_-$ is bounded as an operator from $\spazio{\reg}$ to itself,
while one has
$$
\norma{(S_-(q), 0)}_{\spazio{s-1,\sigma}}\leq
\frac{C}{N}\norma{q}_{\reg}
\ . 
$$
\qed

Introducing the Birkhoff coordinates and using the standard formulae
for the pull back of vector fields\footnote{Namely $$
[\Phi_N^*X](x)=d\Phi_N^{-1}(\Phi_N(x))X(\Phi_N(x))\ 
$$ which gives the vector field of the transformed Hamiltonian due to
the fact that $\Phi_N$ is canonical} one has the following 
\begin{corollary}
\label{birhl}
Fix $s\geq 1$ 
and  $\sigma\geq 0$, 
then there exist constants $R_\reg,C_1, C_2>0$,
independent of $N$, such that for all $w\equiv(\phi, \psi) \in
B^{\reg}(R_\reg/N^2)$  one has
\begin{align}
\label{fpu.bc.est.1} 
&\norm{X_{H_l\circ T\circ\Phi_N}(w)}_{\spazio{\reg}} \leq
\frac{C_1^{l}}{(l+1)!}  \norm{w}_{\spazio{\reg}}^{l+1}\ , \\
\label{fpu.bc.est.2}
&\norm{ X_{H_l\circ T\circ\Phi_N}(w) }_{\spazio{s-1, \sigma }} \leq
\frac{C_2^l}{N(l+1)!} \norm{w}_{\spazio{\reg}}^{l+1}.
\end{align}
\end{corollary}

\begin{remark}
\label{rem.hp}
Write
\begin{equation}
\begin{aligned}
\label{H_fpu.1} 
\tilde H_{FPU} \equiv H_{FPU}\circ T \circ \Phi_N =\tilde
H_{Toda}+\tilde H_P\ ,
\end{aligned}
\end{equation}
where 
\begin{equation}
\label{ddd}
\tilde H_{Toda}:= H_{Toda}\circ T \circ \Phi_N \ ,\quad \tilde
H_P:=(\beta -1) H_2\circ T \circ \Phi_N + H^{(3)}\circ T \circ
\Phi_N\ ,
\end{equation}
then, provided $R$ is small enough the vector field of
$\tilde H_P$ fulfills the following estimates 
\begin{align}
\label{fpu.bc.est.5} 
&\norm{X_{\tilde H_P}(w)}_{\spazio{\reg}} \leq C
\left[|\beta-1|\norm{w}_{\spazio{\reg}}^3+C\norm{w}_{\spazio{\reg}}^4   \right]\ ,
\\
\label{fpu.bc.est.6}
&\norm{ X_{\tilde H_P}(w) }_{\spazio{s-1, \sigma }}
\leq\frac{C}{N}
\left[|\beta-1|\norm{w}_{\spazio{\reg}}^3+C\norm{w}_{\spazio{\reg}}^4
  \right] \ ,
\end{align}
for all $w\in B^{\reg}(R/N^2)$.
\end{remark}

In the following we denote by $v(t)\equiv (\xi(t),\bar\xi(t))$ the
solution of the FPU model in the original Cartesian coordinates (we
restrict to the real subspace). We denote by $w(t):=\Phi_N^{-1}(v(t))$
the same solution in Birkhoff coordinates. 

\begin{lem}
\label{lem:E-E0}Fix $s \geq 2\,$ and  $\sigma \geq 0$.
Then there exist $R'_\reg, \, T, \, C_2 >0$ such that $v_0\in
\Br^{\reg}\left(\frac{R}{N^2}\right)$ with $ R\leq R'_{\reg}$
implies $v(t)\in \Br^{\reg}\left(\frac{4R}{N^2}\right)$ for
\begin{equation}
\label{tempi.dim}
|t|\leq \frac{T}{R^2\mu^4[|\beta-1|+C_2R\mu^2]}\ .
\end{equation}
\end{lem}
\begin{proof}
First consider $w_0:=\Phi^{-1}_N(v_0)$ and remark that (provided
$R'_{\reg}$ is small enough) one has $w_0\in
\Br^{\reg}\left(\frac{2R}{N^2}\right) $. Denote by
$M(w):=\norma w_{\spazior{\reg}}^2$. Since $\left\{
M,\tilde H_{Toda}\right\}\equiv 0$, one has
\begin{equation}
\label{poi.M}
M(w(t))=M(w_0)+\int_0^t\left\{ M; \tilde H_P\right\}(w(s)) \di s\ .
\end{equation}
Denoting $\bar M(t):=\sup_{|s|\leq t}M(w(s))$, one has
\begin{align}
\label{A.B.1}
\bar M(w(t))\leq M(w_0)+\int_0^t\left|\left\{ M; \tilde
H_P\right\}(w(s))\right| \di s
\\
\nonumber
\leq
M(w_0)+\int_0^t\left(C\norma{w(s)}_{\spazior{\reg}}^4\left|\beta-1\right|
+ C\norma{w(s)}_{\spazior{\reg}}^5\right) \di s
\\
\nonumber
\leq 
M(w_0)+\int_0^tC\bar M(t)^2\left( \left|\beta-1\right|+C\bar
M(t)^{1/2}\right) \di s
\\
\label{A.B}
\leq M(w_0)+|t| C\bar M(t)^2\left(\left|\beta-1\right|+C\bar
M(t)^{1/2}\right)\ , 
\end{align}
where, in order to prove the second inequality we used $\left\{ M;
\tilde H_P\right\}:=\di MX_{\tilde H_P}$ and  
$$ \norm{dM(w)}_{\L(\spazio{\reg}, \C)} \leq C \norm{w}_{\spazio{\reg}} \ ,
$$ which follows from an explicit computation. Taking $t$ as in the
statement of the Lemma we have that \eqref{A.B.1}-\eqref{A.B}
ensures $\bar M(t)\leq 9M(w(0))/4$, which implies
$w(t)\in\Br^{\reg}\left(\frac{3R}{N^2}\right)$ from which the
thesis immediately follows.
\end{proof}

\noindent {\it Proof of Theorem \ref{N.1}.} Inequality \eqref{N.1.2} is
a direct consequence of Lemma \ref{lem:E-E0}. To prove inequality 
\eqref{N.1.3}  remark that
$\dot I_k = \{I_k, \tilde H_P \} = x_k \frac{\partial \tilde H_P}{\partial y_k}  -  y_k \frac{\partial \tilde H_P}{\partial x_k}$. 
Thus
\begin{align*}
&\frac{1}{N}\sum_{k=1}^{N-1} [k]_N^{2s-2}e^{2 \sigma [k]_N} \om{k}\mmod{\{I_k, \tilde H_P \}} = \frac{1}{N}\sum_{k=1}^{N-1} [k]_N^{2s-2}e^{2 \sigma [k]_N} \om{k}\mmod{ y_k \frac{\partial \tilde H_P}{\partial y_k} - x_k \frac{\partial \tilde H_P}{\partial x_k}}\\
&  \leq \left( \frac{1}{N}\sum_{k=1}^{N-1} [k]_N^{2s-2}e^{2 \sigma [k]_N}\om{k}(y_k^2+ x_k^2)\right)^{1/2}
\left( \frac{1}{N}\sum_{k=1}^{N-1} [k]_N^{2s-2}e^{2 \sigma [k]_N}\om{k}\left(\mmod{\frac{\partial \tilde H_P}{\partial y_k}}^2 + \mmod{\frac{\partial \tilde H_P}{\partial x_k}}^2\right)\right)^{1/2} \\
&  \leq 2\norm{w}_{\spazior{s-1, \sigma}} \norm{X_{\tilde H_P}(w)}_{\spazior{s-1, \sigma}} \leq\frac{C}{N}
\left[|\beta-1|\norm{w}_{\spazio{\reg}}^4+C\norm{w}_{\spazio{\reg}}^5
  \right] \ ,
\end{align*}
where in the last inequality we used \eqref{fpu.bc.est.6}.
Using that $\mmod{I_k(w(t)) - I_k(w(0))} \leq \int_0^t \mmod{\{I_k, \tilde H_P \}(w(s))} \, ds $, one gets 
$$
\frac{1}{N}\sum_{k=1}^{N-1} [k]_N^{2s-2}e^{2 \sigma [k]_N} \om{k} \mmod{I_k(w(t)) - I_k(w(0))} \leq \frac{|t| C}{N} \sup_{|s|\leq t}
\left[|\beta-1|\norm{w(s)}_{\spazio{\reg}}^4+C\norm{w(s)}_{\spazio{\reg}}^5
  \right] \ ,
$$
which, using $w(t)\in\Br^{\reg}\left(\frac{3R}{N^2}\right)$
immediately implies the thesis.
\qed

\appendix

\section{Properties of normally analytic maps} 
\label{propnagerms}

In this section we study the properties of the space $\mathcal{N}_{\rho}(\spazio{w^1}, \, \spazio{w^2})$ and
$\mathcal{A}_{w^1, \rho}^{w^2}$ defined in section \ref{KP.section}, with weights $w^1 \leq w^2$. In particular, we consider the operations on germs
defined in \cite{kuksinperelman} and perform quantitative estimates.

\begin{lem}Let $w^1 \leq w^2 \leq w^{3}$ be weights. Let $G\in\Nc_{\rho}(\spazio{w^1}, \spazio{w^2})$ with $\mmod{\und{G}}_{\rho}\leq
  \sigma$ and $ F \in
  \mathcal{N}_{\sigma}(\spazio{w^2}, \spazio{w^{3}})$. Then $F \circ G \in \mathcal{N}_\rho(\P^{w^1},\P^{w^{3}})$ and
$\mmod{\und{F \circ G}}_\rho\leq \mmod{\und{F}}_\sigma.$
\label{FGinN}
\end{lem}
\begin{proof}Exploiting the obvious inequality $\underline{F\circ
    G}(|v|)\leq \underline{F}\circ \underline{G}(|v|)$(cf \cite{kuksinperelman}), one has
\begin{align*}
\mmod{\und{F \circ G}}_\rho &\equiv \sup_{v\in
  B^{w^1}(\rho)}\norm{\und{F\circ G}(|v|)}_{w^{3}}\leq \sup_{v\in
  B^{w^1}(\rho)}\norm{\und{F}( \und{G}(|v|))}_{w^{3}}\leq \sup_{u \in
  B^{w^2}(\sigma)}{\norm{\und{F}(|u|)}}_{w^{3}}\equiv \mmod{\und{F}}_\sigma.
\end{align*}
\end{proof}

\begin{lem}
\label{Ginverso}
Let  $F \in \mathcal{N}_\rho(\spazio{w^1}, \, \spazio{w^2})$,  $F=O(v^2)$ and $\mmod{\und{F}}_\rho \leq \rho/e$. Then the map $\uno + F$ is invertible in $B^{w^1}(\mu \rho)$, $\mu$ as in \eqref{mu.def}. Moreover 
there exists  $G\in\Nc_{\mu\rho}(\spazio{w^1}, \, \spazio{w^2})$, $G=O(v^2)$, such that 
  $(\uno+F)^{-1}=\uno-G$, and  
\begin{equation}
\label{G.inv.est}
\mmod{\und{G}}_{\mu\rho}\leq
\frac{\mmod{\und{F}}_\rho}{8} .
\end{equation}
\end{lem}
\begin{proof}We look for $G$ in the form $G=\sum_{n \geq 2}G^n$, with the homogeneous polynomial $G^n$ to be determined at every order $n$. Note  that the equation defining $G$ can be given in
  the form $F(v-G(v))=G(v)$, which can be recasted in a recursive way
  giving the formula
\begin{equation}
\label{recursive_inverse}
{G}^n(v)=
\sum_{r=2}^{n}\sum_{k_1+\cdots+k_r=n}{\tilde{F}}^r\Big({G}^{k_1}(v),\cdots,{G}^{k_r}(v)\Big), \qquad \forall \, n \geq 2 \ .
\end{equation}
In the formula above  $k_1, \ldots, k_r \in \N$, and we write  $F = \sum_{r \geq 2} F^r$, where $
F^r$ is a homogeneous polynomial of degree $r$ and $\tilde{F}^r$ is its
associated multilinear map (see \eqref{polin}). Moreover we write
$G^1(v) := v$.  We show now that the formal series $G = \sum_{n \geq 2}
G^n$ with $G^n$ defined by \eqref{recursive_inverse} is normally
analytic in $B^{w^1}(\mu \rho)$. Note that
\begin{equation}
 \und{G^n}(|v|)\leq \sum_{r=2}^{n}\sum_{k_1+\cdots+k_r=n}
 \und{\tilde{F}^r}\Big(\und{G^{k_1}}(|v|),\ldots,\und{G^{k_r}}(|v|)\Big).
\end{equation}
In order to prove that the series  $\sum_{n \geq 2}\und{G^n}$ is convergent in $B^{w^1}(\mu \rho)$, we prove that there exists a constant  $A>0$ such that
\begin{equation}
\norm{\und{G^n}(|v|)}_{w^2}\leq \frac{\mmod{\und{F}}_\rho}{8Sn^2}A^n\norm{v}_{w^1}^n, \qquad \forall n \geq 2.
\label{induzioneserie}
\end{equation}
The proof is by induction on $n$. We will use in the following the chain of inequalities  
$$\norm{\und{\tilde{F}}^r}\leq
e^r\norm{\und{F}^r}\leq e^r \mmod{\und{F}}_\rho / \rho^r \qquad \forall r \geq 1 \ ,$$
see \cite{mujica}. For $n=2$, by \eqref{recursive_inverse} it follows that   $G^2(v) = \tilde{F}^2(v,v)$. Since
$$
\norm{\und{G^2}(|v|)}_{w^2} \leq \norm{\und{\tilde F^2}}\norm{v}^2_{w^1} \leq e^2\frac{\mmod{\und{F}}_\rho}{\rho^2} \norm{v}^2_{w^1},
$$
it follows that  \eqref{induzioneserie} holds for $n=2$ with   $A =  \frac{e(32S)^{1/2}}{\rho}$.
We prove now the inductive step $n-1 \leadsto n$. Assume therefore that \eqref{induzioneserie} holds up to order $n-1$. Then one has
\begin{align*}
 \norm{\und{G^n}(|v|)}_{w^2}&\leq\sum_{r=2}^{n}
 \sum_{k_1+\cdots+k_r=n}
 \norm{\und{\tilde{F}^r}}
 \norm{\und{G^{k_1}}(|v|)}_{w^2}\cdots
 \norm{\und{G^{k_r}}(|v|)}_{w^2}\\ 
 &\leq A^n\norm{v}_{w^1}^n
 \sum_{r=2}^{n}\sum_{k_1+\cdots+k_r=n}
 e^r\frac{\mmod{\und{F}}_\rho}{\rho^r}\frac{\mmod{\und{F}}_\rho^r}{8^r
   S^r k_1^2 \cdots k_r^2}\\ &\leq\frac{\mmod{\und{F}}_\rho}{4 S
   n^2}A^n
 \norm{v}_{w^1}^n\sum_{r=2}^\infty{\left(\frac{e\mmod{\und{F}}_\rho}{2\rho}\right)^r}\leq
 \frac{\mmod{\und{F}}_\rho}{8Sn^2}A^n\norm{v}_{w^1}^n
\end{align*}
where in the first inequality we used the fact that $w^1 \leq w^2$, in the second the inductive assumption and  in the last  we used  the hypothesis $\mmod{\und{F}}_\rho \leq \rho/e$. 
Finally to pass from the second to the third line we used the following inequality, proved in Lemma \eqref{disserieconv} below:
\begin{equation}
n^2\sum_{k_1+\cdots+k_r=n}\frac{1}{ k_1^2 \cdots k_r^2}\leq (4S)^{r-1}, \qquad n \geq 1 \ .
\label{disserieconv2}
\end{equation}
\newline
  Hence, choosing $\mu \rho= 1/A =\rho/e(32 S)^{1/2}$ one proves \eqref{G.inv.est}.
\end{proof}

Now it is easy to prove the following lemma, giving closedness of the
class $\A_{w^1, \rho}^{w^2}$ under different operations. 
\begin{lem}
\label{FGinA}
  Let $w^1\leq w^2$ be weights and let $\mu$ be as in \eqref{mu.def}. Then the following holds true:
\begin{enumerate}[i)]
	\item Let $F \in \A_{w^1,\rho}^{w^2}$ and
          $G\in\A_{w^1,\mu\rho}^{w^2}$ with
          $\norm{G}_{\A_{w^1,\mu\rho}^{w^2}}<\tfrac{\mu\rho}{e}$. Then
          $H(v):= F(v+G(v))$ is of class $ \A_{w^1, \mu\rho}^{w^2}$ and
$$
\norm{H}_{\A_{w^1, \mu\rho}^{w^2}} \leq 2 \norm{F}_{\A_{w^1, \rho}^{w^2}} \ .
$$
\item Let $F \in \mathcal{A}_{w^1, \rho}^{w^2}$ and $\|F\|_{\A_{w^1,\rho}^{w^2}}\leq\rho/e$. Then
  $(\uno+F)^{-1}= \uno + G$, with $G\in \mathcal{A}_{w^1,\mu\rho}^{w^2}$. Moreover one has 
\begin{equation}
\label{sti.inv.2}
\norma{G}_{\A_{w^1,\mu\rho}^{w^2}}\leq 2\norma{F}_{\A_{w^1,\rho}^{w^2}} \ .
\end{equation}
\item Let $F \in \A_{w^1,\rho}^{w^2}$, then the function  $H(v) := dF(v) v$ is in the class $\A_{w^1,\mu\rho}^{w^2}$ and
$$\norma{H}_{\A_{w^1,\mu\rho}^{w^2}} \leq 2 \norma{F}_{\A_{w^1,\rho}^{w^2}} \ . $$
\item Let $F^0, G^0 \in \A_{w^1,\rho}^{w^2}$ with $\norm{F^0}_{\A_{w^1,\rho}^{w^2}} \leq \tfrac{\rho}{e}$. Denote  $F = \uno + F^0$. Then  $H(v):=dG^0(v)^*(F(v))$ is in the class $\A_{w^1,\mu\rho}^{w^2}$ and $$ \norma{H}_{\A_{w^1,\mu\rho}^{w^2}} \leq 2 \norma{G^0}_{\A_{w^1,\rho}^{w^2}}.$$
\end{enumerate}
\end{lem}
\begin{proof}
\begin{enumerate}
\item[$i)$] Since $\und{H}(|v|) \leq \und{F}(|v| + \und{G}(|v|))$ it follows that $\mmod{\und{H}}_{\mu\rho} \leq \mmod{\und{F}}_{2\mu\rho} \leq \mmod{\und{F}}_{\rho}$. Furthermore, since $dH(v) = dF(v+ G(v)) (\uno + dG(v))$  one gets that $\und{dH}(|v|) \leq \und{dF}(|v| + \und{G}(|v|)) + \und{dF}(|v| + \und{G}(|v|)) \und{dG}(|v|)$, which implies that
$\mu\rho\mmod{\und{dH}}_{\mu\rho} \leq \mmod{\und{dF}}_{\rho}(\mu\rho+\mu\rho\mmod{\und{dG}}_{\mu\rho})\leq \mmod{\und{dF}}_{\rho}\mu\rho(1 + 1/e)$.
The adjoint $dH(v)^*$ is estimated analogously, thus the claimed estimate  follows.
\item[$ii)$] It follows from the formula
$
dG(v)= [\uno-dF(v-G(v))]^{-1}dF(v-G(v)) , 
$
arguing as in item $i)$.
\item[$iii)$] It follows  from $dH(v)u = dF(v)u + d^2F(v)(u,v),$ arguing as in item $i)$.
\item[$iv)$] To estimate  $\und{H}(|v|)$ and $\und{dH}(|v|)$ one proceeds as in item $i)$. In order to estimate $\und{dH}(|v|)^*$ remark that (see \cite{kuksinperelman}) 
$
dH(v)^*u = (dF^0(v)^* + \uno) dG^0(v)u + d_v (dG^0(v)^*u)(F(v)), 
$
thus
$$
\und{dH}(|v|)^*|u| \leq (\und{dF^0}(|v|)^* + \uno) \und{dG^0}(|v|)|u| + d_{|v|} (\und{dG^0}(|v|)^*|u|)(\und{F}(|v|)) \ .
$$
The claimed estimate follows easily.
\end{enumerate}
\end{proof}
Now we analyze the flow generated by a vector field of class
$\A_{w^1,\rho}^{w^2}$.  Given a time dependent vector field $V_t(v)$,
consider the differential equation
\begin{equation}
\label{ode}
\begin{cases}
\dot{u}(t)=V_t(u(t))\\
u(0)=v\ .
\end{cases}
\end{equation}
We will denote by $\phi^t(v)$ the corresponding flow map whose
existence and properties are given in the next lemma. 

\begin{lem}
\label{flussoinA}
Assume that the map $\left[0,1\right]\ni t\mapsto
V_t\in\A_{w^1,\rho}^{w^2} $ is continuous and furthermore fulfills $\sup_{t \in
  [0,1]}\norma{V_t}_{\A_{w^1,\rho}^{w^2}}\leq \rho/e$;  then for each
$t\in \left[0,1\right]$, $\phi^t-\uno \in \A_{w^1,\mu\rho}^{w^2}$ with
$\mu$ as in \eqref{mu.def}. Furthermore one has
\begin{equation}
\label{sti.flu}
\norma{\phi^t-\uno}_{\A_{w^1,\mu\rho}^{w^2}}\leq
2\sup_{t\in[0,1]}\norma{V_t}_{\A_{w^1,\rho}^{w^2}}\ .
\end{equation}
\end{lem}

\begin{proof}
We look for a solution $u(t,v)= \sum_{j \geq 1} u^j(t,v)$ in power
series of $v$, with $u^j(t,v)$ a homogeneous polynomial of degree $j$ in $v$. Expanding the vector field
$V_t(v)=\sum_{r \geq 2}V_t^r(v)$ in Taylor series, one obtains the recursive formula for the solution
\begin{equation}
\label{ric.1}
u^1(t,v)=v,\qquad
u^n(t,v)=\sum_{r=2}^{n}\sum_{k_1+\cdots+k_r=n}\int_0^t{\tilde{V}^r_s(u^{k_1}(s,v),\ldots,u^{k_r}(s,v)) \;ds}\qquad
\forall n\geq 2,
\end{equation}
where $\tilde{V}^r_s$ is the multilinear map associated to $V^r_s$ (see \eqref{polin}).
Arguing as in the proof of \eqref{Ginverso} one gets the bounds 
\begin{equation}
\norm{\underline u^n(t,v)}_{w^2} \leq \frac{\sup_{t \in
    [0,1]}\mmod{\underline {V_t}}_\rho}{8Sn^2}A^n\norm{v}_{w^1}^n \qquad \forall n\geq
2,
\label{ode.rec.est}
\end{equation} 
with $A= \tfrac{e}{\rho} (32S)^{1/2}$,  
from which it follows that $\mmod{\und{\phi^t-\uno}}_{\mu\rho} \leq \sup_{t \in [0,1]}\mmod{\und{V_t}}_\rho/8.$

		We come to the estimate of the
differential of $u(t,v)$ and of its adjoint. We differentiate equation
\eqref{ric.1} getting the recursive formula
\begin{equation}
\label{ric.2}
d u^n(t,v)\xi =\sum_{r=2}^{n}\sum_{k_1+\cdots+k_r=n}
\int_0^t{\left[ \tilde{V}^r_s(du^{k_1}(s,v)\xi,\ldots,u^{k_r}(s,v))
+\cdots+\tilde{V}^r_s(u^{k_1}(s,v),\ldots,du^{k_r}(s,v)\xi)
    \right]ds} \ .
\end{equation}
To estimate such an expression remark that, defining $E_t(v):=d
V_t(v)$ (where the differential is with respect to the $v$ variable
only), one has 
$$ d^{r-1}E_s(u^{k_2}(s,v),\ldots,u^{k_r}(s,v))\xi=
\tilde{V}^r_s(\xi,u^{k_2}(s,v),\ldots,u^{k_r}(s,v))
$$
which allows to write formula \eqref{ric.2} as 
\begin{equation}\begin{aligned}
\label{ric.3}
 d u^n(t,v)\xi &=
\sum_{r=2}^{n} \sum_{k_1+\cdots+k_r=n}
\int_0^t \left[ d^{r-1}E_s (u^{k_2}(s,v)\ldots,u^{k_r}(s,v)) du^{k_1}(s,v)\xi + \ldots \right. \\
&\left. \ldots +d^{r-1} E_s(u^{k_1}(s,v),\ldots,u^{k_{r-1}}(s,v)) du^{k_r}(s,v)\xi \right] ds \ .
\end{aligned}\end{equation} 
This formula allows to proceed exactly as in the estimate of $\und{u^n}$,
namely making the inductive assumption that 
$$
\norma{\underline{d u^n}(t,v)}_{\L(\spazio{w^1}, \spazio{w^2})}\leq \frac{\sup_{t \in
    [0,1]}\mmod{\und{dV_t}}_\rho}{8Sn^2}A^n\norm{v}_{w^1}^n
$$
and proceeding as above one gets the thesis. Finally one has to
estimate $[\underline{d u^n}]^*$, but again equation \eqref{ric.3} allows
to obtain a formula whose estimate is obtained exactly as the estimate
of $\underline {du}$.
\end{proof}
\vspace{1em}
We prove now a  useful inequality.
\begin{lem}
\label{disserieconv}
\textit{\cite{treves}} Let $r \in \N$ be fixed and $S = \sum_{k\geq 1}\tfrac{1}{k^2}$. Then for every $n \in \N$ it holds that
 $$n^2\sum_{\substack{k_1, \ldots, k_r \in \N \\ k_1+\cdots+k_r=n}}\frac{1}{ k_1^2 \cdots k_r^2}\leq (4S)^{r-1} \ .$$
\end{lem}
\begin{proof}
The proof is by induction, the case $n=1$ being trivial. For $n>1$ one gets 
$$n^2\sum_{k_1+\cdots+k_r=n}\frac{1}{ k_1^2 \cdots k_r^2}=\sum_{k_1+j=n}\frac{n^2}{ k_1^2 j^2}\sum_{k_2+\cdots+k_r=j}\frac{ j^2}{k_2^2\cdots k_r^2}\leq \sum_{k_1+j=k}\frac{n^2}{ k_1^2 j^2} (4S)^{r-2}$$
by the induction assumption. Now it is enough to note that
$$\sum_{k_1+j=n}\frac{n^2}{ k_1^2 j^2}=\sum_{k_1+j=n}\frac{n^2}{ k_1^2 (n-k_1)^2}\leq 2\sum_{k_1=1}^{n-1}\Big(\frac{1}{k_1^2}+\frac{1}{(n-k_1)^2}\Big)\leq4\sum_{k_1=1}^{n-1}{\frac{1}{k_1^2}}\leq 4S.$$
\end{proof}

\section{Discrete Fourier Transform}
\label{app.DFT} 
In this section we collect some well-known properties of the discrete
Fourier transform (DFT).  For $u \in \mathbb{C}^{N}$, $N \in \N$, the
DFT of $u$ is the vector $\hat{u} \in \C^N$ whose $k^{th}$ component
is defined by
\begin{equation}
\hat{u}_k=\frac{1}{\sqrt{N}} \sum_{j=0}^{N-1}{u_j e^{2\pi \im jk/N}}, \qquad \forall \, 0\leq k \leq N-1.
\label{DFT}
\end{equation}
When the DFT is considered as a map, it will be denoted by $\F$, i.e.  $\F: u \mapsto \hat{u}$.

For any $s\geq 0$ and $\sigma \geq 0$  we endow $\C^N$ with the norm $\norm{\cdot}_\reg$ defined in \eqref{norm.reg}. Such a space will be denoted by $\C^\reg$.
\begin{rem}Let $j$ be an integer such that $0\leq j \leq N-1$. Then
\begin{equation}\sum_{k=0}^{N-1}{e^{\im 2\pi jk/N}}=\left\{
                                           \begin{array}{cc}
                                           0 &\mbox{if}\  j\neq 0 \\
                                           N &\mbox{if}\  j=0
                                           \end{array}
                                           \right.
                                           \label{lemfourier}
\quad \mbox{ and } \quad
\sum_{k=0}^{2N-1}{u_k \,e^{\im \pi kj/N}}=\begin{cases}
2\sqrt{N}\, \hat{u}_{l} , & \qquad j\;\mbox{even}, \, j=2l \\
 0 & \qquad  j\;\mbox{odd}
\end{cases} 
\end{equation}
\end{rem}
\begin{remark}
\label{lem:normaNs}
 Fix $s > \tfrac{1}{2}$ and $\sigma \geq 0$. Then there exists a constant $C_{\reg}>0$, independent of $N$, such that for every $u \in \C^{N}$ the following estimate holds:
$$\sup_{0 \leq j \leq N-1} |u_j| \leq C_{\reg} \norm{u}_{\reg} \ .$$
\end{remark}

For  $u, v \in \C^N$, we denote by  $u \cdot v$ the component-wise product of $u$ and $v$, namely the vector whose $j^{th}$ component is given by the product of the $j^{th}$ components of $u$ and $v$:
\begin{equation}
\label{comp.prod}
(u \cdot v)_j := u_j v_j, \qquad 0 \leq j \leq N-1 \ .
\end{equation}
We denote by $u \ast v$ the convolution product of $u$ and $v$, a vector whose $j^{th}$ component is defined by 
\begin{equation}
\label{conv.prod}
( u \ast v )_j := \sum_{k=0}^{N-1} u_k v_{j-k}, \qquad 0 \leq j \leq N-1 \ ,
\end{equation} 
where in the summation above $u$ and $v$ are extended periodically defining $v_{k + lN} \equiv v_k$ for $l \in \Z$.
The DFT maps the component-wise product in convolution:
\begin{lem}\label{bullet} For $s > \tfrac{1}{2}$ and  $\sigma \geq 0 $ there exists a constant $C_\reg >0$, independent of $N$, such that the following holds:
\begin{enumerate}[(i)]
\item  $\widehat{u \cdot v} = \frac{1}{\sqrt{N}}\; \hat{u} \ast \hat{v}$;
\item $\norm{u \cdot v}_{\reg} \leq C_\reg \norm{u}_{\reg} \norm{v}_{\reg}$;
\item the map $X: u \mapsto u^2$, has bounded modulus w.r.t.  the exponentials, and  $\norm{\und{X }(u)}_{\reg} \leq C_\reg \norm{u}_{\reg}^2.$
\end{enumerate}
\end{lem}
\begin{proof}
Item $(i)$ is standard and the details of the proof are omitted.

We prove now item $(ii)$. 
To begin, note that, by periodicity, one has
$$
\norm{u}_{\reg}^2 = \frac{1}{N} \sum_{k \in K_N^0} [k]^{2s} e^{2 \sigma |k|} \mmod{\hat{u}_k}^2 \ ,
$$
where the set 
\begin{equation}
\label{index.set}
K_N^0 := \left\{ k \in \Z: -(N-1)/2 \leq k \leq (N-1)/2 \right\} \cup \{ \lfloor N /2 \rfloor \},
\end{equation}
 while $[k]:= \max(1, |k \bmod{N} |)$.
By  item $(i)$, one has that
\begin{equation}
\label{eq:B.1}
\norm{ u \cdot v}_\reg^2 = \frac{1}{N} \sum_{k \in K_N^0} [k]^{2s} e^{2 \sigma |k|} |\widehat{(u \cdot v)}_k|^2 = \frac{1}{N^2} \sum_{k \in K_N^0} [k]^{2s} e^{2 \sigma |k| } \mmod{ \sum_{l=0}^{N-1} \hat{u}_l \hat{v}_{k-l} }^2.
\end{equation}
Introduce now the quantities
$$
\gamma_{k,l}:= \frac{[k]^s}{[l]^s \, [k-l]^s} \cdot \frac{e^{\sigma |k|}}{e^{\sigma |l|} e^{\sigma |k-l|}} \ .
$$
 For $s > \tfrac{1}{2}$ and  $\sigma \geq 0 $,  it holds that
 $
 \gamma_{k,l}^2 \leq 4^{s} \frac{([k-l]^{2s} + [l]^{2s})\,  e^{2 \sigma (|k-l| + |l|)} }{[k-l]^{2s} \, [l]^{2s}\, e^{2\sigma |l|}\, e^{2\sigma |k-l|}} \leq 4^{s} \left( \frac{1}{[l]^{2s}}  + \frac{1}{[k-l]^{2s}} \right) , 
 $
 from which it follows that there exists a constant $C_\reg >0$, independent of $N$, such that 
 \begin{equation}
 \label{gamma.est}
\sup_{0 \leq k \leq N-1} \sum_{l=0}^{N-1} \gamma_{k,l}^2 \leq C_\reg^2 \ .
 \end{equation}
 By Cauchy-Schwartz one has
\begin{align*}
[k]^{s} e^{ \sigma |k| }  \sum_{l=0}^{N-1} |\hat{u}_l|\, |\hat{v}_{k-l}| & 
=  \sum_{l=0}^{N-1} \gamma_{k,l}\, [l]^s\, e^{ \sigma |l|} \, |\hat{u}_l|\, [k-l]^s\, e^{ \sigma |k-l|} \, |\hat{v}_{k-l}| \\
& \leq \left( \sum_{l=0}^{N-1} \gamma_{k,l}^2 \right)^{1/2} \, \left( \sum_{l=0}^{N-1} [l]^{2s}\, e^{2 \sigma |l|} \, |\hat{u}_l|^2 \, [k-l]^{2s}\, e^{2 \sigma |k-l|} \, |\hat{v}_{k-l}|^2\right)^{1/2}.
\end{align*}
Inserting the inequality above in \eqref{eq:B.1}, one has 
\begin{align*}
\norm{ u \cdot v}_\reg & \leq \frac{C_\reg}{N} \left(\sum_{l=0}^{N-1}[l]^{2s}\, e^{2 \sigma |l|} \, |\hat{u}_l|^2 \right)^{1/2} \, \left( \sum_{k=0}^{N-1}[k-l]^{2s}\, e^{2 \sigma |k-l|} \, |\hat{v}_{k-l}|^2 \right)^{1/2} \\
& \leq C_\reg \norm{u}_\reg \, \norm{v}_\reg.
\end{align*}
We  prove now item $(iii)$.
 Consider $\widehat{X}:= \F X \F^{-1}$. By item $(i)$ one has  $\widehat{X}: \{\hat{u}_j\}_{j \in \Z} \mapsto \{\tfrac{1}{\sqrt{N}} \sum_l \hat{u}_l \hat{u}_{j-l} \}_{j \in \Z}$. Thus $\und{\widehat{X}} \equiv \widehat{X}$ and the claim follows.
\end{proof}

\begin{remark}
\label{lem:shift}
Let $S_\pm$ be the difference operators defined in \eqref{diff.oper}. Let $\hat{\omega}_\pm$ be the vectors whose $k^{th}$ components are given by $\hat{\omega}_{\pm,k}:= 1 -e^{\mp 2\pi \im k/N}$. Then the following holds:
\begin{enumerate}[(i)]
\item the map $\widehat{S}_\pm:= \F S_\pm \F^{-1}$ is  a multiplication by the vector $\hat{\omega}_\pm$: $\widehat{S}_\pm: \hat{u} \mapsto \hat{\omega}_\pm \cdot  \hat{u}$.
\item  $\mmod{\und{\widehat{S}_\pm}(\hat u)}  \leq \omega \cdot |\hat u|$, where   $\omega \equiv \lbrace \om{k} \rbrace_{k=1}^{N-1}$ is the vector of the linear frequencies. 
\end{enumerate}
\end{remark} 

\begin{rem}
\label{rem:s.ext} 
Consider $q=q(\xi, \eta)$ as a function of the linear Birkhoff variables defined in \eqref{xi_variable}. Then one has
$ \norm{\und{S_\pm}(q)}_\reg \leq \norm{(\xi, \eta)}_{\spazio{\reg}}$.
\end{rem}

\section{Proof of Proposition \ref{prop:Thata_na}}
\label{Theta_na}
We prove now property $(\Theta 1)$.   
Let $T: (\xi, \eta) \mapsto
 (p,q)$ be the map introducing linear Birkhoff coordinated. Explicitly $(p, q) = T(\xi, \eta)$ iff $( \hat{p}_0, \, \hat{q}_0)  = (0, 0)$ and 
 $$
( \hat{p}_k, \, \hat{q}_k ) = \left(  \sqrt{\frac{1}{2}\om{k}} \ (\xi_k + \eta_{N-k} ), \  \frac{1}{\im \sqrt{2 \om{k}}} (\xi_k - \eta_{N-k} ) \right), \ 1 \leq k \leq N-1 \ .
 $$

 Then  $\Theta_\Xi \equiv \Theta \circ T$  and in particular $d \Theta_\Xi(0,0) = d\Theta(0,0) T$. Using the formula above  and the fact that 
$d \Theta(0,0)(P,Q) = \left(-P, \tfrac{1}{2}S_+(Q) \right)$, where $S_+$ is defined in \eqref{diff.oper}, one obtains easily formula \eqref{Theta.lin}.
The estimate of $\norm{\und{d\Theta_\Xi}
(0,0)}_{\L(\spazio{\reg},\, \Cs^{\reg} )}$ is trivial, and is omitted. 

We prove now the estimate for 
$ \norm{\und{d\Theta_\Xi}
(0,0)^{*}}_{\L(\Cs^{s+1, \sigma},  \, \spazio{\reg})}$. 
Using the explicit formula \eqref{Theta.lin}, one computes that  $(\xi, \eta ) =  d\Theta_\Xi(0,0)^* (B,A)$ iff 
$$
\left( \xi_k, \ \eta_k \right) = \left(-\sqrt{\tfrac{1}{2}\om{k}}\hat{B}_k + \frac{\varpi_k}{\im \sqrt{2 \om{k}}} \hat{A}_k, \ -\sqrt{\tfrac{1}{2}\om{k}}\hat{B}_{N-k} - \frac{\varpi_k}{\im \sqrt{2 \om{k}}} \hat{A}_{N-k}  \right) 
$$
for $1 \leq k \leq N-1$. 
Thus there exist  constants $ C, C_{\Theta_2} >0$, independent of $N$, such that
$$
\norm{\und{d\Theta_\Xi}
(0,0)^{*}(B, A)}_{\spazio{\reg}} \leq  C \left(\frac{1}{N} \sum_{k=1}^{N-1} [k]_N^{2s} e^{2\sigma [k]_N} \om{k}^2 (|\widehat{B}_k|^2 + |\widehat{A}_k|^2)\right)^{1/2} \leq \frac{C_{\Theta_2}}{N} \norm{(B,A)}_{\Cs^{s+1, \sigma}} \ ,
$$
where we used that $\mmod{\om{k}}^2 \leq \frac{\pi^2 [k]_N^2}{N^2}$. Thus the second of \eqref{dtheta.est} is proved.
\vspace{0.5em}\\

We prove now property $(\Theta 2)$. Denote by $\Theta_b$ the map $ p \mapsto -p$ and by $\Theta_a$ the map $ q \mapsto \exp\left(\tfrac{1}{2} S_+(q)\right)-1$. Then 
 $(b,a) = \Theta(p, q) \equiv \left(\Theta_b(p), \Theta_a(q)\right)$.  Introduce on $\C^N$ the norm $\norm{\cdot}_\reg$ defined in \eqref{norm.reg}.  Then $\norm{\Theta(p,q)}_{\Cs^\reg}^2 \equiv \norm{\Theta_b(p)}_{\reg}^2 + \norm{\Theta_a(q)}_{\reg} ^2$.
The analyticity of  $p \mapsto \Theta_b(p)$ is obvious. Consider now the map $q \mapsto \Theta_a(q)$. 
Expand $\Theta_a$  in Taylor series with center at the origin to get 
\begin{equation}
\label{theta.taylor}
\Theta_{a}(q) = \sum_{r \geq 1}\Theta_{a}^r(q), \qquad \Theta^r_a (q):=  \frac{1}{ r!\, 2^r} (S_+(q))^r, \qquad \forall r \geq 1.
\end{equation}
Consider $q$ as a function of the linear Birkhoff variables $\xi, \eta$. 
Then  Lemma \ref{bullet} and Remark \ref{rem:s.ext} imply that for any $s \geq 0$, $\sigma \geq 0$
\begin{equation}
\label{exp.theta.fou}
\norm{ \und{\Theta_a^r} (q)}_{s+1,\sigma} \leq C_1^r \norm{\und{S_+}(q)}_{s+1, \sigma}^r \leq C_2^r \norm{(\xi, \eta)}_{\spazio{s+1, \sigma}}^r \leq C_3^r N^r \norm{(\xi, \eta)}_{\spazio{\reg}}^r \ , \  \forall r \geq 2 ,
\end{equation}
where $C_1, C_2, C_3 >0$ are positive constants independent of $N$.
 Therefore for $\epsilon < \frac{1}{C_3}$ one has
\begin{align*}
  \sup_{ \norm{(\xi, \eta)}_{\spazio{\reg}} \leq \epsilon/N^2} \norm{\und{\Theta^0_\Xi}(\xi,\eta)}_{\Cs^{s+1, \sigma}}\leq\sum_{r \geq 2} \sup_{ \norm{(\xi, \eta)}_{\spazio{\reg}} \leq \epsilon/N^2} \norm{ \und{\Theta_\Xi^r}(\xi, \eta)}_{\Cs^{s+1, \sigma}}  \leq \sum_{r \geq 2} C^r_3 N^r \frac{\epsilon^r}{N^{2r}} \leq  \frac{2 C_3^2 \epsilon^2}{N^2} \ .
 \end{align*}
 This proves the first estimate in $(\Theta 2)$.  We show now that for any $s \geq 0$, $\sigma \geq 0$ one has $[d
   \Theta_\Xi^0]^* \in \Nc_{\epsilon/N^2}(\spazio{\reg}, \L(\Cs^{s+2,
   \sigma}, \spazio{s+1,\sigma}) ) $. Note that $d\Theta_\Xi(\xi, \eta)^* =
 T^* d\Theta(T(\xi, \eta))^*$. Using the explicit expression of $T$, one verifies that
 $(\xi, \eta) = T^*(P, Q)$ iff
\begin{equation}
\label{T*.formula}
( \xi_k, \, \eta_k ) = \left(\sqrt{\tfrac{1}{2}\om{k}}\widehat{P}_k + \frac{1}{\im \sqrt{2 \om{k}}} \widehat{Q}_k, \ \sqrt{\tfrac{1}{2}\om{k}}\widehat{P}_{N-k} - \frac{1}{\im \sqrt{2 \om{k}}} \widehat{Q}_{N-k},  \right) 
\end{equation}
for $1 \leq k \leq N-1$. Thus one has that for any $s \geq 0$, $\sigma \geq 0$
\begin{equation}
\label{T*.est}
\norm{T^*( 0, Q)}_{\spazio{\reg}} \leq \norm{Q}_{\reg} \ .
\end{equation}
  Using \eqref{theta.taylor} one
 verifies that $d\Theta^r(p,q) (P, Q) =\frac{1}{(r-1)!\, 2^r} \left(0,
 \, S_+(q)^{r-1 }\cdot S_+(Q)\right)\,$, $\forall r \geq 2$, from
 which it follows that
 $$
 d\Theta^r(p,q)^*(B,A) = \frac{1}{(r-1)!\, 2^r}\left(0, \  \overline{S_+(q)}^{r-1 } \cdot S_-(A)\right) \ , \qquad \forall r \geq 2 \ .
 $$
 Thus, using  estimate \eqref{T*.est},  
there exists a constant $C_4 >0$, independent of $N$, such that
 \begin{align*}
  \norm{\und{ d \Theta_\Xi^r}(\xi, \eta)^*(B,A)}_{\spazio{s+1,\sigma}} \leq  C_4^{r} \norm{\und{S_+}(\und{q}(\xi, \eta))}_{s+1,\sigma }^{r-1} \norm{\und{S_-}(A)}_{s+1,\sigma} \leq  C_4^{r} N^{r-2} \norm{(\xi, \eta)}_{\spazio{\reg}}^{r-1} \norm{(B, A)}_{\Cs^{s+2,\sigma}} \ . 
 \end{align*}
Then there exists  $C_5, \epsilon_0 > 0$, independent of $N$, such that $\forall \, 0 < \epsilon \leq \epsilon_0$
\begin{align*}
\sup_{\norm{(\xi, \eta)}_{\spazio{\reg}}   \leq \epsilon/N^2}\norm{\und{d\Theta_\Xi^0}
(\xi , \eta)^{*}}_{\L(\Cs^{s+2,\sigma}, \,  \spazio{s+1,\sigma})} 
& \leq 
\sum_{ r \geq 2} \sup_{\norm{(\xi, \eta)}_{\spazio{\reg}}\leq \epsilon/N^2} \norm{\und{ d \Theta_\Xi^r}(\xi, \eta)^*}_{\L(\Cs^{s+2,\sigma}, \, \spazio{s+1,\sigma})} \\
&  \leq  \sum_{ r \geq 2} C_4^r N^{r-2}\frac{\epsilon^{r-1}}{N^{2(r-1)}} \leq 
 \frac{C_5\epsilon}{N^2}.
\end{align*}

\section{Proof of Lemma \ref{lem:spec.prop1} and Corollary \ref{lem:spec.prop}}
\label{spectral.object}
{\em Proof of Lemma \ref{lem:spec.prop1}.}  Since the map $(b,a) \mapsto L_p(b,a)$ is linear, it is enough to prove that it is continuous from $\Cs^\reg$ to $\L(\C^{2N})$. In particular we will prove that
\begin{equation}
\label{L_p_esti_sup}
\norm{L_p}_{\L(\C^{2N})}\leq \sup_{0 \leq j \leq N-1} \left( |b_j| + 2\sup_j|a_j|\right).
\end{equation}
This estimate, together with Lemma \ref{lem:normaNs}, proves \eqref{L_p_esti}. In order to prove \eqref{L_p_esti_sup}, write  $L_p=D+A^+ +A^-$, where $D$ is the diagonal part of $L_p$ and $A^{\pm}$ are defined by
\begin{equation*}
	A^+= \left( 
	\begin{array}{cccc}
	0 & a_{0} &  &  \\
	 & 0 & \ddots &  \\
   &  & 0 & a_{N-1}\\
	a_{N-1} &  &  & 0
	\end{array}
	\right),	
	\ \ 
	A^-=\left( 
	\begin{array}{cccc}
	0 &  &  & a_{N-1}\\
	a_{0} & 0 &  &   \\
   & \ddots & 0 &   \\
   &  & a_{N-2} & 0
	\end{array}
	\right).
\end{equation*}
To estimate the norms of $D, A^+$ and $A^-$ is enough to observe that for every $x \in \C^{2N}$ one has
\begin{equation*}
\norm{Dx}_{\C^{2N}}^2 :=\sum_{j=0}^{2N-1}{\left|b_j x_j\right|^2}\leq \left( \sup_{0 \leq j \leq N-1}{|b_j|}\right)^2\norm{x}_{\C^{2N}}^2, \qquad 
\norm{A^\pm x}_{\C^{2N}}^2\leq \left(\sup_{0 \leq j \leq N-1}{|a_j|}\right)^2\norm{x}_{\C^{2N}}^2,
\end{equation*}
where $\norm{\cdot}_{\C^{2N}}$ is the standard euclidean norm on $\C^{2N}$. 
Thus \eqref{L_p_esti_sup} follows.
\qed
\vspace{1em}\newline
{\em Proof of Corollary \ref{lem:spec.prop}.}
 Item $(i)$ follows by  standard perturbation theory, and the details are omitted. 
 We prove now item $(ii)$. 
Let $\Gamma_j$ be the circle  defined by  $\Gamma_j:=\left\{\lambda \in \mathbb{C}:\left|\lambda_{2j}^0-\lambda \right|=\frac{1}{2N^2} \right\},$ counter-clockwise oriented. By item $(i)$, for any  $\norm{(b,a)}_{\Cs^\reg} \leq \tfrac{\epsilon_*}{N}$, $\lambda_{2j}(b,a)$ and $\lambda_{2j-1}(b,a)$ are inside the ball enclosed by $\Gamma_j$. 
 Write $L_{b,a} - \lambda = L_0 - \lambda + L_p = \left(L_0-\lambda\right)
\left(1+\left(L_0-\lambda \right)^{-1}L_p\right)$; its  inverse 
\begin{equation}
\label{L-lambda}
\left(L_{b,a} - \lambda\right)^{-1}=
\Bigg(\sum_{n=0}^{\infty}{\Big(-\left(L_0- \lambda\right)^{-1} L_p\Big)^n}\Bigg)\left(L_0-\lambda\right)^{-1}
\end{equation}
is well defined as a Neumann operator when $\norm{\left(L_0- \lambda\right)^{-1} L_p}_{\L(\C^{2N})} < 1$.
Since $L_0-\lambda$ is diagonalizable with $\lbrace (\lambda_{j}^0 - \lambda) \rbrace_{0 \leq j \leq 2N-1}$ as eigenvalues,  the norm of its inverse is bounded by the inverse of the smallest eigenvalue:
\begin{equation}
\sup_{\lambda \in \Gamma_j} \norm{\left(L_0-\lambda\right)^{-1}}_{\L(\C^{2N})}\leq \sup_{\substack{\lambda \in \Gamma_j \\ 0 \leq k \leq 2N-1}}{\left|\frac{1}{\lambda_{k}^0-\lambda}\right|}< 2N^2
\label{stimaL0-l}
\end{equation}
where the last estimates is due to the form of $\Gamma_j$.
Therefore for $0< \epsilon \leq \epsilon_*$ and $\norm{(b,a)}_{\Cs^\reg} < \eN$ one gets, using \eqref{L_p_esti},
$$
\norm{\left(L_0-\lambda\right)^{-1} L_p}_{\L(\C^{2N})} \leq \norm{L_p}_{\L(\C^{2N})} \norm{\left(L_0-\lambda\right)^{-1}}_{\L(\C^{2N})} \leq C_\reg \norm{(b,a)}_{\Cs^\reg} 2N^2 < 2C_\reg \epsilon_*,
$$
which proves the convergence of the Neumann series \eqref{L-lambda} for $\epsilon_* \leq \tfrac{1}{2 C_\reg}$.
\newline
Substituting  \eqref{L-lambda}   in \eqref{P}  we get, for $1 \leq j \leq N-1$,
\begin{equation}
P_j(b,a)=P_{j0} -\frac{1}{2\pi \im }\oint_{\Gamma_j}{\left(\sum_{n=1}^{\infty}{\left(-\left(L_0- \lambda\right)^{-1} L_p\right)^n}\right)\left(L_0-\lambda\right)^{-1} \diff{\lambda}}.
\label{sviluppo_proiettore}
\end{equation}
Since the series inside the integral is absolutely and uniformly convergent for $(b,a) \in B^{\Cs^\reg}\left(\eN\right)\,$,  $(b,a) \mapsto P_j(b,a)$ is analytic as a map from $B^{\Cs^\reg}\left(\eN\right)$ to $\L(\C^{2N})$.
Estimate \eqref{P_j-P_j0} follows easily from \eqref{sviluppo_proiettore}.

We prove now item $(iii)$.  Properties $(U1)-(U3)$ are standard
\cite{kato}. The analyticity of the map $(b,a)\mapsto U_j(b,a)$
follows from item $(ii)$. Indeed, in order for $U_j(b,a)$ to be
defined as a Neumann series one needs
$\norm{P_j(b,a)-P_{j0}}_{\L(\C^{2N})}<1$, which follows from
\eqref{P_j-P_j0}. Estimate \eqref{U_j-P_j} follows by expanding
\eqref{defUj} in power series of $P_j(b,a) - P_{j0}$.  \qed

\section{Proof of Proposition \ref{propZn}}
\label{proofZ}
Denote by $D: \C^{N-1} \to \C^{N-1}$ the diagonal operator
\begin{equation}
\label{D.def}
D: \{ \xi_j\}_{1 \leq j \leq N-1} \mapsto \{ D_j \xi_j\}_{1 \leq j \leq N-1}, \quad \mbox{where} \quad D_j := \left(\tfrac{2}{N}\om{j}\right)^{-1/2} \ .
\end{equation}
{\em Proof of properties $(Z1)-(Z3)$.} Property $(Z1)$  follows from formula \eqref{def:zj}, since
 \footnote{to simplify the notation, we write  $f_{j}\equiv f_j(b,a)$ and $U_j \equiv U_j(b,a)$}: 
\begin{align*}
\overline{z_j(b,a)}&=D_j \overline{\left\langle \left(L_{b,a}-\lambda_{2j}^0\right)U_{j}
f_{2j,0},\overline{U_{j}f_{2j,0}}\right\rangle}= D_j
\left\langle \overline{U_{j}f_{2j,0}}, \left(L_{b,a}-\lambda_{2j}^0\right)U_{j}f
_{2j,0}\right\rangle= \\
&= D_j \left\langle U_{j}f_{2j-1,0}, \left(L_{b,a}-\lambda_{2j}^0\right)
\overline{U_{j}f_{2j-1,0}}\right\rangle
=D_j \left\langle \left(L_{b,a}-\lambda_{2j}^0\right)
f_{2j-1},\overline{f_{2j-1}}\right\rangle=w_{j}(b,a).
\end{align*}
We  prove now $(Z2)$. Using Lemma \ref{lem:spec.prop} $(iv)$ and the fact that $\overline{f_{2j,0}}=f_{2j-1,0}$, decompose $f_{2j,0}$ and $f_{2j}$ in real and imaginary part:
\begin{align*}
&  f_{2j,0}=e_{j,0}+\im h_{j,0}, \quad f_{2j}=e_{j}+\im h_{j}\\
& f_{2j-1,0}=e_{j,0}-\im h_{j,0}, \quad f_{2j-1}=e_{j}-\im h_{j} , 
\end{align*}
where
\begin{align*}
&e_{j,0}:=Re \, f_{2j,0}, \quad h_{j,0}:=Im \, f_{2j,0}, \qquad \mbox{ and } \qquad e_{j}:=Re\,f_{2j}=U_{j}e_{j,0}, \quad h_{j}:=Im\, f_{2j}\,=U_{j}h_{j,0}.
\end{align*}
The vectors $\{e_j, h_{j}\}$ form a real  orthogonal basis for $E_{j}(b,a)$. Let $M_{j}(b,a)$ be the matrix of
 the selfadjoint operator $\left.L_{b,a}-\lambda_{2j}^0\right|_{E_{j}(b,a)}$ with respect to this basis:
$$M_{j}(b,a)=\left( 
	\begin{array}{cc}
	\alpha_j & \sigma_j\\
	\sigma_j & \beta_j \\
	\end{array}
	\right).
$$
The eigenvalues of  $M_{j}$ are obviously $\lambda_{2j}-\lambda_{2j}^0$ and $\lambda_{2j-1}-\lambda_{2j}^0$, hence
\begin{align*}
&\mbox{Tr}\; M_{j}= \alpha_j + \beta_j= \left(\lambda_{2j}-\lambda_{2j}^0\right) + \left(\lambda_{2j-1}-\lambda_{2j}^0\right),\\ 
&\mbox{Det}\; M_{j}=\alpha_j \beta_j - \sigma_j^2=\left(\lambda_{2j}-\lambda_{2j}^0\right)\left(\lambda_{2j-1}-\lambda_{2j}^0\right).
\end{align*}
Now observe that 
\begin{align*}
z_j(b,a)&=D_j \left\langle  \left(L_{b,a}-\lambda_{2j}^0\right)(e_j + \im h_{j}), (e_j-\im h_{j})\right\rangle=\\
        &=D_j \left\langle  \left(L_{b,a}-\lambda_{2j}^0\right)e_j ,e_j\right\rangle - 
           D_j\left\langle  \left(L_{b,a}-\lambda_{2j}^0\right) h_{j}, h_{j}\right\rangle +
           2\im D_j\left\langle  \left(L_{b,a}-\lambda_{2j}^0\right)e_j, h_{j}\right\rangle=\\
        &=\left(\tfrac{2}{N}\om{j}\right)^{-1/2} (\alpha_j - \beta_j +\im 2\sigma_j).
\end{align*}
Finally one computes
\begin{align*}
\left(\lambda_{2j}-\lambda_{2j-1}\right)^2
&= (\mbox{Tr } M_j)^2 - 4 \mbox{Det }M_j = 
\left(\alpha_j+\beta_j\right)^2-4\alpha_j\beta_j+
4\sigma_j^2 \\ &=\left(\alpha_j-\beta_j\right)^2+4\sigma_j^2
=\left(\mbox{Re}\;z_j\right)^2+  \left(\mbox{Im}\;z_j\right)^2=\left(\tfrac{2}{N}\om{j}\right)|z_j(b,a)|^2.
\end{align*}
We prove now $(Z3)$. The first order terms of $z_j$ and $w_j$ in $(b,\; a)$ are given by
$$dz_j(0,0)(b,a)=D_j \left\langle L_p f_{2j,0},\,\overline{f_{2j,0}} \right\rangle,\quad 
dw_j(0,0)(b,a)=D_j\left\langle L_p f_{2j-1,0},\,\overline{f_{2j-1,0}} \right\rangle, \qquad 1 \leq j \leq N-1 \ .
$$
Using the explicit formula  for  $f_{2j,0}$  in   Lemma \ref{lem:spettroL0}, one computes
 \begin{equation}
\begin{aligned}
\label{Lp.coeff.1}
\left\langle L_p f_{2j,0},\overline{f_{2j,0}} \right\rangle
&=\frac{1}{2N}\sum_{l=0}^{2N-1}{b_l e^{\im 2\rho_j l} + a_{l-1}e^{\im 2\rho_j (l-1)}e^{\im \rho_j} + a_l e^{\im 2\rho_j l}e^{\im \rho_j}}\\
&=\frac{1}{2N}\sum_{l=0}^{2N-1}{b_l e^{\im 2\pi jl/N} + a_{l-1}e^{\im 2\pi (l-1)j/N}e^{\im \rho_j} + a_l e^{\im 2\pi lj/N}e^{\im \rho_j}}\\
&=\frac{1}{\sqrt{N}}\left(\hat{b}_j + 2e^{\im \rho_j}\hat{a}_j \right)=\frac{1}{\sqrt{N}}\left(\hat{b}_j - 2e^{\im \pi j/N}\hat{a}_j \right).
\end{aligned}
\end{equation}
The formula for $dz_j(0,0)(b,a)$ immediately follows. The one for
$dw_j(0,0)(b,a)$ is proved in the same way and the details are
omitted.

 The estimate \eqref{dZ(0).estimate} for $\und{dZ}(0,0)$
follows immediately. We estimate now the norm of  $\und{dZ}(0,0)^*$. One checks that $(B,A) = dZ(0,0)^*(\xi,\eta) $ iff $\widehat{B}_0 = \widehat{A}_0 = 0$ and for $1 \leq k \leq N-1$
$$
(\widehat{B}_k, \widehat{A}_k) = \left( \frac{1}{\sqrt{2 \om{k}}} (\xi_k + \eta_{N-k}),\  
\frac{2}{\sqrt{2 \om{k}}}  (e^{\im \pi k/N} \xi_k + e^{\im \pi (N-k)/N} \eta_{N-k}) \right)   \ .
$$
Thus there exist  constants $C, C', C_Z >0$, independent of $N$, such that
\begin{align*}
\norm{\und{dZ}(0,0)^*(\xi,\eta)}_{\Cs^{s+2, \sigma}}^2 
& \leq \frac{C'}{N} \sum_{k = 1}^{N-1} [k]_N^{2s} \, e^{2 \sigma [k]_N}\om{k} \frac{[k]_N^4}{\om{k}^2} \left(|\xi_k|^2 + |\eta_k|^2 \right) \leq  C_Z^2 N^4 \norm{(\xi,\eta)}_{\spazio{\reg}}^2 
\end{align*}
where in the last inequality we used that $[k]_N^4 / \om{k}^2 \leq C'' N^4$ for some constant $C'' > 0$ independent of $N$. Thus the second of \eqref{dZ(0).estimate} is proved.
\vspace{1em}\\
{\em Proof of property $(Z4)$.}  We will prove
that $Z$ is normally analytic. Recall that, as mentioned in the
discussion before Proposition \ref{prop:Thata_na}, the map $Z$ is said to be
normally analytic if $\check{Z} := Z \F$ is normally analytic. With an
abuse of notations, we omit the ``check'' from $Z$.

We begin by expanding the components of $Z$, denoted by $Z_j(b,a) :=
(z_j(b,a),\, w_j(b,a))$, in Taylor series with center at $(b,a) =
(0,0)$.  The first two terms of the expansions are given by
\begin{equation}
\begin{aligned}
\label{taylor.exp}
&z_j(b,a) = D_j\langle L_p f_{2j,0}, \overline{f_{2j,0}} \rangle + D_j\langle L_p \left(L_0 - \lambda_{2j}^0 \right)^{-1} \left(\uno - P_{j0}\right) L_p f_{2j,0},  \overline{f_{2j,0}}\rangle + O((b,a)^3), \\
&w_j(b,a) = D_j\langle L_p f_{2j-1,0}, \overline{f_{2j-1,0}} \rangle + D_j\langle L_p \left(L_0 - \lambda_{2j}^0 \right)^{-1} \left(\uno - P_{j0}\right) L_p f_{2j-1,0},  \overline{f_{2j-1,0}}\rangle + O((b,a)^3).
\end{aligned}
\end{equation}
To perform the Taylor expansion at every order it is convenient to proceed in the following way.
Write $z_j(b,a)=z_{j,1}(b,a)+z_{j,2}(b,a)$ and $w_j(b,a) = w_{j,1}(b,a) + w_{j,2}(b,a)$ where
\begin{equation}
  \begin{split}
	  &z_{j,1}(b,a)=D_j\left\langle \left(L_0-\lambda_{2j}^0\right) f_{2j}(b,a),\overline{f_{2j}(b,a)}\right\rangle, \quad 
	  z_{j,2}(b,a)=D_j\left\langle L_p f_{2j}(b,a),\overline{f_{2j}(b,a)}\right\rangle \ ,
  \end{split}
	\label{zj1ezj2}
\end{equation} 
while $w_{j,1}(b,a)$ and $w_{j,2}(b,a) $ are defined as in
\eqref{zj1ezj2}, but with $f_{2j-1}(b,a)$ replacing
$f_{2j}(b,a)$. 

Expand $z_{j,\varsigma}(b,a)$, $\varsigma=1,2$, in
Taylor series with center at $(b,a)=(0,0)$: $z_{j,\varsigma}(b,a) =
\sum_{n \geq 1} z_{j,\varsigma}^n(b,a)$, with $z_{j,\varsigma}^n$ a
homogeneous polynomial of degree $n$ in $b, a$. We write an analogous
expansion for $w_{j,\varsigma}(b,a)$.  Therefore one has 
$$Z_j^n(b,a) := (z_j^n(b,a),\, w_j^n(b,a)) \equiv \left(z_{j,1}^n(b,a)
+ z_{j,2}^n(b,a), \,w_{j,1}^n(b,a) + w_{j,2}^n(b,a)\right).$$ In order
to write explicitly $z_{j,\varsigma}^n(b,a)$ as a function of $b$ and
$a$, one needs to expand the vectors $f_{2j}(b,a)$ and $f_{2j-1}(b,a)$
in Taylor series of $b,\,a$.  Rewrite \eqref{defUj}, \eqref{def:fj} as
$$f_{2j}(b,a)=U_{j}(b,a)f_{2j,0}=\Big(\uno-\left(P_{j}(b,a)-P_{j0}\right)^2\Big)^{-1/2}\Big(\uno+(P_{j}(b,a)-P_{j0})\Big)f_{2j,0}$$
and expand the r.h.s. above in power series of $P_{j}(b,a)-P_{j0}$,
getting:
\begin{equation}
f_{2j}(b,a)=\sum_{m=0}^{\infty}c_m\left(P_{j}(b,a)-P_{j0}\right)^m f_{2j,0}, \qquad f_{2j-1}(b,a)=\sum_{m=0}^{\infty}c_m\left(P_{j}(b,a)-P_{j0}\right)^m f_{2j-1,0} \ ,
\label{espansione di fj}
\end{equation}
where the $c_m$'s are the coefficients of the Taylor series of the
 function $\phi(x)=\frac{1+x}{(1-x^2)^{1/2}}$. Note that  $c_{2k+1}  = c_{2k} \equiv (-1)^k \binom{-1/2}{k}$, where $\binom{-1/2}{k}:= -\tfrac{1}{2}(-\tfrac{1}{2} - 1) \cdots (-\tfrac{1}{2} - k + 1 )$ is the product of $k$ negative terms, thus   $(-1)^k \binom{-1/2}{k}\geq 0,\, $  $\forall k \geq 0$, and therefore $c_m \geq 0, \,$ $\forall m$.\\ 

 By  Corollary \ref{cor:spec.prop1} (see also formula \eqref{sviluppo_proiettore}) one has,   in the ball $B^{\Cs^\reg}(\epsilon_*/N^2)$,
\begin{equation}
P_{j}(b,a)-P_{j0}=\frac{\im }{2\pi }\sum_{n=1}^{\infty}(-1)^n\oint_{\Gamma_j}{T^n(b,a,\lambda)\left(L_0-\lambda\right)^{-1}\diff{\lambda}} 
\label{espansionePj}
\end{equation}
where the $\Gamma_j$'s are defined as in equation \eqref{P},  and 
$$ T(b,a,\lambda):=\left(L_0-\lambda\right)^{-1}L_p \ .$$
Substituting \eqref{espansionePj}  in \eqref{espansione di fj} we get that
\begin{equation}
\label{expressionfj}
\begin{aligned}
&f_{2j}(b,a)=f_{2j,0}+\sum_{n\geq 1}\sum_{1\leq m \leq n}c_m \sum_{\multindex{\alpha}{m}{n}}f_{2j,m}^\alpha(b,a),\\
&f_{2j,m}^\alpha(b,a):=\\
&\left(\frac{\im}{2\pi}\right)^m (-1)^{|\alpha|} \oint_{\Gamma_j}\ldots \oint_{\Gamma_j}T^{\alpha_1}(b,a,\lambda_1)\left(L_0-\lambda_1\right)^{-1}\ldots T^{\alpha_m}(b,a,\lambda_m)\left(L_0-\lambda_m
\right)^{-1}f_{2j,0}\diff{\lambda_1}
\ldots \diff{\lambda_m}. 
\end{aligned}
\end{equation}
An analogous expansion holds for $f_{2j-1}(b,a)$, with $f_{2j-1,0}$
substituting $f_{2j,0}$ in the integral formula above. In order to
write explicitly the expression inside the integral, one needs to
compute the iterated terms $T^n(b,a,\lambda) f_{2j,0}$ and
$T^n(b,a,\lambda) f_{2j-1,0}$. The computation turns out to be simpler
if we express $L_pf_{2j,0}$ in the basis of the eigenvectors of
$L_0$. To simplify the notations we relabel the eigenvectors of $L_0$
in the following way:
$$g_{0} := f_{00}, \quad g_{N}:=f_{2N-1,0}, \qquad g_{j} := f_{2j,0},
\quad g_{-j}:= f_{2j-1,0}, \quad \mbox{ for } 1 \leq j \leq N-1$$ and
the eigenvalues of $L_0$ as
$$\llambda_0 := \lambda_{0}^0, \quad \llambda_{N}:= \lambda_{2N-1}^0,
\qquad \llambda_j := \lambda_{2j}^0, \quad \llambda_{-j} :=
\lambda_{2j-1}^0, \quad \mbox{ for } 1 \leq j \leq N-1.$$ For every $1
\leq j \leq N-1$ one has that $\overline{g_j} = g_{-j}$, formally, one
can also write $g_{j+2N} = g_j$, $\llambda_j = \llambda_{-j} $ and
$\llambda_{j+2N} = \llambda_j$, as one verifies using the explicit
expressions of the $g_j$'s and $\llambda_j$'s.  In this notation, for
$\lambda\neq \llambda_{\pm j}$, one has $(L_0-\lambda)^{-1} g_{\pm j}=
g_{\pm j}/(\llambda_{\pm j}-\lambda)$. 
With a computation analogous to the one in \eqref{Lp.coeff.1} (using also the second formula in \eqref{lemfourier}), one verifies that the projection of $L_p g_{j}$
on the vector $g_{k}$ is given by
\begin{equation}
\label{coeff.xlj}
\langle L_p \,g_{j}, g_{k} \rangle = \frac{1}{\sqrt{N}} \left(\hat{b}_{\frac{j-k}{2}} -2 \cos\left( \tfrac{k \pi}{N} \right) \hat{a}_{\frac{j-k}{2}} \right) \delta_{(j-k; \mbox{ even })}, 
\end{equation}
where $\delta_{(j-k; \mbox{ even })} = 1$ if $j-k$ is an even integer,
and equals $0$ otherwise.  Formula \eqref{coeff.xlj} implies that $L_p
g_{j}$ is supported only on the vectors $g_{k}$ whose index $k$
satisfies $k=j-2l$ for some integer $l$. Therefore we can write
\begin{equation}
\label{x^l_j.formula}
T(b,a,\lambda)g_{j} = \sum_{l \in K_N^0 } \frac{x^l_j}{\llambda_{j-2l} - \lambda} g_{j-2l}, \quad 
x^l_j:=\langle L_p g_{j}, g_{j-2l}\rangle = 
\frac{1}{\sqrt{N}} \left(\hat{b}_l -2 \cos\left( \tfrac{(j-2l) \pi}{N} \right) \hat{a}_l \right) \ ,
\end{equation}
where $K_N^0$ is the set of indexes defined in \eqref{index.set}.
Note that  $\mmod{x^l_j} \leq \frac{2}{\sqrt{N}} \left( |\hat{b}_l| + |\hat{a}_l| \right)$ uniformly in $j$, and $x^{l+N}_j = x^l_j$.
Iterating \eqref{x^l_j.formula} one gets 
$$T^{n}(b,a,\lambda)\left(L_0-\lambda\right)^{-1}g_{j}=\sum_{i_1,\cdots,i_n \in K_N^0}\frac{x_j^{i_1}x_{j-2i_1}^{i_2}\ldots x_{j-2i_1-\cdots -2i_{n-1}}^{i_n}}{(\llambda_j-\lambda)
\prod_{l=1}^n
\left(\llambda_{j-2\sum_{m=1}^l i_m}-\lambda\right)}g_{j-2i_1-\cdots -2i_{n}}.$$
More generally, for a vector  $\alpha=(\alpha_1,\ldots,\alpha_m) \in \N^m$ with $|\alpha|=n$ and $\lambda_1, \cdots, \lambda_m \in \Gamma_j$, one has
\begin{equation}
\begin{aligned}
&T^{\alpha_m}(b,a,\lambda_m)\left(L_0-\lambda_m\right)^{-1}
\cdots T^{\alpha_1}(b,a,\lambda_1)\left(L_0-\lambda_1\right)^{-1}g_{j}=\\
&\quad = \sum_{i_1,\ldots,i_n \in K_N^0 }\frac{x_j^{i_1}x_{j-2i_1}^{i_2}\ldots x_{j-2i_1-\cdots - 2i_{n-1}}^{i_n}}{(\llambda_j-\lambda_1)
\prod_{l=1}^n\left(\llambda_{j-2\sum_{m=1}^l i_m}-\mu_l\right)
\prod_{l=1}^{m-1}\left(\llambda_{j- 2\sum_{h=1}^{\alpha_1 + \cdots + \alpha_l} i_h}
-\lambda_{l+1}\right)}g_{j- 2i_1-\cdots-2i_n}
\end{aligned} 
\label{expressionTalfa}
\end{equation}
where
\begin{equation}
\mu_l= \lambda_1  \mbox{ for } 1\leq l \leq \alpha_1, \quad \mbox{and} \quad  \mu_l = \lambda_k \mbox{ for } \sum_{h=1}^{k-1} \alpha_h + 1 \leq l \leq \sum_{h=1}^k \alpha_h, \quad 2 \leq k \leq m \ .
\label{defmul}
\end{equation}
To  obtain the explicit expression of  $z^n_{j,\varsigma}$ and $w^n_{j, \varsigma}$, $\varsigma= 1,2$, in terms of the Fourier variables $\hat{b}, \, \hat{a}$, we substitute \eqref{expressionTalfa} in \eqref{expressionfj} and the obtained result  in \eqref{zj1ezj2}. By  \eqref{expressionfj}, $z_{j,1}^n$  is a sum of terms of the form 
$\left\langle \left(L_0 - \lambda_{2j}^0\right)f_{2j,p_1}^{\alpha},\overline{f_{2j,p_2}^{\beta}}  \right\rangle $ 
over $(p, \alpha, \beta) \in \N^{2} \times \N^{p_1} \times \N^{p_2}$ with $|p| = p_1 + p_2 \leq n$ and $|\alpha| + |\beta| = n$.
For $|\alpha| = r$, $|\beta| = n-r$ one gets
\begin{equation}
\label{expressionL}
\begin{aligned}
&\left\langle \left(L_0 - \llambda_{j}\right)f_{2j,p_1}^{\alpha},\overline{f_{2j,p_2}^{\beta}}  \right\rangle = 
\left(\frac{\im}{2\pi}\right)^{|p|}(-1)^n \oint_{\Gamma_j}\ldots\oint_{\Gamma_j}\kappa_{j,1}^{p,\alpha,\beta}(i)
x_j^{i_1}x_{j-2i_1}^{i_2}\ldots x_{j-2i_1-\cdots - 2i_{r-1}}^{i_r} \times \\
& \quad \times x_j^{i_n}x_{j-2i_n}^{i_{n-1}}\ldots x_{j-2i_n-\cdots -2i_{r+2}}^{i_{r+1}} 
\left\langle g_{j-2i_1-\cdots-2i_r}, \overline{g_{j-2i_{r+1}-\cdots-2i_n}}\right\rangle
 \diff{\lambda_1}\ldots
  \diff{\lambda_{|p|}},
\end{aligned}
\end{equation}
where, writing $ \bi = (i_1, \cdots, i_n)$, 
\begin{equation}
\begin{aligned}
&\kappa_{j,1}^{p,\alpha,\beta}(\bi):=\frac{\left(\llambda_{j-2\sum_{m=1}^r i_m}-\llambda_j\right)}
{(\llambda_j-\lambda_1)\prod_{l=1}^{r}
\left(\llambda_{j-2\sum_{m=1}^l i_m}-\mu_l\right)\prod_{l=1}^{p_1-1}\left(\llambda_{j- 2\sum_{h=1}^{\alpha_1 + \cdots + \alpha_l} i_h}-\lambda_{l+1}\right)}\times \\
&\qquad \times
\frac{1}{(\llambda_j-\lambda_{p_1+1})\prod_{l=r+1}^{n}\left(\llambda_{j-2\sum_{m=l}^n i_m}-\tilde{\mu}_l\right)\prod_{l=1}^{p_2-1}
\left(\llambda_{j- 2\sum_{h=1}^{\beta_1 + \cdots + \beta_l} i_h}- \lambda_{l+1}\right)}
\end{aligned}
\label{kappa}
\end{equation}
and the $\tilde{\mu}_l$'s are defined as in  \eqref{defmul}, but with the multi-index $\beta$ replacing $\alpha$. 
Similarly, the term $z^n_{j,2} $ is a sum of terms of the form $\left\langle L_p \, f_{2j,p_1}^{\alpha},\overline{f_{2j,p_2}^{\beta}}  \right\rangle$ over  
$(p, \alpha, \beta) \in \N^{2} \times \N^{p_1} \times \N^{p_2}$ with $|p|  \leq n$ and $|\alpha| + |\beta| = n-1$.
The term $\left\langle L_p \, f_{2j,p_1}^{\alpha},\overline{f_{2j,p_2}^{\beta}}  \right\rangle$ 
 has an expression similar to \eqref{expressionL}, and  for $|\alpha| = r$ and $|\beta| = n-1-r$ the kernel 
$\kappa_{j,2}^{p,\alpha,\beta}(\bi)$ is given by
\begin{equation}
\begin{aligned}
&\kappa_{j,2}^{p,\alpha,\beta}(\bi):=\frac{1}
{(\llambda_j-\lambda_1)\prod_{l=1}^{r}
\left(\llambda_{j-2\sum_{m=1}^l i_m}-\mu_l\right)\prod_{l=1}^{p_1-1}\left(\llambda_{j- 2\sum_{h=1}^{\alpha_1 + \cdots + \alpha_l} i_h}-\lambda_{l+1}\right)}\times \\
&\qquad \times
\frac{1}{(\llambda_j-\lambda_{p_1+1})\prod_{l=r+2}^{{n}}\left(\llambda_{j-2\sum_{m=l}^n i_m}-\tilde{\mu}_l\right)\prod_{l=1}^{p_2-1}
\left(\llambda_{j- 2\sum_{h=1}^{\beta_1 + \cdots + \beta_l} i_h}- \lambda_{l+1}\right)}
\end{aligned}.
\label{kappa2}
\end{equation}
Using the explicit form of the eigenvectors $\left\{g_k\right\}_{-(N-1) \leq k \leq N}$ (see Lemma \ref{lem:spettroL0}), one verifies that  
\begin{align*}
&\left\langle g_{j-2i_1-\cdots-2i_r}, \overline{g_{j-2i_{r+1}-\cdots-2i_n}}\right\rangle = \delta\left(j, \sum_{m=1}^n i_m\right), \quad \left\langle g_{N-j-2i_1-\cdots-2i_r}, \overline{g_{N-j-2i_{r+1}-\cdots-2i_n}}\right\rangle =\delta\left(-j, \sum_{m=1}^n i_m\right).
\end{align*}
This is used to simplify the last term  in  \eqref{expressionL}. Moreover, using  $j = \sum_{m=1}^n i_m$ and  the identity  $\llambda_j=\llambda_{-j}$, one gets   that
\begin{align}
\label{eigenv.new} 
&\llambda_{j-2i_n}=\llambda_{j-2\sum_{m=1}^{n-1}i_m},\quad  \ldots , \quad \llambda_{j-2i_n-2i_{n-1}- \cdots - 2i_{r+1}}=\llambda_{j-2\sum_{m=1}^{r}i_m}.
\end{align}
 Recalling the definition of the coefficients $x_j^l$ (formula
 \eqref{x^l_j.formula}), we can write, for $\varsigma = 1,2$,
\begin{equation}
\label{z^n.exp}
z_{j,\varsigma}^n(\hat{b}, \hat{a}) = \frac{1}{N^{n/2}}\left(\tfrac{2}{N}\om{j}\right)^{-1/2} \sum_{\bii \in \Delta^n} \mathcal{K}^n_{j,\varsigma}\bii \, u_{i_1, \iota_1} \ldots u_{i_n, \iota_n}
\end{equation}
where 
the set
$$
\Delta^n := \left\{\bii\in \Z^n \times \N^n: \,  i_l \in K_N^0, \quad \iota_l \in \{1, 2\}, \quad  \forall 1 \leq l\leq  n \right\},
$$
the variables $u = (u_{i_1, \iota_1}, \cdots, u_{i_n, \iota_n})$ are defined by
$$
u_{i_r, 1} := \hat{b}_{i_r}, \qquad  u_{i_r, 2} :=\hat{a}_{i_r},
$$
the kernels $\mathcal{K}^n_{j,\varsigma}\bii$ are defined for $\bii \in \Delta^n$ by
\begin{equation}
\label{kernelK.def}
\mathcal{K}^n_{j,\varsigma}\bii :=
\tilde{\mathcal{K}}^n_{j,\varsigma}(\bi) \prod_{\{1 \leq l \leq n\} } \left(-2 \cos \left( \tfrac{(j-2i_i - \cdots
  -2i_l)\pi}{N}\right) \right)^{\iota_l-1}, 
\end{equation}
\begin{equation}
\label{kernelK.def2}
\tilde{\mathcal{K}}^{n}_{j,\varsigma}(\bi)  = \sum_{\substack{r+s=n-(\varsigma - 1) \\ p=(p_1,p_2)\in \mathbb{N}^2, \, |p|\leq n}}c_{p_1}c_{p_2}
\sum_{\substack{(\alpha, \beta) \in \N^{p_1}\times \N^{p_2}\\  |\alpha|=r,  \, |\beta| = s}} \mathcal{S}^{p,\alpha,\beta}_{j,\varsigma}(\bi)
\end{equation}
and finally
\begin{equation}
\label{def.Spab}
\mathcal{S}^{ p,\alpha,\beta}_{j,\varsigma}(\bi)= \delta\left(j, \sum_{m=1}^n i_m \right)
 \left(\frac{\im}{2\pi}\right)^{|p|}(-1)^n\oint_{\Gamma_j}\ldots\oint_{\Gamma_j}
 \kappa_{j,\varsigma}^{p,\alpha,\beta}(\bi)
 \diff{\lambda_1}\ldots 
 \diff{\lambda_{|p|}}. 
 \end{equation}
An analogous expansion holds also for $w^n_{j,1}$ and $w^n_{j,2}$.\\
We need now to get estimates of the kernels  $\mathcal{K}^n_{j,\varsigma}$, which will follow from estimates on  the denominators of $\kappa_{j,\varsigma}^{p,\alpha,\beta}$.
\begin{lem}
\label{stimaautov}
Let $\mu \in \Gamma_j:=\left\{\lambda \in \C: \;
\mmod{\lambda-\lambda_{2j}^0}= \min
\left(\frac{\langle j \rangle}{2N^2}, \frac{\langle N -j \rangle}{2N^2} \right)  \right\}$, where $\langle j \rangle = \left( 1+ |j|^2\right)^{1/2}$. Then  there exists a constant $R>0$, independent of $N$, such that for every $-(N-1) \leq k \leq N$ one has
\begin{equation}
\begin{aligned}
 \mmod{\llambda_k- \mu} \geq 
\begin{cases} 
 R  \langle j-k \rangle \langle j+k \rangle /N^2, & \qquad \mbox{ if } 0 \leq |j| \leq \lfloor N /2 \rfloor \\
 R \langle j-k \rangle \langle (N -j) + (N-k) \rangle/N^2, & \qquad \mbox{ if } \lfloor N/2 \rfloor + 1 \leq |j| \leq N 
 \end{cases}
 \end{aligned}
\label{stimaautovalori2}
\end{equation}
\end{lem}
\begin{proof}
Consider first the situation in which both  the eigenvalues $\llambda_j$ and $\llambda_k$ are  in the low half of the spectrum, namely   $0\leq |j|, |k| \leq \lfloor N/2 \rfloor$. In this case one has 
$$
|\llambda_{k} - \llambda_{j}| \equiv |\lambda^0_{2|k|} - \lambda^0_{2|j|}| = 2\mmod{\cos \left(\tfrac{|k| \pi}{N}\right)- \cos \left(\tfrac{|j| \pi}{N}\right)} = 2\mmod{\cos \left(\tfrac{k \pi}{N}\right)- \cos \left(\tfrac{j \pi}{N}\right)} \geq  \frac{4 |j^2 - k^2|}{N^2}.
$$
Therefore,  for $k\neq j$, there exists a positive constant $R_1$ such that for $\forall \mu \in \Gamma_j$
	\begin{align}
	\label{eig2}
	\mmod{\llambda_k - \mu} &\geq \mmod{\llambda_k - \llambda_j} - \frac{\langle j \rangle}{2N^2} \geq \frac{4|j^2-k^2|}{N^2}-\frac{\< j \>}{2N^2} \geq R_1 \frac{\< j-k \> \< j + k \>}{N^2},
	\end{align}
	where we used the inequality $\< j \> \leq 2 \< j-k \> \< j+k\> $, which holds since $j, k$ are integers. If $k=j$, then the claimed estimate follows trivially since $|\llambda_k - \mu| = \langle j \rangle/2N^2.$\\
Consider now the case when $\llambda_j$ is in the low half of the spectrum, while $\llambda_k$ is in the high half, i.e. $0 \leq |j| \leq \lfloor N/2 \rfloor, \,$ while $\lfloor N/2 \rfloor < |k| \leq N$. In this case the distance of the eigenvalues $\llambda_j$ and $\llambda_k$ is of order $\tfrac{1}{N}$, therefore the  estimate \eqref{stimaautovalori2} holds as well. More precisely, using $\cos x \geq 1-\frac{2}{\pi}x$ for $0\leq x \leq \pi/2$, one has
$$
|\llambda_{k} - \llambda_{j}| = |\lambda^0_{2|k|} - \lambda^0_{2|j|}| =  2\mmod{\cos \left(\tfrac{(N - |k|) \pi}{N}\right)+ \cos \left(\tfrac{j \pi}{N}\right)} \geq  \frac{4 (|k| - |j|)}{N} \geq \frac{\< j-k \> \< j + k \>}{N^2},
$$
where the last inequality holds since $\< l\>/N \leq 4, \,  $ $\forall |l| \leq 2N$. The inequality above implies that
\begin{align}
\label{diff.eig.case.2}
 \mmod{\llambda_k - \mu} \geq 
 |\llambda_{k} - \llambda_{j}|- \frac{\< j \>}{2N^2} 
  \geq \frac{\< j-k \> \< j + k \>}{N^2}-\frac{\<j\>}{2N^2} \geq R_2 \frac{\< j-k \> \< j + k \>}{N^2},
\end{align}
for some $R_2 >0$. Thus the first of \eqref{stimaautovalori2} is proved.\\
The proof of the second inequality of \eqref{stimaautovalori2} follows by symmetry and is omitted.
\end{proof}

We can now estimate the kernels $\mathcal K^n_{j, \varsigma}$ defined in \eqref{kernelK.def}.
\begin{lem}
\label{K.estimate}
 There exists a constant $R>0$, independent of $N$, such that $\mathcal{K}^{n}_{j, \varsigma}\bii$, $\varsigma=1,2$, satisfy, for every $n \geq 2$ and $1 \leq j \leq \lfloor N/2 \rfloor$, the estimates
\begin{equation}
\label{kernel.punct.est} 
\begin{aligned}
&\mmod{\mathcal{K}^{n}_{j, \varsigma}\bii} \leq R^n N^{2(n-1)}\delta\left( j, \sum_{l=1}^n i_l\right)\; \frac{1}{\prod_{l=1}^{n-1} \< \sum_{k=1}^l i_k \> \, \< \sum_{k=1}^l i_k - j \>},\\
&\mmod{\mathcal{K}^{n}_{N-j, \varsigma}\bii} \leq R^n N^{2(n-1)} \delta\left( -j, \sum_{l=1}^n i_l\right) \; \frac{1}{\prod_{l=1}^{n-1} \< \sum_{k=1}^l i_k \> \, \< \sum_{k=1}^l i_k - j \>}.
\end{aligned}
\end{equation}
\end{lem} 
 \begin{proof}
We start by estimating  $\kappa_{j, \varsigma}^{p,\alpha,\beta}(\bi)$, defined in \eqref{kappa} and \eqref{kappa2}. For every $-(N-1) \leq k \leq N$ and $\mu \in \Gamma_j$ one has $\mmod{\llambda_k - \mu} \geq \mmod{\llambda_j - \mu} \geq \min
\left(\frac{\langle j \rangle}{2N^2}, \frac{\langle N -j \rangle}{2N^2} \right) $, therefore
\begin{align*}
&\mmod{(\llambda_j-\lambda_1)\prod_{l=1}^{p_1-1}
\left(\llambda_{j- 2\sum_{h=1}^{\alpha_1 + \cdots + \alpha_l} i_h}-\lambda_{l+1}\right) (\llambda_j-\lambda_{p_1+1})
\prod_{l=1}^{p_2-1}
\left(\llambda_{j- 2\sum_{h=1}^{\beta_1 + \cdots + \beta_l} i_h}- \lambda_{l+1}\right)}\\
 & \qquad \geq \left[\min
\left(\frac{\langle j \rangle}{2N^2}, \frac{\langle N -j \rangle}{2N^2} \right)\right]^{|p|}.
\end{align*}
Let now $1 \leq j \leq \lfloor N/2 \rfloor$. 
By Lemma \ref{stimaautov}, formula \eqref{eigenv.new} and the inequality $\frac{|\llambda_{j-2\sum_{m=1}^r i_m}-\llambda_j|}{|\llambda_{j-2\sum_{m=1}^r i_m}-\mu|} \leq 2$ (which is used to estimate just $\kappa_{j, 1}^{p,\alpha,\beta}(\bi)$), it follows  that, for $\varsigma = 1,2$, 
 \begin{align*}
 \mmod{\kappa_{j, \varsigma}^{p,\alpha,\beta}(\bi)} \leq \frac{2}{\left[\min
\left(\frac{\langle j \rangle}{2N^2}, \frac{\langle N -j \rangle}{2N^2} \right)\right]^{|p|} \prod_{l=1}^{n-1}\mmod{\llambda_{j-2\sum_{m=1}^l i_m}-\mu_l}} \leq
\frac{2 \, a_j(i_1, \cdots, i_{n-1})}{\left[\min
\left(\frac{\langle j \rangle}{2N^2}, \frac{\langle N -j \rangle}{2N^2} \right)\right]^{|p|}}
 \end{align*}
 where $$a_j(i_1, \cdots, i_{n-1}) := \frac{R^{n-1}N^{2(n-1)}}{\prod_{l=1}^{n-1} \< \sum_{k=1}^l i_k \> \, \< \sum_{k=1}^l i_k - j \>}\ .$$ 
 
 To estimate  $\mathcal{S}^{p,\alpha,\beta}_{j, \varsigma}$ consider  \eqref{def.Spab}. The $\mathcal{S}^{p,\alpha,\beta}_{j, \varsigma}$'s are  defined by integrating the kernels $\kappa_{j,\varsigma}^{p,\alpha,\beta} $ over $\Gamma_j$ $|p|$-times. Since $\mmod{\Gamma_j}=2\pi \min
\left(\frac{\langle j \rangle}{2N^2}, \frac{\langle N -j \rangle}{2N^2} \right)  $, one gets  
$$
\mmod{\mathcal{S}^{ p,\alpha,\beta}_{j, \varsigma}(\bi)}\leq \left[  \min \left( \tfrac{\langle j \rangle}{N^2},  \tfrac{\langle N-j \rangle}{N^2} \right)\right]^{|p|} \delta\left(j, \sum_{l=1}^n i_l \right) \mmod{\kappa_{j, \varsigma}^{p,\alpha,\beta}(\bi)} \leq 2 \delta\left(j, \sum_{l=1}^n i_l \right) a_j(i_1, \cdots, i_{n-1}).
$$
 Finally consider  $\mathcal{K}^n_{j, \varsigma}$. From \eqref{kernelK.def} one has $\mmod{\mathcal{K}^n_{j, \varsigma}\bii} \leq 2^n \mmod{\tilde{\mathcal{K}}_{j, \varsigma}^{n}(\bi)} $, and from \eqref{kernelK.def2}   
 \begin{align*}
\left|\tilde{\mathcal{K}}_{j, \varsigma}^{n}(\bi)\right|&\leq \delta \left(j, \sum_{l=1}^n i_l \right) \, a_j(i_1, \cdots, i_{n-1}) \sum_{\substack{r+s=n-(\varsigma - 1) \\ p=(p_1,p_2)\in \mathbb{N}^2, \, |p|\leq n}}c_{p_1}c_{p_2}
\sum_{\substack{ (\alpha, \beta) \in \N^{p_1}\times \N^{p_2} \\  |\alpha|=r,  \, |\beta| = s}} 1\\
&  \leq C^n \delta \left(j, \sum_{l=1}^n i_l \right) \, a_j(i_1, \cdots, i_{n-1}) \ ,
\end{align*}
thus the first estimate of \eqref{kernel.punct.est} follows. The proof of the second one is similar, and is omitted.
 \end{proof}
Define now $\mathcal{K}^n_j := \mathcal{K}^n_{j,1} + \mathcal{K}^n_{j,2}$. Then 
 \begin{equation}
\label{w^n.exp3}
\begin{aligned}
&z_{j}^n(\hat{b}, \hat{a}) = \frac{D_j}{N^{n/2}} \sum_{\bii \in \Delta^n} \mathcal{K}^n_{j}\bii \, u_{i_1, \iota_1} \ldots u_{i_n, \iota_n}, \quad
w_{j}^n(\hat{b}, \hat{a}) = \frac{D_j}{N^{n/2}} \sum_{\bii \in \Delta^n} \mathcal{H}^n_{j}\bii \, u_{i_1, \iota_1} \ldots u_{i_n, \iota_n},
\end{aligned}
\end{equation}
where $\mathcal{H}^n_{j}\bii = \overline{\mathcal{K}^n_j(-{\bf i}, {\boldsymbol \iota})}$. The second formula holds since  
for $b, a$ real one has $w^n(b,a) = \overline{z^n(b,a)}$.
 \begin{cor}
 \label{K.norm}
 Let 
 $
 \Delta^n_{ j} := \{ \bii \in \Delta^n: \, \sum_{l=1}^n i_l =  j \}.  $
Then for $1\leq j \leq \lfloor N/2 \rfloor $ one has $\mbox{supp } \mathcal{K}^n_j \subseteq \Delta^n_j$ and $\mbox{supp } \mathcal{K}^n_{N-j} \subseteq \Delta^n_{-j}$.  Moreover  
 \begin{equation}
 \label{K.est.2}
\norm{\mathcal{K}^{n}_j}_{\Delta^n_j}, \quad \norm{\mathcal{K }^{n}_{N-j}}_{\Delta^n_{-j}} \leq \frac{R^n N^{2(n-1)}}{\< j \>^{n-1}} ,
 \end{equation}
 where $\norm{\mathcal{K}^{n}_j}^2_{\Delta^n_{ j}} := \sup_{\iota_1, \cdots, \iota_n \in \{1,2\}} \sum_{i_1+ \cdots + i_n =  j}\mmod{\mathcal{K}^{n}_j\bii}^2$.
 \end{cor}
\proof Just remark that $\frac{\langle j \rangle^2}{\langle k \rangle^2 \langle k-j \rangle^2} \leq 4 \left(\frac{1}{\langle k \rangle^2} + \frac{1}{\langle k-j \rangle^2} \right)$.\qed 

\vspace{1em}
We   prove now bounds on the  map $\und{Z^n}(\hat{b}, \hat{a}) := (\und{z^n}(\hat{b}, \hat{a}), \und{w^n}(\hat{b}, \hat{a}))$. 
\begin{lem}
\label{Z^n.ext}
There exists a constant $C>0$, independent of $N$,  such that for any $s \geq 0$ and $ \sigma \geq 0$ 
\begin{equation}
\label{Zn.est1}
\norm{\und{Z}^n(|\hat{b}|,|\hat{a}|)}_{\spazio{s+1, \sigma}}\leq C^n N^{2(n-1)}\norm{(b,a)}^n_{\Cs^\reg}, \qquad \forall \,
n\geq 2.
\end{equation}
\end{lem}
\begin{proof}  By formula \eqref{z^n.exp} one has that for $1 \leq j \leq \lfloor N/2 \rfloor$
\begin{equation}
\begin{aligned}
\label{expansion.Zn.1}
&\mmod{\und{z}_j^{n}(|\hat{b}|,|\hat{a}|)} \leq \frac{D_j}{N^{n/2}} \sum_{\bii \in \Delta_j^n}  \mmod{\mathcal{K}^{ n}_j\bii} |u_{i_1, \iota_1}| \ldots |u_{i_n, \iota_n}|, \\
&\mmod{\und{z}_{N-j}^{n}(|\hat{b}|,|\hat{a}|)} \leq
\frac{D_j}{N^{n/2}} \sum_{\bii \in \Delta_{-j}^n}  \mmod{\mathcal{K}^{
    n}_{N-j}\bii} |u_{i_1, \iota_1}| \ldots |u_{i_n, \iota_n}|.
\end{aligned}
\end{equation}
Introduce $\Lambda(\bi) := [i_1] \cdots [i_n]$, where $[i_r]=\max (1, |i_r|)$ $\forall 1 \leq r \leq n$,  and remark that for some  constant $R>0$ one has
$$
\sup_{i_1 + \cdots + i_n=  j} \Lambda(\bi)^{-1} \leq \frac{R^n}{\langle j \rangle}, \qquad \forall j \in \Z.
$$
Therefore, by Corollary \ref{K.norm}, 
\begin{align*}
&\mmod{\und{z}_j^{n}(|\hat{b}|,|\hat{a}|)}^2 \leq \frac{1}{N^{n}}D_j^2 \norm{\mathcal{K}^{ n}_j}^2_{\Delta^n_j} \left(\sup_{i_1 + \cdots + i_n = j}\Lambda(\bi)^{-2s}\right) \sum_{\bii \in \Delta^n_j} [i_1]^{2s} |u_{i_1, \iota_1}|^2 \ldots [i_n]^{2s} |u_{i_n, \iota_n}|^2, \\
&\mmod{\und{z}_{N-j}^{n}(|\hat{b}|,|\hat{a}|)}^2 \leq \frac{1}{N^{n}}D_j^2 \norm{\mathcal{K}^{ n}_{N-j}}^2_{\Delta^n_{-j}} \left(\sup_{i_1 + \cdots + i_n = -j}\Lambda(\bi)^{-2s}\right) \sum_{\bii \in \Delta^n_{-j}} [i_1]^{2s} |u_{i_1, \iota_1}|^2 \ldots [i_n]^{2s} |u_{i_n, \iota_n}|^2.
\end{align*}
Use now inequalities \eqref{K.est.2}, the definition of $D_j$, the fact that  $e^{2\sigma |j|}\leq e^{2\sigma |i_1|}\cdots e^{2\sigma |i_{n-1}|}e^{2\sigma |j-i_1-\cdots-i_{n-1}|}$,  and the bounds $|u_{l,\iota_l}| \leq |\hat{b}_l| + |\hat{a}_l|$, to deduce that, for any $n \geq 2$, 
\begin{align*}
&\frac{1}{N}\sum_{j=1}^{\lfloor N/2 \rfloor}
 [j]^{2(s+1)}e^{2\sigma |j|}\om{j}\left(\mmod{\und{z}_{j}^{n}(|\hat{b}|,|\hat{a}|)}^2  + \mmod{\und{z}_{N-j}^{n}(|\hat{b}|,|\hat{a}|)}^2 \right)  \\
  & \quad \leq   N^{4(n-1)}  \frac{C^n}{N^{n}} 
  \sum_{j=1}^{\lfloor N/2 \rfloor}
 [j]^{2(2-n)} e^{2\sigma |j|} \sum_{\bii \in \Delta^n_{\pm j}} [i_1]^{2s} |u_{i_1, \iota_1}|^2 \ldots [i_n]^{2s} |u_{i_n, \iota_n}|^2 \\
& \quad \leq  N^{4(n-1)} C^n \norm{(b,a)}^{2n}_{\Cs^\reg}.
\end{align*}
Since  $w^n(\hat{b}, \hat{a})$ satisfies the same inequality,   estimate \eqref{Zn.est1} holds.
\end{proof}
Consider now the map $(\hat{b},\hat a)\mapsto dZ^n(\hat b,\hat a)^*$, where $dZ^n(\hat b, \hat a)^*$ is the adjoint of the differential of $Z^n$. Explicitly, if  $\xi, \eta$ are vectors in $\C^{N-1}$ and $h, g$ are vectors in $\C^N$ such that
$(h, g)\equiv dZ^n(\hat{b}, \hat{a})^* (\xi, \eta)$,  then the $j^{th}$ components of $h$ and $g$ are given by
\begin{equation}
\label{diff.trasp.comp}
(h_j, \, g_j)  = \left( \sum_{k=1}^{N-1} \left( \overline{\frac{\partial z_{k}^n}{\partial \hat{b}_j}(\hat{b}, \hat{a})} \xi_{k} +  \overline{\frac{\partial w_{k}^n}{\partial \hat{b}_j}(\hat{b}, \hat{a})} \eta_{k} \right), \quad 
  \sum_{k=1}^{N-1} \left( \overline{\frac{\partial z_{k}^n}{\partial \hat{a}_j}(\hat{b}, \hat{a})} \xi_{k} +  \overline{\frac{\partial w_{k}^n}{\partial \hat{a}_j}(\hat{b}, \hat{a})} \eta_{k} \right) \right).
\end{equation}
Denote by $\und{h}, \und{g}$ the vectors of $\C^N$ whose components are given by
\begin{equation}
\label{diff.trasp.comp.und}
(\underline{ h_j}, \, \und{ g_j})  = \left( \sum_{k=1}^{N-1} \left( \und{ \frac{\partial z_{k}^n}{\partial \hat{b}_j}}(|\hat{b}|, |\hat{a}|) |\xi_{k}| +  \und{ \frac{\partial w_{k}^n}{\partial \hat{b}_j}}(|\hat{b}|, |\hat{a}|) |\eta_{k}| \right), \quad 
  \sum_{k=1}^{N-1} \left( \und{ \frac{\partial z_{k}^n}{\partial \hat{a}_j}}(|\hat{b}|, |\hat{a}|) |\xi_{k}| + \und{ \frac{\partial w_{k}^n}{\partial \hat{a}_j}}(|\hat{b}|, |\hat{a}|) |\eta_{k}| \right) \right).
\end{equation}
We begin to study the case $n=2$.
\begin{lem}
\label{dZ^n=2}
There exists a constant $R>0$, independent of $N$,  such that $\forall s \geq 0 \,$, $\sigma \geq 0$  one has 
\begin{equation}
\label{dZ*2}
\norm{\und{dZ^2}(|\hat{b}|, |\hat{a}|)^* (|\xi|, |\eta|)}_{\Cs^{s+2, \sigma}}\leq R N^{3} \norm{(b,a)}_{\Cs^{\reg}}
 \norm{(\xi, \eta)}_{\spazio{\reg}} \ .
 \end{equation}
\end{lem}
\begin{proof} By \eqref{taylor.exp}, one computes that the second order terms $Z^2 = (z^2, w^2)$  are given by 
 \begin{align*}
&z^2_k(\hat{b}, \hat{a}) = \tfrac{D_k}{N}\sum_{l \neq 0} \left(\hat{b}_l - 2 \cos(\tfrac{(k-2l)\pi}{N})\hat{a}_l\right) \left(\hat{b}_{k-l} - 2 \cos(\tfrac{k\pi}{N})\hat{a}_{k-l}\right)/(\lambda_{2(k-2l)}^0 - \lambda_{2k}^0)\\
&w^2_k(\hat{b}, \hat{a}) = \tfrac{D_k}{N}\sum_{l \neq 0} \left(\hat{b}_{N-l} - 2 \cos(\tfrac{(k-2l)\pi}{N})\hat{a}_{N-l}\right) \left(\hat{b}_{l-k} - 2 \cos(\tfrac{k\pi}{N})\hat{a}_{l-k}\right)/(\lambda_{2(k-2l)}^0 - \lambda_{2k}^0).
\end{align*}
Let  $\und{h}, \, \und{g}$ be as in \eqref{diff.trasp.comp.und} with $n=2$.
Using the explicit expressions for $z_k^2$ and $w_k^2$, one computes that for $0 \leq j \leq \lfloor N/2\rfloor $ 
\begin{align*}
|\und{h_j}|  & \leq \frac{1}{N} \sum_{k=1}^{N-1}  \frac{\left(|\hat{b}_{k-j}|+ 2|\hat{a}_{k-j}|\right)\, D_k( |\xi_k| +  |\eta_k|) }{|\lambda_{2(k-2j)}^0 - \lambda_{2k}^0|} \\
& \leq N \sum_{k=1}^{\lfloor N/2\rfloor } \frac{\left(|\hat{b}_{k-j}|+ 2|\hat{a}_{k-j}|\right)\, D_k( |\xi_k| +  |\eta_k|)}{\langle k-j \rangle \langle j \rangle} + N \sum_{k=\lfloor N/2 \rfloor +1}^{N-1 } \frac{\left(|\hat{b}_{k-j}|+ 2|\hat{a}_{k-j}|\right)\, D_k( |\xi_k| +  |\eta_k|)}{\langle N-k +j \rangle \langle j \rangle}\\
& \leq N \sum_{k=1}^{\lfloor N/2\rfloor } \frac{\left(|\hat{b}_{k-j}|+ 2|\hat{a}_{k-j}|\right)\, D_k( |\xi_k| +  |\eta_k|) }{\langle k-j \rangle \langle j \rangle} + \frac{ \left(|\hat{b}_{N-k-j}|+ 2|\hat{a}_{N-k-j}|\right)\, D_k( |\xi_{N-k}| +  |\eta_{N-k}|)}{\langle k+j \rangle \langle j \rangle}\\
& \leq \frac{  N^2}{\langle j \rangle} \sum_{k=1}^{\lfloor N/2\rfloor } \frac{\left(|\hat{b}_{k-j}|+ 2|\hat{a}_{k-j}|\right)\, \langle k \rangle^{1/2}( |\xi_k| +  |\eta_k|) }{\langle k-j \rangle \langle k \rangle} + \frac{ \left(|\hat{b}_{N-k-j}|+ 2|\hat{a}_{N-k-j}|\right)\, \langle k \rangle^{1/2}( |\xi_{N-k}| +  |\eta_{N-k}|)}{\langle k+j \rangle \langle k \rangle}
\end{align*}
where in the last inequality we used that $D_k \leq  N/\langle k \rangle^{1/2}$. 
With analogous computations, one verifies that 
$$
|\und{h_{N-j}}|  \leq \frac{N^2}{\langle j \rangle} \sum_{k=1}^{\lfloor N/2\rfloor } \frac{\left(|\hat{b}_{k+j}|+ 2|\hat{a}_{k+j}|\right)\, \langle k \rangle^{1/2}( |\xi_k| +  |\eta_k|) }{\langle k+j \rangle \langle k \rangle} + \frac{ \left(|\hat{b}_{j-k}|+ 2|\hat{a}_{j-k}|\right)\, \langle k \rangle^{1/2}( |\xi_{N-k}| +  |\eta_{N-k}|)}{\langle k-j \rangle \langle k \rangle}.
$$
Proceeding as in the proof of Lemma \ref{bullet}, one obtains that there exist constants $C, C' >0$, independent of $N$, such that
\begin{align}
\nonumber
& \frac{1}{N}\sum_{j=0}^{\lfloor N/2 \rfloor}[j]^{2(s+2)} e^{2 \sigma |j|}(|\und{h_j}|^2 + |\und{h_{N-j}}|^2) \\
\nonumber
& \qquad \leq C N^3 \left( \sum_{k=0 }^{N-1} [k]_N^{2s} e^{2\sigma [k]_N} (|\hat{b}_k|^2 + |\hat{a}_k|^2) \right) \left( \sum_{l=1}^{N-1} [l]_N^{2s} e^{2\sigma [l]_N} [l]_N ( |\xi_l|^2 + |\eta_l|^2) \right)\\
\label{und.h.2.est}
& \qquad \leq C'  N^6 \norm{(b,a)}_{\Cs^{\reg }}^2 \norm{(\xi, \eta)}_{\spazio{\reg}}^2 
\end{align}
where in the last inequality we used that $[l]_N \leq  N \om{l}$ for $l$ integer.
One verifies that  $\und{g}$ satisfies the same inequality as \eqref{und.h.2.est}. Thus estimate \eqref{dZ*2} follows from the following inequality:
\begin{equation}
\label{dZ2.norm2}
\norm{\und{dZ^2}(|\hat{b}|, |\hat{a}|)^* (|\xi|, |\eta|)}_{\Cs^{s+2, \sigma}}^2 \leq \frac{1}{N}\sum_{j = 0 }^{N-1} [j]_N^{2s+4} e^{2\sigma [j]_N} \left(|\und{h_j}|^2 + |\und{g_j}|^2 \right).
\end{equation}
\end{proof}
We study now  $ dZ^n(\hat{b}, \hat{a})^* $ for $n \geq 3$.
\begin{lem}
\label{dZ^n>3}
There exists a constant $R>0$, independent of $N$,  such that for every $s \geq 0$, $\sigma \geq 0$ and  $n \geq 3$
\begin{equation}
\label{dZn.1}
\norm{\und{dZ^n}(|\hat{b}|, |\hat{a}|)^* (|\xi|, |\eta|)}_{\Cs^{s+2, \sigma}}\leq R^n N^{2n-1}\norm{(b,a)}^{n-1}_{\Cs^{\reg}}
 \norm{(\xi, \eta)}_{\spazio{\reg}}.
 \end{equation}
\end{lem}
\begin{proof} Let $h,g$ be as in \eqref{diff.trasp.comp.und}. We concentrate on $h$ only, the estimates for $g$ being analogous. Write $h_j =\sum_{k=1}^{N-1}  \frac{\partial z_{k}^n}{\partial \hat{b}_j} \xi_{k} + \sum_{k=1}^{N-1} \frac{\partial w_{k}^n}{\partial \hat{b}_j} \eta_{k} =: h_{j,1}+ h_{j,2}$. By \eqref{w^n.exp3}  one gets that  
\begin{align*}
h_{j,1}   = \frac{1}{N^{n/2}} \sum_{l=1}^n A_{j}^{n,l}(D\xi, u, \ldots, u), \qquad h_{j,2} &  = \frac{1}{N^{n/2}} \sum_{l=1}^n B_{j}^{n,l}(D\eta, u, \ldots, u)
\end{align*}
where  $D$ is defined in \eqref{D.def},  the  multilinear map  $A_{j}^{n,l}$ is defined by
$$
A_{j}^{n,l}(h,u, \ldots , u) = \sum_{\bii \in \Delta^n} \mathcal{A}_{j}^{ n,l}\bii u_{i_1, \iota_1} \ldots h_{i_l} \ldots u_{i_n, \iota_n}, 
$$
$B_{j}^{n,l}$ is defined analogously but with kernel $\mathcal{B}_{j}^{ n,l}\bii$, 
and finally $\mathcal{A}_j^{n,l}$ and  $\mathcal{B}^{ n,l}_{j}$ are  defined for $ 1 \leq j \leq \lfloor N/2 \rfloor$ by
 \begin{align*}
& \mathcal{A}^{n,l}_{j}\bii := \mathcal{K}^{n}_{i_l}\Big( (i_1, \ldots, i_{l-1}, j, i_{l+1}, \ldots, i_n),\,  (\iota_1, \ldots, \iota_{l-1}, 1, \iota_{l+1}, \ldots, \iota_n)\Big),\\
& \mathcal{A}^{n,l}_{N-j}\bii := \mathcal{K}^{n}_{i_l}\Big( (i_1, \ldots, i_{l-1}, -j, i_{l+1}, \ldots, i_n),\,  (\iota_1, \ldots, \iota_{l-1}, 1, \iota_{l+1}, \ldots, \iota_n)\Big),
 \end{align*}
 while $\mathcal{B}^{n,l}_{j}\bii = \overline{\mathcal{A}^{n,l}_{j}(-{\bf i},{\boldsymbol \iota})}$ and $\mathcal{B}^{n,l}_{N-j}\bii= \overline{\mathcal{A}^{n,l}_{N-j}(-{\bf i},{\boldsymbol \iota})}$, see \eqref{w^n.exp3}.
 By  Corollary \eqref{K.norm} it follows that 
 \begin{align*}
 &\mbox{ supp } \mathcal{A}^{ n,l}_{j}=\mbox{ supp } \mathcal{B}^{ n,l}_{N-j} \equiv 
  \left\{ \bii: \, i_1 + \cdots +  i_{l-1} -i_l +  i_{l+1}+ \cdots + i_n =-j, \, \iota_l= 1 \right\} \subseteq \Delta^n_{-j}, \\
  &\mbox{ supp } \mathcal{A}^{ n,l}_{N-j}=\mbox{ supp } \mathcal{B}^{ n,l}_{j} \equiv 
  \left\{ \bii: \, i_1 + \cdots +  i_{l-1} -i_l +  i_{l+1}+ \cdots + i_n =j, \, \iota_l= 1 \right\} \subseteq \Delta^n_{j}.
 \end{align*}
Proceeding as in the proof of Corollary \ref{K.norm}, one proves that there exists a constant  $R>0$, independent of $N$, such that  (see \cite{kuksinperelman})
 \begin{equation}
 \label{A.B.est}
\max_{1\leq l \leq n} \left(\norm{\mathcal{A}_{j}^{ n,l}}_{\Delta^n_{-j}}, \, \norm{\mathcal{A}_{N-j}^{ n,l}}_{\Delta^n_j}, \, \norm{\mathcal{B}_{j}^{ n,l}}_{\Delta^n_j}, \, \norm{\mathcal{B}_{N-j}^{ n,l}}_{\Delta^n_{-j}} \right) \leq 
\frac{R^n N^{2(n-1)}}{\langle j\rangle^2}, \quad \forall n \geq 3 \ .
\end{equation}  
Thus $\und{h}$, defined in \eqref{diff.trasp.comp.und}, satisfies
 $$
 |\und{h_j}| \leq \frac{1}{N^{n/2}} \sum_{l=1}^n \left(\und{A_{j}^{n,l}}(|D\xi|, |u|, \ldots, |u|)+\und{B_{j}^{n,l}}(|D\eta|, |u|, \ldots, |u|)\right), 
$$
where $ \und{A_{j}^{ n,l}}(h,u, \ldots , u) = \sum_{\bii \in \Delta^n} \mmod{\und{\mathcal{A}_{j}^{n,l}\bii}} u_{i_1,\iota_1} \ldots h_{i_l} \ldots u_{i_n, \iota_n}$, and $\und{B_{j}^{ n,l}}$ is defined in analogous way. 
Then, using \eqref{A.B.est} and arguing as in the proof of Lemma \ref{Z^n.ext}, one proves the estimate
\begin{align*}
 \frac{1}{N}\sum_{j=0}^{N-1}[j]_N^{2(s+2)} e^{2 \sigma [j]_N} |\und{h_j}|^2  &
 \leq R^n N^{4n-5} \norm{(b,a)}_{\Cs^\reg}^{2(n-1)}\left(\frac{1}{N} \sum_{l=1}^{N-1} [l]_N^{2s} e^{2\sigma [l]_N} D_l^2 ( |\xi_l|^2 + |\eta_l|^2) \right)\\
 & \leq  R^n N^{4n-2} \norm{(b,a)}^{2(n-1)}_{\Cs^\reg} \norm{(\xi, \eta)}_{\spazio{s-1, \sigma}}^2 \  ,
\end{align*}
where in the last inequality we used that $D_l^2 \leq  \frac{N^3}{ [l]_N ^2} \om{l}$.
One verifies that $\und{g}$ satisfies the same inequality, thus estimate \eqref{dZn.1} follows.
\end{proof}
\vspace{1em}
We can finally prove  property $(Z 4)$.  Let $s \geq 0$, $\sigma \geq 0$ be fixed.
By Lemma \ref{Z^n.ext}, \ref{dZ^n=2} and \ref{dZ^n>3}, there exists $C_1, C_2, \epsilon_*>0$, independent of $N$,  such that for every $0 < \epsilon \leq \epsilon_* $ it holds that
\begin{align*}
  \sup_{\norm{(b,a)}_{\Cs^\reg} \leq \epsilon/N^2} \norm{\und{Z^0}(b,a)}_{\spazio{s+1, \sigma}} & \leq \sum_{n \geq 2} \, \sup_{\norm{(b,a)}_{\Cs^\reg} \leq \epsilon/N^2} \norm{\und{Z^n}(b,a)}_{\spazio{s+1, \sigma}} \\
&  \leq \sum_{n \geq 2} R^n N^{2(n-1)}\frac{\epsilon^n}{N^{2n}} \leq \frac{C_1 \epsilon^2}{N^2} ,\\
 \sup_{\norm{(b,a)}_{\Cs^\reg} \leq \epsilon/N^2} \norm{\und{dZ^0}(b,a)^*}_{\L(\spazio{\reg},\, \Cs^{s+2, \sigma})} &\leq \sum_{n \geq 2} \, \sup_{\norm{(b,a)}_{\Cs^\reg} \leq \epsilon/N^2} \norm{\und{dZ^n}(b,a)^*}_{\L(\spazio{\reg},\, \Cs^{s+2, \sigma})}\\
& \leq\sum_{n \geq 2} R^n N^{2n-1} \frac{\epsilon^{n-1}}{N^{2(n-1)}} \leq  C_2 N \epsilon \ . 
\end{align*}



\begin{thebibliography}{BGGK93}

\bibitem[BCP13]{benettin_ponno2}
G.~Benettin, H.~Christodoulidi, and A.~Ponno.
\newblock The {F}ermi-{P}asta-{U}lam problem and its underlying integrable
  dynamics.
\newblock {\em J. Stat. Phys.}, 152(2):195--212, 2013.

\bibitem[BG06]{bambusi.grebert}
D.~Bambusi and B.~Gr{\'e}bert.
\newblock Birkhoff normal form for partial differential equations with tame
  modulus.
\newblock {\em Duke Math. J.}, 135(3):507--567, 2006.

\bibitem[BGG04]{berchiallagalganigiorgilli}
Luisa Berchialla, Luigi Galgani, and Antonio Giorgilli.
\newblock Localization of energy in {FPU} chains.
\newblock {\em Discrete Contin. Dyn. Syst.}, 11(4):855--866, 2004.

\bibitem[BGGK93]{kappelerfibrationtoda}
D.~B{\"a}ttig, B.~Gr{\'e}bert, J.-C. Guillot, and T.~Kappeler.
\newblock Fibration of the phase space of the periodic {T}oda lattice.
\newblock {\em J. Math. Pures Appl. (9)}, 72(6):553--565, 1993.

\bibitem[BGP04]{berchiallagiorgillipaleari}
Luisa Berchialla, Antonio Giorgilli, and Simone Paleari.
\newblock Exponentially long times to equipartition in the thermodynamic limit.
\newblock {\em Physics Letters A}, 321(3):167 -- 172, 2004.

\bibitem[BGPU03]{thierrygolse}
A.~Bloch, F.~Golse, T.~Paul, and A.~Uribe.
\newblock Dispersionless {T}oda and {T}oeplitz operators.
\newblock {\em Duke Math. J.}, 117(1):157--196, 2003.

\bibitem[BKP09]{bambuthomas3}
Dario Bambusi, Thomas Kappeler, and Thierry Paul.
\newblock De {T}oda \`a {K}d{V}.
\newblock {\em C. R. Math. Acad. Sci. Paris}, 347(17-18):1025--1030, 2009.

\bibitem[BKP13a]{bambuthomas2}
D.~{Bambusi}, T.~{Kappeler}, and T.~{Paul}.
\newblock {Dynamics of periodic Toda chains with a large number of particles}.
\newblock {\em ArXiv e-prints}, arXiv:1309.5441 [math.AP], September 2013.

\bibitem[BKP13b]{bambuthomas}
D.~{Bambusi}, T.~{Kappeler}, and T.~{Paul}.
\newblock {From Toda to KdV}.
\newblock {\em ArXiv e-prints}, arXiv:1309.5324 [math.AP], September 2013.

\bibitem[BP06]{BamPon}
Dario Bambusi and Antonio Ponno.
\newblock On metastability in {FPU}.
\newblock {\em Comm. Math. Phys.}, 264(2):539--561, 2006.

\bibitem[BP11]{benettin_ponno}
G.~Benettin and A.~Ponno.
\newblock Time-scales to equipartition in the {F}ermi-{P}asta-{U}lam problem:
  finite-size effects and thermodynamic limit.
\newblock {\em J. Stat. Phys.}, 144(4):793--812, 2011.

\bibitem[Car07]{carati}
A.~Carati.
\newblock An averaging theorem for {H}amiltonian dynamical systems in the
  thermodynamic limit.
\newblock {\em J. Stat. Phys.}, 128(4):1057--1077, 2007.

\bibitem[CM12]{carati_maiocchi}
Andrea Carati and Alberto~Mario Maiocchi.
\newblock Exponentially long stability times for a nonlinear lattice in the
  thermodynamic limit.
\newblock {\em Comm. Math. Phys.}, 314(1):129--161, 2012.

\bibitem[Eli90]{eliasson}
L.~H. Eliasson.
\newblock Normal forms for {H}amiltonian systems with {P}oisson commuting
  integrals---elliptic case.
\newblock {\em Comment. Math. Helv.}, 65(1):4--35, 1990.

\bibitem[FFM82]{Ferguson}
W.~E. Ferguson, Jr., H.~Flaschka, and D.~W. McLaughlin.
\newblock Nonlinear normal modes for the {T}oda chain.
\newblock {\em J. Comput. Phys.}, 45(2):157--209, 1982.

\bibitem[Fla74]{flaschka}
H.~Flaschka.
\newblock The {T}oda lattice. {I}. {E}xistence of integrals.
\newblock {\em Phys. Rev. B (3)}, 9:1924--1925, 1974.

\bibitem[FM76]{Flaschka_McLaughlin}
H.~Flaschka and D.~W. McLaughlin.
\newblock Canonically conjugate variables for the {K}orteweg-de {V}ries
  equation and the {T}oda lattice with periodic boundary conditions.
\newblock {\em Progr. Theoret. Phys.}, 55(2):438--456, 1976.

\bibitem[FPU65]{FPU}
E.~{Fermi}, J.~{Pasta}, and S.~{Ulam}.
\newblock Studies of non linear problems.
\newblock In {\em Enrico Fermi Collected Papers, vol. II}, pages 977--988.
  University of Chicago Press/Accademia Nazionale dei Lincei, Chicago/Roma,
  1965.

\bibitem[GPP12]{giorgilliPP}
Antonio Giorgilli, Simone Paleari, and Tiziano Penati.
\newblock Extensive adiabatic invariants for nonlinear chains.
\newblock {\em J. Stat. Phys.}, 148(6):1106--1134, 2012.

\bibitem[H{\'e}n74]{henon}
M.~H{\'e}non.
\newblock Integrals of the {T}oda lattice.
\newblock {\em Phys. Rev. B (3)}, 9:1921--1923, 1974.

\bibitem[HK08a]{kapphen3}
Andreas Henrici and Thomas Kappeler.
\newblock Birkhoff normal form for the periodic {T}oda lattice.
\newblock In {\em Integrable systems and random matrices}, volume 458 of {\em
  Contemp. Math.}, pages 11--29. Amer. Math. Soc., Providence, RI, 2008.

\bibitem[HK08b]{kapphen1}
Andreas Henrici and Thomas Kappeler.
\newblock Global action-angle variables for the periodic {T}oda lattice.
\newblock {\em Int. Math. Res. Not. IMRN}, (11):Art. ID rnn031, 52, 2008.

\bibitem[HK08c]{kapphen2}
Andreas Henrici and Thomas Kappeler.
\newblock Global {B}irkhoff coordinates for the periodic {T}oda lattice.
\newblock {\em Nonlinearity}, 21(12):2731--2758, 2008.

\bibitem[HL12]{lubich}
Ernst Hairer and Christian Lubich.
\newblock On the energy distribution in {F}ermi-{P}asta-{U}lam lattices.
\newblock {\em Arch. Ration. Mech. Anal.}, 205(3):993--1029, 2012.

\bibitem[Kap91]{kappelerKdV}
Thomas Kappeler.
\newblock Fibration of the phase space for the {K}orteweg-de {V}ries equation.
\newblock {\em Ann. Inst. Fourier (Grenoble)}, 41(3):539--575, 1991.

\bibitem[Kat66]{kato}
Tosio Kato.
\newblock {\em Perturbation theory for linear operators}.
\newblock Die Grundlehren der mathematischen Wissenschaften, Band 132.
  Springer-Verlag New York, Inc., New York, 1966.

\bibitem[KP03]{kamkdv}
Thomas Kappeler and J{\"u}rgen P{\"o}schel.
\newblock {\em Kd{V} \& {KAM}}, volume~45 of {\em Ergebnisse der Mathematik und
  ihrer Grenzgebiete. 3. Folge. A Series of Modern Surveys in Mathematics
  [Results in Mathematics and Related Areas. 3rd Series. A Series of Modern
  Surveys in Mathematics]}.
\newblock Springer-Verlag, Berlin, 2003.

\bibitem[KP10]{kuksinperelman}
Sergei Kuksin and Galina Perelman.
\newblock Vey theorem in infinite dimensions and its application to {K}d{V}.
\newblock {\em Discrete Contin. Dyn. Syst.}, 27(1):1--24, 2010.

\bibitem[KST13]{kappelerschaad}
T.~Kappeler, B.~Schaad, and P.~Topalov.
\newblock Qualitative {F}eatures of {P}eriodic {S}olutions of {K}d{V}.
\newblock {\em Comm. Partial Differential Equations}, 38(9):1626--1673, 2013.

\bibitem[MBC14]{maiocchi_bambusi_carati}
A.~Maiocchi, D.~Bambusi, and A.~Carati.
\newblock An {A}veraging {T}heorem for {FPU} in the {T}hermodynamic {L}imit.
\newblock {\em J. Stat. Phys.}, 155(2):300--322, 2014.

\bibitem[Muj86]{mujica}
Jorge Mujica.
\newblock {\em Complex analysis in {B}anach spaces}, volume 120 of {\em
  North-Holland Mathematics Studies}.
\newblock North-Holland Publishing Co., Amsterdam, 1986.
\newblock Holomorphic functions and domains of holomorphy in finite and
  infinite dimensions, Notas de Matem{\'a}tica [Mathematical Notes], 107.

\bibitem[Nik86]{nikolenko}
N.~V. Nikolenko.
\newblock The method of {P}oincar\'e normal forms in problems of integrability
  of equations of evolution type.
\newblock {\em Uspekhi Mat. Nauk}, 41(5(251)):109--152, 263, 1986.

\bibitem[PCSF11]{ponno}
A.~Ponno, H.~Christodoulidi, Ch. Skokos, and S.~Flach.
\newblock The two-stage dynamics in the fermi-pasta-ulam problem: From regular
  to diffusive behavior.
\newblock {\em Chaos: An Interdisciplinary Journal of Nonlinear Science},
  21(4):--, 2011.

\bibitem[SW00]{schneider_wayne}
Guido Schneider and C.~Eugene Wayne.
\newblock Counter-propagating waves on fluid surfaces and the continuum limit
  of the {F}ermi-{P}asta-{U}lam model.
\newblock In {\em International {C}onference on {D}ifferential {E}quations,
  {V}ol. 1, 2 ({B}erlin, 1999)}, pages 390--404. World Sci. Publ., River Edge,
  NJ, 2000.

\bibitem[{Tod}67]{toda}
M.~{Toda}.
\newblock {Vibration of a Chain with Nonlinear Interaction}.
\newblock {\em Journal of the Physical Society of Japan}, 22:431, February
  1967.

\bibitem[Tr{\`e}70]{treves}
Fran{\c{c}}ois Tr{\`e}ves.
\newblock An abstract nonlinear {C}auchy-{K}ovalevska theorem.
\newblock {\em Trans. Amer. Math. Soc.}, 150:77--92, 1970.

\bibitem[Vey78]{vey}
J.~Vey.
\newblock Sur certains syst\`emes dynamiques s\'eparables.
\newblock {\em Amer. J. Math.}, 100(3):591--614, 1978.

\bibitem[vM76]{moerbeke}
Pierre van Moerbeke.
\newblock The spectrum of {J}acobi matrices.
\newblock {\em Invent. Math.}, 37(1):45--81, 1976.

\end{thebibliography}

\end{document}